%% file: Main.tex
\documentclass[11pt]{article}

\usepackage[top=59pt,bottom=59pt,left=62pt,right=69pt]{geometry}
\usepackage{amsmath}
\usepackage{amssymb}
\usepackage{amsthm}
\usepackage[all]{xy}
\usepackage{relsize}
\usepackage{stmaryrd}
\usepackage{lscape}
\usepackage{float}
\usepackage{arydshln}
\usepackage{graphicx}
\usepackage{tikz}
\usepackage{xcolor}
\usepackage{booktabs}
\usepackage{thmtools} 
\usepackage{thm-restate}
\usepackage[pdfencoding=auto,hidelinks]{hyperref}
\usepackage{bookmark}
\usepackage[utf8]{inputenc}

\usepackage[all]{xy}
\xymatrixrowsep{5mm}
\xymatrixcolsep{5mm}

\theoremstyle{plain}
\newtheorem{theorem}{Theorem}[section]
\newtheorem{lemma}[theorem]{Lemma}

\newtheorem{corollary}[theorem]{Corollary}

\newtheorem{proposition}[theorem]{Proposition}

\newtheorem{theoremm}{Theorem}

\theoremstyle{definition}
\newtheorem{definition}[theorem]{Definition}
\newtheorem{example}[theorem]{Example}

\newtheorem{myexample}{Example}
\newtheorem{myexampleA}{Example}
\newtheorem{myexampleB}{Example}
\newtheorem{myexampleC}{Example}
\newtheorem{myexampleD}{Example}
\newtheorem{myexampleE}{Example}

\newcommand{~}{\;}
\newcommand{\NatNum}{\mathbb{N}}
\newcommand{\NatType}{\mathbf{N}}
\newcommand{\UnitType}{\mathbf{1}}
\newcommand{\TypeOne}{\tau}
\newcommand{\TypeTwo}{\rho}

\newcommand{\tis}[2]{#1[{\in #2}]}
\newcommand{\Rise}[2]{{{#1}^{\uparrow} #2}}

\newcommand{\prebas}{\preccurlyeq}

\newcommand{\modbel}{\curlyeqprec}
\newcommand{\LogNeg}{\mathcal{V}}
\newcommand{\LogPos}{\mathcal{V}^+}
\newcommand{\PureLogNeg}{\mathcal{F}}
\newcommand{\PureLogPos}{\mathcal{F}^+}

\newcommand{\denote}[1]{\llbracket #1 \rrbracket}
\newcommand{\Denote}[1]{\left\llbracket #1 \right\rrbracket}
\newcommand{\PrePos}{\sqsubseteq_{\LogPos}}
\newcommand{\PreGen}{\sqsubseteq_{\LogNeg}}

\newcommand{\Obser}{\mathcal{O}}

\newcommand{\Dtrees}{TT\UnitType}
\newcommand{\Treef}{\textit{CF}_{\LogPos}\,(\UnitType)}

\newcommand{\simil}{\sqsubseteq_{s}}
\newcommand{\bisim}{\sqsubseteq_{b}}

\newcommand{\Chi}{{\chi}}

\begin{document}

\title{Behavioural Equivalence via Modalities for Algebraic Effects
}

\author{
	Alex Simpson\footnote{Email:\! Alex.Simpson@fmf.uni-lj.si. Supported by the Slovenian Research Agency, research core funding No.\ P1--0294.}
	\quad \quad \quad Niels Voorneveld\footnote{Email:\! Niels.Voorneveld@fmf.uni-lj.si. This material is based upon work supported by the Air Force Office of Scientific Research under award number FA9550-17-1-0326. This project has received funding from the European Union’s Horizon 2020 research and innovation programme under the Marie Skłodowska-Curie grant agreement No 731143.}\\~\\
	University of Ljubljana\\
	Ljubljana, Slovenia.
}
\date{}

\maketitle
\begin{abstract}
	The paper investigates behavioural equivalence between programs in a call-by-value functional language extended with a signature of (algebraic) effect-triggering operations. Two programs are considered as being behaviourally equivalent if they enjoy the same behavioural properties. To formulate this, we define a logic whose formulas specify behavioural properties. A crucial ingredient is a collection of \emph{modalities} expressing effect-specific aspects of behaviour. We give a general theory of such modalities. If two conditions, \emph{openness} and \emph{decomposability}, are satisfied by the modalities then the logically specified behavioural equivalence coincides with a modality-defined notion of applicative bisimilarity, which can be proven to be a congruence by a generalisation of Howe's method. We show that the openness and decomposability conditions hold for several examples of algebraic effects: nondeterminism, probabilistic choice, global store and input/output. 
\end{abstract}

\input{1_Introduction}
\input{2_Language}
\input{3_Logic}
\input{4_Equivalence}

\input{5_Similarity}
\input{6_Howe}
\input{7_Contextual}
\input{8_Pure}

\input{9_Reasoning}

\input{10_Conclusion}

\bibliographystyle{plain}
{\footnotesize
\bibliography{biblio}
}

\input{A_Proofs}

\end{document}

%% file: 1_Introduction.tex
\section{Introduction}\label{section:introduction}

Many tasks in software development and analysis rely on abstracting away from 
program syntax to an appropriate notion of program \emph{behaviour}. For example, the goal of  \emph{specification} is to specify (constraints on) the behaviour of a program. Similarly, \emph{verification} concerns validating that a program indeed exhibits  the behaviour specified.
Closely associated with the general concept of \emph{behaviour} is the related 
concept of \emph{behavioural equivalence}, under which two programs are deemed equivalent if they exhibit the same behaviour.  

Studying the interrelated concepts of \emph{behaviour} and  \emph{equivalence}  is important from both theoretical and practical perspectives. On the theoretical side, these are fundamental notions, whose understanding sheds light not just on  any particular programming language under consideration, but more generally on the question of how to mathematically model the process of computation. On the practical side, the engineering tasks of program specification, verification and synthesis all depend on having a precise mathematical model of program behaviour; and notions of program equivalence play a key role in applications, such as compiler optimisation, that involve program transformation. 

For  applications such as those described above to be possible, it is crucial that the mathematical notion of behaviour is appropriately chosen. For example, in the case of compiler transformations, it is essential that aspects of computational behaviour, such as execution time, that one specifically does not want preserved, are ignored. Whereas, for other applications, for example ones in which resources need to be quantified, it may be important to have a notion of behaviour in which execution time is taken into account. In general, therefore, there is no single all-encompassing approach to  defining behaviour and equivalence. 
Nonetheless, as general desirability criteria,  one would like to have definitions that are, 
on the one hand, mathematically natural and convenient to work with and, on the other,  suitable for practical applications. 

In the present paper, we explore and relate two complementary methodologies for  \emph{defining} notions of behaviour and equivalence, within a particular programming context: \emph{call-by-value functional programming with effects}. The methodologies themselves make sense, however, within a more-or-less arbitrary programming context, so we introduce and motivate them at this greater generality.

\paragraph{Behavioural logic} The first methodology is to specify \emph{program behaviour} via a formal logic, which we call a \emph{behavioural logic}, 
whose formulas $\phi$ are constructed using operators that express primitive properties of program behaviour. 
Mathematically, one defines a \emph{satisfaction relation} $M \models \phi$, expressing that program $M$ satisfies behavioural property $\phi$. The idea is that the logic should be designed in such a way that its formulas $\phi$ 
are capable of expressing 
all properties of programs that are \emph{bona fide} properties of program behaviour (as opposed to, for example,  properties of program syntax).

Given such a program logic, one derives a corresponding notion of behavioural equivalence.
Two programs $M,N$ are said be \emph{logically equivalent}
	if, for all formulas $\phi$, it holds that $M \models \phi$ iff $N \models \phi$; that is if they exhibit the same behaviour.
	
\paragraph{Bisimilarity} It is a remarkably general phenomenon that, in numerous computation contexts, program 
behaviour can be  modelled as an interactive process, leading to a natural coinductive definition of program equivalence as 
\emph{bisimilarity}: roughly speaking,  the largest (equivalence) relation that relates
interaction points only if their local behaviour is indistinguishable modulo the relation.

Having given such a mathematical definition of bisimilarity, one can derive an associated  notion of {behavioural property}. Namely, a property of programs is behavioural if it respects bisimilarity; that is, whenever a program $M$ satisfies the property, then so does any bisimilar program $M'$.

\medskip

The above complementary approaches to defining behaviour and equivalence have been particularly prominent in concurrency theory. Indeed, the idea of bisimilarity as a notion of behavioural equivalence between systems was first introduced in that context, in the work of Milner and Park~\cite{Milner82,Park81}. The logical approach to defining behaviour
emerged around the same time, with the characterisation, by Hennessy and Milner, of bisimilarity as the behavioural equivalence induced by  an infinitary propositional modal logic, now known as 
\emph{Hennessy-Milner logic}~\cite{Henessy85}.

In the case of bisimilarity, Abramsky realised that a similar style of definition generalises to other programming contexts. In particular, he developed the notion of \emph{applicative bisimilarity} for functional languages \cite{Abramsky90}. Subsequently, numerous variant notions of bisimilarity have been given across a plethora of computational contexts (for example, \cite{Sangiorgi_book, Sangiorgi:2011, Bisim_object, Lassen99}). 
To highlight one recent example, which is  important for the present paper, Dal Lago, Gavazzo and Levy have provided a uniform generalisation of {applicative bisimilarity} to a functional programming language with effects~\cite{Relational}.

A major goal of the present paper is to show that the 
logical approach to defining program behaviour  can also be adapted very naturally to the context of  functional programming languages with effects. In doing so, we establish that the corresponding behavioural equivalence coincides with effectful applicative bisimilarity in the style of Dal Lago \emph{et al.}~\cite{Relational}.

More precisely, we consider a typed call-by-value functional programming language with \emph{algebraic effects} in the sense of Plotkin and Power \cite{effect}. Broadly speaking, effects are those aspects of computation that involve a program interacting with its `environment'; for example: nondeterminism, probabilistic choice (in both cases, the choice is deferred to the environment); input/output; mutable store (the machine state is modified); control operations such as exceptions, jumps and handlers  (which interact with the continuation in the evaluation process); etc. Such general effects collectively enjoy common properties identified in the work of Moggi on monads \cite{monad}. Among them, algebraic effects play a special role. They can be  included in a programming language by adding  effect-triggering operations, whose `algebraic' nature means that effects act independently of the continuation. From the aforementioned examples of effects, only jumps and handlers are non-algebraic. Thus the notion of algebraic effect covers a broad range of effectful computational behaviour.  Call-by-value functional languages provide a natural context for exploring effectful programming. From a theoretical viewpoint, other programming paradigms are subsumed; for example, 
imperative programs can be recast as effectful functional ones. From a practical viewpoint, the combination of effects with call-by-value
leads to the natural programming style supported by impure functional languages such as OCaml.

In order to focus on the main contributions of the paper (the behavioural logic and its  induced behavioural equivalence), we 
instantiate ``call-by-value functional language with algebraic effects'' using a very simple language. 
Our language is 
a simply-typed $\lambda$-calculus with a base type 
of natural numbers, general recursion, call-by-value function evaluation, and algebraic effects.
That is, it is a call-by-value version of PCF \cite{PCF}, extended with effects. A very similar language is used by Plotkin and Power~\cite{effect};
although, for technical convenience, we adopt an alternative (but equivalent) formulation using
\emph{fine-grained} call-by-value \cite{CBV}. 
The language is defined precisely in Section~\ref{section:language}, using an operational semantics that evaluates programs to \emph{effect trees}~\cite{effect,op_meta}.

Section~\ref{section:logic} introduces the behavioural logic. In our impure functional setting, the evaluation of a program of type $\tau$
results in a computational process that may or may not invoke effects, and which may or may not terminate with a return \emph{value} of type $\tau$. 
The key ingredient in our logic is an effect-specific family $\mathcal{O}$ of \emph{modalities}, where each modality $o \in \mathcal{O}$ converts a property $\phi$ of values of some given type $\tau$ to a property $o\,\phi$ of general programs (called \emph{computations}) of type $\tau$.
The idea is that such modalities capture all relevant effect-specific behavioural properties of the effects under consideration.
For example, in the context of probabilistic computation, we have a modality $\mathsf{P}_{\!>q}$ for every rational $q \in [0,1)$, where the formula
$\mathsf{P}_{\!>q}\,\phi$ is satisfied by a computation 
if the evaluation of the computation has a probability greater than $q$ of terminating with a return value 
satisfying $\phi$. 

A main contribution of the paper is to give a general framework for defining such effect modalities, applicable across a wide range of algebraic effects.
The technical setting for this is that we have a signature $\Sigma$ of effect operations, which determines the programming language, and a  collection $\mathcal{O}$ of modalities, which determines the behavioural logic. In order to specify the semantics of the logic, we require
each modality to be assigned a set of unit-type effect trees, which defines the meaning of the modality. 
For example, 
the modality $\mathsf{P}_{\!>q}$ is specified by the set of effect trees that have probability greater than $q$ of terminating if considered as a Markov chain. 
Several further examples and a detailed general explanation are given in Section~\ref{section:logic}.

In Section~\ref{section:preorder}, we consider the relation of \emph{behavioural equivalence} between programs
determined by the logic. 
This equivalence directly relates to the notion of behaviour that is made explicit by the modalities. 
A fundamental well-behavedness property that any reasonable program equivalence should enjoy is that it
should be a \emph{congruence} with respect to the syntactic constructs of the programming language.
(If an equivalence is not a congruence then there are aspects of program behaviour that are not distinguished by the equivalence but which can be separated within the programming language itself.) 
The major result of the paper  about the logically induced behavioural equivalence is that it is indeed a congruence, as long as 
two conditions, \emph{openness} and \emph{decomposability},  hold of
the collection $\mathcal{O}$ of modalities (Theorem~\ref{Main_Theorem}). These two conditions do indeed hold for the natural sets  $\mathcal{O}$ of modalities associated with the principal examples of algebraic effects.

As a second indication of the reasonableness  of the logically defined behavioural equivalence, we  establish that it has an alternative characterisation  as an effect-sensitive version of Abramsky's notion of 
\emph{applicative bisimilarity}~\cite{Abramsky90}. This is achieved, in Section~\ref{section:similarity}, by using the modalities $\mathcal{O}$ in a logic-free context to define a relation of \emph{applicative $\mathcal{O}$-bisimilarity}. The resulting equivalence relation is closely related to the effect-sensitive version
of bisimilarity 
developed by Dal Lago \emph{et al.}~\cite{Relational} for an untyped  language with general algebraic effects. Our use of modalities as a uniform method for  generating the bisimilarity relation is, however, novel. Theorem~\ref{log_is_sim} shows that applicative $\mathcal{O}$-bisimilarity coincides with the logically defined relation of behavioural equivalence.

In addition to its conceptual value, establishing the coincidence of the logically defined equivalence with applicative $\mathcal{O}$-bisimilarity also serves an important technical purpose: this result is used crucially  in the proof of Theorem~\ref{Main_Theorem}.
We adapt the well-known proof method of Howe~\cite{How89,How} to show that applicative $\mathcal{O}$-bisimilarity is a congruence. By the coincidence theorem, this means that the logically defined equivalence is indeed a congruence; establishing Theorem~\ref{Main_Theorem}.
Our adaptation of Howe's method is presented in Section~\ref{section:Howe}. 
Although the argument is technically involved, a similar adaptation of Howe's method to  a language with algebraic effects has been previously given by Dal Lago \emph{et al.}~\cite{Relational}. 
Accordingly, this proof is not a main contribution of the present paper, 
and we give only an outline argument in the main body of the paper, with details deferred to Appendix~\ref{section:proofs}.

In the discussion thus far we have ignored one nuance within the development of Sections \ref{section:logic}--\ref{section:Howe}.
In addition to working with the full behavioural logic, we also identify a \emph{positive fragment} of the logic, which omits negation. Whereas the full logic defines an equivalence relation that coincides with applicative bisimilarity, the positive logic defines a \emph{preorder} between programs that coincides with the coarser relation of \emph{applicative similarity}, from which it follows  that the equivalence relation determined by the positive logic coincides with
\emph{mutual applicative similarity}. In Section~\ref{section:contextual} we compare the above equivalence relations 
and preorders, and we relate them to  \emph{contextual preorder} and \emph{equivalence}, which can also be naturally defined in terms of  the set $\mathcal{O}$ of modalities. In general, applicative bisimilarity is contained in mutual applicative similarity which is contained in contextual equivalence. 
For some effects, however, these inclusions are strict. For example, as has been shown by
Lassen~\cite{Lassen} and Ong~\cite{Ong}, there are examples involving nondeterminism that separate the equivalences.
For the sake of completeness, we recall such separating examples in Section~\ref{section:contextual}, with the novelty that we use the logical characterisations of bisimilarity and mutual similarity  as a tool for proving  and disproving the 
equivalences in question.

In cases in which the equivalences differ, there is a question of which (if any) of the equivalences should be taken as being the most fundamental. In the literature, {contextual equivalence} is often taken as the equivalence of choice
for  applicative languages. From this viewpoint, bisimilarity and mutual similarity are considered important for  providing sound  (but incomplete) proof methods for reasoning about contextual equivalence, which is the relation of ultimate interest. We would like to argue instead that the logical viewpoint makes a strong case for considering bisimilarity as being the primary equivalence. 
Every formula $\phi$  in the  logic in this paper states a meaningful property of program behaviour. 
It is a truism that, when two programs are not bisimilar, there is some  property of behaviour, expressible in the logic, that distinguishes them. If one takes the logic seriously, then the desirable property that 
$M \models \phi$ and $M \equiv N$ imply $N \models \phi$ holds only in the case that the equivalence relation is bisimilarity, because of its characterisation as the logically induced equivalence.

The above (almost trivial) argument of course relies on accepting arbitrary formulas $\phi$ as expressing behaviourally meaningful properties. Readers can read Section~\ref{section:logic} and make up their own minds about this. Alternatively, one might take a  more pragmatic viewpoint. Another reason for accepting a logic for expressing properties of programs is if it provides the expressivity needed to formulate proof principles for reasoning about programs. In Section~\ref{section:reasoning}, we show that our logic does indeed support such principles. Furthermore, they arise in the form of  \emph{compositional} proof rules that allow properties of a compound program to be proved by establishing appropriate properties of its constituent subprograms. In order to achieve this, we provide, in Section~\ref{section:pure}, a reformulation of our behavioural logic, which is of interest in its own right. This reformulation enjoys the property that the syntax of the logic is independent of the syntax of the programming language, which is not true of the logic of Section~\ref{section:logic}. This property is 
appealing, as one would like to be able to specify behaviour without knowing the syntax in which programs are written. More practically, the reformulation is used to formulate, in Section~\ref{section:reasoning}, the compositional proof principles referred to above. There is one significant qualification to add to this perspective.
The infinitary propositional logics considered in the present paper are by no means suitable for serving as practical logics for
specification and verification. Nevertheless, we view these logics as relevant to the development of more practical
non-propositional but finitary logics, as they can potentially act as low-level `target' logics into which high-level practical logics can be `compiled'. We elaborate further on this point in  Section~\ref{section:reasoning}. However, the development of practical logics is left as a topic for future work.

Finally, in Section~\ref{section:conclusion} we discuss other related and further work.

This paper is an extended and revised
version of a conference paper \cite{ESOP}, presented at the \emph{27th European Symposium on Programming (ESOP)}. One  principal difference from the conference version is that full proofs are included. Other extensions include: normal form results for the logic (Propositions~\ref{proposition:logA} and~\ref{proposition:logB});
a significantly  expanded presentation of the crucial \emph{decomposability} properties in Section~\ref{section:preorder};
explicit descriptions of the $\mathcal{O}$-relators determined by our running examples in Section~\ref{section:similarity}; 
the comparison between bisimilarity, similarity and contextual equivalence in Section~\ref{section:contextual};
and the discussion of compositional reasoning principles in Section~\ref{section:reasoning}.

%% file: 2_Language.tex
\section{A simple programming language}\label{section:language}

As motivated in the introduction, our chosen base language is a simply-typed call-by-value functional language with general recursion and a ground type of natural numbers, to which we add (algebraic) effect-triggering operations. This means that our language is a call-by-value variant of PCF~\cite{PCF}, extended with algebraic effects, resulting in a language similar to the one considered in Plotkin and Power~\cite{effect}. In order to simplify the technical treatment of the language, we present it in the style of \emph{fine-grained call-by-value}~\cite{CBV}. This means that we make a syntactic distinction between \emph{values} and \emph{computations}, separating
the static and dynamic aspects of the language respectively. Furthermore, all \emph{sequencing} of computations is performed using
a single language construct, the $\textbf{let}$ construct. The resulting  language is straightforwardly intertranslatable with the more traditional call-by-value formulation. 
But the encapsulation of all sequencing within a single construct has the benefit of avoiding redundancy in proofs. 

Our types are just the simple types obtained by iterating the function type construction over two base types:
$\NatType$  of natural numbers, and also a unit type $\bf{1}$.

\vspace{1mm}
\noindent
\textbf{Types}: $\TypeOne,\TypeTwo ::= ~ \bf{1} \mid \NatType \mid \TypeTwo \rightarrow \TypeOne$

\vspace{1mm}
\noindent
\textbf{Contexts}: $\Gamma ::= ~ \emptyset \mid \Gamma, \,x : \TypeOne$

\vspace{1mm}
\noindent
As usual, term variables $x$ are taken from a countably-infinite stock of such variables, and the context $\Gamma, \, x : \TypeOne$ can only be formed if the variable $x$ does not already appear in $\Gamma$.

As discussed above, program terms are separated into two mutually defined but disjoint categories: \emph{values} and \emph{computations}.

\vspace{1mm}
\noindent
\textbf{Values}: $V,W ::=  ~ * \mid Z \mid S(V) \mid \lambda x : \TypeOne.M \mid x$

\vspace{1mm}
\noindent
\textbf{Computations}: $M,N ::= VW \mid \textbf{return } V \mid \textbf{let } M \Rightarrow x \textbf{ in } N \mid  \textbf{fix}(V)$ 

$\hspace{30mm} \mid \textbf{case } V \textbf{ in } \{Z \Rightarrow M, S(x) \Rightarrow N\}$

\vspace{1mm}
\noindent
Here, $*$ is the unique value of the unit type. The values of the type of natural numbers are the \emph{numerals} represented using zero $Z$ and successor $S$. The values of function types are the $\lambda$-abstractions. And a variable $x$ can be considered a value, because, under the call-by-value evaluation strategy of the language, it can only be instantiated with a value. 

The computations are: function application  $VW$; the computation that does nothing but  return a value $V$;
a $\textbf{let}$ construct for sequencing; a $\textbf{fix}$ construct for recursive definition\footnote{By defining $\textbf{fix}\, V$ to be a computation which reduces to a lambda term, it holds that the only values of function types are lambda terms.}; and a $\textbf{case}$ construct that branches according to whether its natural-number argument is zero or positive. The  
computation  `$\textbf{let } M \Rightarrow x \textbf{ in } N$' 
implements sequencing in the following sense. First the computation $M$ is evaluated. Only in  the case that the evaluation of $M$ terminates, with return value $V$, does the thread of execution continue to $N$. In this case, the computation $N[V/x]$ is evaluated, and its return value (if any) is then returned as the result of the $\textbf{let }$ computation.

To the {pure} functional language described above, we add \emph{effect operations}. The collection of effect operations is specified by a set $\Sigma$ (the \emph{signature}) of such operations, together with, for each $\sigma \in \Sigma$ an associated \emph{arity} which takes one of the four forms below\footnote{These four forms of arity suffice for the examples we consider. Similar choices were made in \cite{op_meta, Plotkin:2009}. Going beyond such \emph{discrete} arities is an interesting direction for future research; see Section~\ref{section:conclusion}.}
\[\alpha^n \rightarrow \alpha ~~~~~ \NatType \times \alpha^n \rightarrow \alpha
  ~~~~~ \alpha^{\NatType} \rightarrow \alpha ~~~~~ \NatType \times \alpha^{\NatType} \rightarrow \alpha 
  \enspace .\]
The notation here is chosen to be suggestive of the way in which such arities are used in the typing rules of Fig.~\ref{fig:typ}, viewing $\alpha$ as a type variable. Each of the forms of arity has an associated term constructor, for building additional computation terms, with which we extend the above grammar for computation terms. 

\vspace{1mm}
\noindent
\textbf{Computations}: ~$M,N ::=  \dots \mid \sigma(M_0,M_1,\dots,M_{n-1}) \mid \sigma(V;M_0,M_1,\dots,M_{n-1}) \mid \sigma(V) \mid \sigma(W;V)$
\vspace{1mm}

\noindent 
Motivating examples of effect operations and their computation terms can be found in Examples~\ref{example:pure}--\ref{example:io} below.

The typing rules for the language are given in Figure \ref{fig:typ}. Note that the choice of typing rule for an effect operation $\sigma \in \Sigma$ depends on its declared arity.

\begin{figure}
\[
\frac{}{\Gamma, x:\TypeOne \vdash x:\TypeOne} \quad \quad
\frac{}{\Gamma \vdash *:\bf{1}}  \quad \quad
\frac{}{\Gamma \vdash Z : \NatType} \quad \quad
\frac{\Gamma \vdash V:\NatType}{\Gamma \vdash S(V):\NatType}
\]
\[
\frac{\Gamma \vdash V : \TypeOne}{\Gamma \vdash \textbf{return}(V) : \TypeOne} \quad \quad 
\frac{\Gamma, x:\TypeOne \vdash M:\TypeTwo}{\Gamma \vdash (\lambda x:\TypeOne.M):\TypeOne \rightarrow \TypeTwo}
\]
\[
\frac{\Gamma \vdash V:\TypeOne \rightarrow \TypeTwo \quad \quad \Gamma \vdash W:\TypeOne}{\Gamma \vdash V W:\TypeTwo} \quad \quad
\frac{\Gamma \vdash V: (\TypeOne \rightarrow \TypeTwo) \rightarrow (\TypeOne \rightarrow \TypeTwo)}{\Gamma \vdash \text{\textbf{fix}}(V): \TypeOne \rightarrow \TypeTwo}
\]
\[
\frac{\Gamma \vdash V:\NatType \quad \quad \Gamma \vdash M:\TypeOne \quad \quad \Gamma, x:\NatType \vdash N:\TypeOne}{\Gamma \vdash \text{ \textbf{case} } V \text{ \textbf{of} } \{Z \Rightarrow M; S(x) \Rightarrow N\} : \TypeOne} \quad
\frac{\Gamma \vdash M:\TypeOne  \quad \quad \Gamma, x:\TypeOne \vdash N:\TypeTwo}{\Gamma \vdash \textbf{let } M \Rightarrow x \textbf{ in } N : \TypeTwo}	
\]
\[
\frac{\sigma : \alpha^n \rightarrow \alpha \quad \quad \Gamma \vdash M_i:\TypeOne}{\Gamma \vdash \sigma(M_0,M_1,\dots,M_{n-1}):\TypeOne}  \quad \quad 
\frac{\sigma : \alpha^{\NatType} \rightarrow \alpha  \quad \quad \Gamma \vdash V: \NatType \rightarrow \TypeOne}{\Gamma \vdash \sigma(V):\TypeOne}
\]
\[
\frac{\sigma : \NatType \times \alpha^n \rightarrow \alpha \quad \quad \Gamma \vdash V : \NatType \quad \quad \Gamma \vdash M_i:\TypeOne}{\Gamma \vdash \sigma(V;M_0,M_1,\dots,M_{n-1}):\TypeOne}
\]
\[
\frac{\sigma : \NatType \times \alpha^{\NatType} \rightarrow \alpha \quad \quad \Gamma \vdash V : \NatType \quad \quad \Gamma \vdash W:\NatType \rightarrow \TypeOne}{\Gamma \vdash \sigma(V;W):\TypeOne}
\]
\caption{Typing rules}
\label{fig:typ}
\end{figure}

The terms of type $\TypeOne$ are the values and computations generated by the constructors above. Every term has a unique \emph{aspect} as either a  value or computation. We write $\textit{Val}(\TypeOne)$ and $\textit{Com}(\TypeOne)$ respectively for closed values and computations. So the closed terms of $\TypeOne$ are $\textit{Term}(\TypeOne) := \textit{Val}(\TypeOne) \cup \textit{Com}(\TypeOne)$. For $n \in \NatNum$ a natural number, we write $\overline{n}$ for the numeral $S^n(Z)$. Thus $\textit{Val}(\NatType) = \{\overline{n} \,|\, n \in \NatNum\}$. 

We now consider some illustrative signatures of computationally interesting effect operations, which will be used as running examples throughout the paper. (We use the same examples as in Johann \emph{et al.}~\cite{op_meta}.)

\setcounter{myexample}{-1}
\begin{myexample}[Pure functional computation]
\label{example:pure}
This is the trivial case (from an effect point of view) in which the signature $\Sigma$ of effect operations is empty. The resulting language is a fine-grained call-by-value variant
of PCF~\cite{PCF}.
\end{myexample}

\begin{myexample}[Error]
\label{example:error}
	We take a set of error labels $E$. For each $e \in E$ there is an effect operation $\textit{raise}_e: \alpha^0 \rightarrow \alpha$ which,
	when invoked by the computation   $\textit{raise}_e()$, aborts evaluation and outputs `$e$' as an error message. 
\end{myexample}
\begin{myexample}[Nondeterminism]
	\label{example:non}
	There is a binary choice operation $\textit{or}: \alpha^2 \rightarrow \alpha$ which gives two options for continuing the computation.
	The choice of continuation is under the control of some external agent, which one may wish to model as being cooperative (\emph{angelic}), antagonistic (\emph{demonic}), or \emph{neutral}.
\end{myexample}
\begin{myexample}[Probabilistic choice]
	\label{example:prob}
	Again there is a single binary choice operation $\textit{p-or}: \alpha^2 \rightarrow \alpha$ which gives two options for continuing the computation.
	In this case, the choice of continuation is probabilistic, with a $\frac{1}{2}$ probability of either option being chosen. Other weighted probabilistic choices can be programmed in terms of this fair choice operation.
\end{myexample}
\begin{myexample}[Global store]
	\label{example:gs}
	We take a set of locations $L$ for storing natural numbers. For each $l \in L$ we have $\textit{lookup}_l: \alpha^{\NatType} \rightarrow \alpha$ and $\textit{update}_l: \NatType \times \alpha \rightarrow \alpha$. The computation $\textit{lookup}_l(V)$ looks up the number at location $l$ and passes it as an argument to the function $V$, and $\textit{update}_l(\overline{n};M)$ stores $n$ at $l$ and then continues with the computation $M$.
\end{myexample}
\begin{myexample}[Input/output]
\label{example:io}
	Here we have two operations, $\textit{read}: \alpha^{\NatType} \to \alpha $ which reads a number from an input channel and passes it as the argument to a function, and $\textit{write}: \NatType \times \alpha \rightarrow \alpha$ which outputs a number (the first argument) and then continues with the computation given as the second argument.
\end{myexample}

We next present an operational semantics for our language, under  which a computation term evaluates to an \emph{effect tree}: essentially, a coinductively generated term using operations from $\Sigma$, and with values
and $\bot$ (nontermination) as the generators. This idea appears in Plotkin and Power~\cite{effect}, and our technical treatment follows the approach of Johann \emph{et al.}~\cite{op_meta}, adapted to (fine-grained) call-by-value.

We define a single-step reduction relation $\rightarrowtail$ between configurations $(S,M)$ consisting of a stack $S$ and a computation $M$. The computation $M$ is the term under current evaluation. The stack $S$ represents a continuation computation awaiting the termination of $M$. First, we define a stack-independent reduction relation on computation terms that do not involve $\textbf{let}$ at the top level.
\begin{align*}
& (\lambda x : \TypeOne.M) \, V ~ \rightsquigarrow ~  M[V/x]
\\
& \text{ \textbf{case} } Z \text{ \textbf{of} } \{Z \Rightarrow M_1; S(x) \Rightarrow M_2\} ~ \rightsquigarrow ~ M_1
\\
& \text{ \textbf{case} } S(V) \text{ \textbf{of} } \{Z \Rightarrow M_1; S(x) \Rightarrow M_2\} \quad \rightsquigarrow \quad M_2[V/x] 
\\
& \text{ \textbf{fix}}(F) \quad \rightsquigarrow \quad \textbf{return } \lambda x : \TypeOne. \textbf{let } F\,(\lambda y:\TypeOne.\textbf{let fix } F \Rightarrow z \textbf{ in } zy) \Rightarrow w \textbf{ in } wx
\end{align*}

\noindent
The behaviour  of $ \textbf{let}$ is implemented using a system of stacks where:

\vspace{1mm}
\noindent
\textbf{Stacks}: $S ::= ~ \textit{id} \mid S \circ (\textbf{let } (-) \Rightarrow x \textbf{ in } M)$
\vspace{1mm}

\noindent
We write $S\{N\}$ for the computation term obtained by `applying' the stack $S$ to $N$, defined by:
\begin{align*}
 & \textit{id}\, \{N\}  ~ =  ~ N \\
 & (S \circ (\textbf{let } (-) \Rightarrow x \textbf{ in } M))\,  \{N\} ~ =  ~ S\{\textbf{let } N \Rightarrow x \textbf{ in } M\}
\end{align*}
We write $\textit{Stack}(\TypeOne,\TypeTwo)$ for the set of stacks $S$ such that for any $N \in \textit{Com}(\TypeOne)$, it holds that $S\{N\}$ is a well-typed expression of type $\TypeTwo$. We define a reduction relation on pairs $\textit{Stack}(\TypeOne,\TypeTwo) \times \textit{Com}(\TypeOne)$ (denoted $(S_1,M_1) \rightarrowtail (S_2,M_2)$) by:
\begin{align*}
& (S,\textbf{let } N \Rightarrow x \textbf{ in } M) ~ \rightarrowtail ~ (S \circ (\textbf{let } (-) \Rightarrow x \textbf{ in } M),N)
\\
& (S,R) ~  \rightarrowtail ~  (S,R') \quad && \text{if $R \rightsquigarrow R'$}
\\
& (S \circ (\textbf{let } (-) \Rightarrow x \textbf{ in } M),\textbf{return } V) ~ \rightarrowtail ~(S,M[V/x])
\end{align*}

We define the notion of  \emph{effect tree} for an arbitrary set $X$, where $X$ is thought of as a set of abstract `values'.
\begin{definition}\label{definition:tree}
	An \emph{effect tree} (henceforth \emph{tree}), over a set $X$, determined by a signature $\Sigma$ of effect operations,  is a labelled and possibly infinite tree whose nodes have the possible forms:
	\begin{enumerate}
		\item A leaf node labelled with $\bot$ (the symbol for nontermination).
		\item A leaf node labelled with $x$ where $x \in X$.
		\item A node labelled $\sigma$ with children $t_0,\dots, t_{n-1}$, when $\sigma\in \Sigma$ has arity $\alpha^n \rightarrow \alpha$.
		\item A node labelled $\sigma$ with children $(t_i)_{i \in \NatNum}$, when $\sigma\in \Sigma$ has arity $\alpha^{\NatType} \rightarrow \alpha$.		
		\item A node labelled $\sigma_m$ where $m \in \NatNum$ with children $t_0,\dots, t_{n-1}$, when $\sigma\in \Sigma$ has arity $\NatType \times \alpha^n \rightarrow \alpha$.
		\item A node labelled $\sigma_m$ where $m \in \NatNum$ with children $(t_i)_{i \in \NatNum}$, when $\sigma\in \Sigma$ has arity $\NatType \times \alpha^{\NatType} \rightarrow \alpha$.
	\end{enumerate}
	\end{definition}

\noindent See Examples \ref{example:nonA} and \ref{example:gsA} later on in this section for examples of effect trees.

We write $TX$ for the set of trees over $X$.
We define a partial ordering on $TX$ where $t_1 \leq t_2$, if $t_1$ can be obtained by pruning $t_2$ by removing a possibly infinite number of subtrees of $t_2$ and putting  leaf nodes labelled $\bot$ in their place. This forms an \emph{$\omega$-complete} partial order, meaning that every ascending sequence $t_1 \leq t_2 \leq \dots$ has a least upper bound $\bigsqcup_n t_n$. Let $\textit{Tree}(\TypeOne) := T(\textit{Val}(\TypeOne))$, we will define a reduction relation from computations to such trees of values.

Given $f: X \to Y$ and a tree $t \in TX$, we write $t[x \mapsto f(x)] \in TY$ or $Tf(t) \in TY$ for the tree whose leaves $x \in X$ are renamed to $f(x)$. We have a function $\mu: TTX \to TX$, which takes a tree $r$ of trees and flattens it to a tree
$\mu r \in TX$, by taking the tree labelling each non-$\bot$ leaf of $r$ to be the subtree rooted at the corresponding node in $\mu r$.
The function $\mu$ is the multiplication associated with the monad structure of the $T$ operation. The unit of the monad is the map $\eta: X \rightarrow TX$ which takes an element $x \in X$ and returns the leaf labelled $x$ qua tree.

The operational mapping from a computation $M \in \textit{Com}(\TypeOne)$ to an effect tree is defined intuitively as follows. Start evaluating the $M$ in the empty stack $\textit{id}$, until the evaluation process  (which is deterministic) terminates.
If termination never happens the tree is $\bot$.
If the evaluation process terminates at a configuration of the form $(\textit{id}, \textbf{return } V)$ then the tree is the leaf $V$. Otherwise the evaluation process can only terminate at a configuration of the form 
$(S, \sigma(\dots))$ for some effect operation $\sigma \in \Sigma$. In this case, create an internal node in the tree of the appropriate kind (depending on  $\sigma$) and continue generating each child tree of this node by repeating the above process by evaluating an appropriate continuation computation, starting from a configuration with the current stack $S$. 

The following (somewhat technical) definition formalises the idea outlined above in a mathematically concise way.
We define a family of maps $|-,-|_{(-)} : \textit{Stack}(\TypeOne,\TypeTwo) \times \textit{Com}(\TypeOne) \times \NatNum \rightarrow \textit{Tree}(\TypeTwo)$ indexed over $\TypeOne$ and $\TypeTwo$ by:
{\small
\begin{align*}
|S,M|_0   & =   \bot \\
|S,M|_{n+1}   & = \! 
\begin{cases}
V & \text{if } S = \textit{id} \wedge M = \textbf{return } V \\
|S',M'|_n & \text{if } (S,M) \rightarrowtail (S',M') \\
\sigma(|S,M_0|_n,\dots,|S,M_{m-1}|_n) & \sigma \!: \! \alpha^m\! \rightarrow\! \alpha, M \!= \!\sigma(M_0,\dots,M_{m-1}) \\
\sigma(|S,V \overline{0}|_n,|S,V \overline{1}|_n,\dots) & \sigma \!:\! \alpha^{\NatType} \!\rightarrow \!\alpha, M \!= \!\sigma(V)\\
\sigma_k(|S,M_0|_n,\dots,|S,M_{m-1}|_n) &  \sigma \!:\! \NatType \!\times\! \alpha^m \!\rightarrow\! \alpha, M \!= \!\sigma(\overline{k}; M_0,\dots, M_{m-1}) \\
\sigma_k(|S,V \overline{0}|_n,|S,V \overline{1}|_n,\dots) & \sigma\!: \!\NatType \!\times \!\alpha^{\NatType}\! \rightarrow \!\alpha, M \!= \!\sigma(\overline{k}; V)\\
\bot & \text{otherwise}
\end{cases}
\end{align*}}
$\!$It follows that $|S,M|_n \leq |S,M|_{n+1}$ in the given ordering on trees. We write 
$$|-|_{(-)}: \textit{Com}(\TypeOne) \times \NatNum \rightarrow \textit{Tree}(\TypeOne)$$
 for the function defined by $|M|_n = |\textit{id},M|_n$. Using this we can give the operational interpretation of computation terms as effect trees
by defining $|-|: \textit{Com}(\TypeOne) \rightarrow \textit{Tree}(\TypeOne)$ by
$$|M| := \bigsqcup_n |M|_n \enspace . $$

We illustrate the above definitions with a couple of examples of effect computations and their corresponding effect trees.

\setcounter{myexampleA}{1}
\begin{myexampleA}[Nondeterminism] 
	\label{example:nonA}
	Nondeterministically generate a natural number:
	\begin{center}
		$?N := \textbf{let } \textbf{fix} (\lambda x:\mathbf{1} \to \mathbf{N}.\, \textit{or}(\lambda y:\mathbf{1}.\,Z \, , \, \lambda y:\mathbf{1}. \,\textbf{let } xy \Rightarrow z \textbf{ in } S(z))) \Rightarrow w \textbf{ in } w *$
		\\[1ex]
		$\xymatrix{ 
			& & \textit{or} \ar@{-}[ld] \ar@{-}[rd] & & & \\
			& \overline{0} & & \textit{or} \ar@{-}[ld] \ar@{-}[rd] & & \\
			|?N| = & & \overline{1} & & \textit{or} \ar@{-}[ld] \ar@{.}[rd] & \\
			& & & \overline{2} & &  \\
		}$
	\end{center}
\end{myexampleA}

\setcounter{myexampleA}{3}
\begin{myexampleA}[Global store] 
	\label{example:gsA}
	Save and load a value, returning its successor:
	\begin{center}
		$V := \lambda y : \NatType.\textit{update}_l(y;\textit{lookup}_l(\lambda x : \NatType.\, \textbf{return}(S(x)))) \quad : \quad \NatType$ $\rightarrow \NatType$
		\\[1ex]
		$\xymatrix{ 
			& & & \textit{update}_l(1) &  & \\
			|V \overline{1}| = &  &  & \textit{lookup}_l \ar@{-}[u]	&  & \\	
			& \overline{1} \ar@{-}[rru]|-{0} & \quad \overline{2} \ar@{-}[ru]|-{1} & \overline{3} \ar@{-}[u]|-{2} & \dots & \overline{n+1} \ar@{-}[llu]|-{n}\\
		}$
	\end{center}
\end{myexampleA}

In the second example above, we see that the resulting tree exhibits redundancies with respect to the expected model of computation with global store. Since the $\textit{update}_l$ operation sets the value of location $l$ to $1$, the ensuing $\textit{lookup}_l $ operation will retrieve the value $1$, and so execution will proceed down the branch labelled $1$ resulting in the return value $2$. The other infinitely many leaves  of the tree are redundant. The issue here is that the operational semantics of the language has been defined independently of any implementation model for  the effect operations. An effect tree provides a normal form that records in its nodes all effect operations that may potentially be performed during execution, and the dependencies between them. But nothing is 
stated about which effect operations will actually be performed in practice, and what effect they have if invoked. 
It is precisely this lack of specificity that allows the operational semantics to be defined in a uniform way depending only on the signature of effect operations and their arities.

In order to be able to reason about programs with effects (for example, to establish properties of or equivalences between them), it is necessary to supply the missing information about how effect operations  {behave} when executed. As motivated in the introduction, we now proceed to do this by introducing a \emph{behavioural logic} for expressing behavioural properties of our language.

%% file: 3_Logic.tex
\section{Behavioural logic and modalities}
\label{section:logic}

The goal of this section is to motivate and formulate a logic for expressing \emph{behavioural properties} of programs. In our language, program means (well-typed) term, and we shall be interested both in properties of \emph{computations} and in properties of \emph{values}.
Accordingly, we define a logic that contains both \emph{value formulas} and \emph{computation formulas}. We shall use lower case Greek letters $\phi,\psi,\dots$ for the former, and upper case Greek letters $\Phi,\Psi,\dots$ for the latter. Our logic will thus have two satisfaction relations
\[
V \models \phi \qquad \qquad \qquad M \models \Phi
\]
which respectively assert that ``value $V$ enjoys the value property expressed by $\phi$'' and ``computation $M$ enjoys the computation property expressed by $\Phi$''.

In order to motivate the detailed formulation  of the logic, it is useful to identify criteria that will guide the design.

\begin{description}

\item[(C1)] \label{item:meaningful} The logic should express only `behaviourally meaningful' properties of programs. This guides us to build the logic upon primitive notions that have a direct behavioural interpretation according to a natural understanding of program behaviour. 

\item[(C2)] The logic should be as expressive as possible within the constraints imposed by criterion~(C1).
\end{description}

For every type $\TypeOne$, we define a collection $\textit{VF}\,(\TypeOne)$ of \emph{value formulas}, and a collection $\textit{CF}\,(\TypeOne)$
of \emph{computation formulas}, as motivated above.

Since boolean logical connectives say nothing themselves about computational behaviour, it is a reasonable general principle that `behavioural properties' should be closed under such connectives. Thus, in keeping with  criterion (C2), which asks for maximal expressivity, we  close
each set $\textit{CF}\,(\TypeOne)$ and $\textit{VF}\,(\TypeOne)$, of computation and value formulas, under infinitary propositional logic. 
 
In addition to closure under infinitary propositional logic, each set $\textit{VF}\,(\TypeOne)$ contains a collection of \emph{basic} value formulas, from which compound formulas are constructed using (infinitary) propositional connectives.\footnote{We call such formulas \emph{basic} rather than \emph{atomic} because they include formulas such as $(V \mapsto \Phi)$, discussed below, which are built from other formulas.}
The choice of basic formulas depends on the type $\TypeOne$.
 
In the case of the natural numbers type, we include a basic value formula $\{n\} \in \textit{VF}\,(\NatType)$, for every $n \in \NatNum$. The semantics of this formula are given by:
\[V \models \{n\} ~~~ \Leftrightarrow ~~~ V = \overline{n} \enspace .\]
By the closure of $ \textit{VF}\,(\NatType)$ under infinitary disjunctions, 
every subset of $\NatNum$ can be represented by some value formula. Moreover, since a general value formula in $ \textit{VF}\,(\NatType)$ 
is an infinitary
boolean combination of basic formulas of the form $\{n\}$, every value formula corresponds to a  subset of  $\NatNum$.

For the unit type, we do not require any basic value formulas. The unit type has only one value, $*$. The 
two subsets of this singleton set of values are defined by the formulas $\bot$ (`falsum', given as an empty disjunction), and $\top$ (the truth constant, given as an empty conjunction).

For a function type $\TypeOne \to \TypeTwo$, we want each basic formula to express a fundamental behavioural constraint on values (i.e., $\lambda$-abstractions) $W$
of type $\TypeOne \to \TypeTwo$. In keeping with the applicative nature of functional programming, the only way in which a $\lambda$-abstraction can be used to 
generate behaviour is to apply it to an argument of type $\TypeOne$, which, because we are in a call-by-value setting, must be a value $V$.
The application of $W$ to $V$ results in a computation $WV$ of type $\TypeTwo$, whose properties can be probed using computation formulas in
$\textit{CF}\,(\TypeTwo)$. Based on this, for every value $V \in \textit{Val}(\TypeOne)$ and computation formula $\Phi \in \textit{CF}\,(\TypeTwo)$, we include
a basic value formula $(V \mapsto \Phi) \in \textit{VF}\,(\TypeOne \to \TypeTwo)$ with the  semantics:
\[W \models (V \mapsto \Phi) ~~~ \Leftrightarrow ~~~ {WV} \models  \Phi \enspace  .\]
Using this simple construct, based on application to a single argument $V$, other natural mechanisms for expressing properties of 
$\lambda$-abstractions are definable, using infinitary propositional logic. For example,  given $\phi \in \textit{VF}\,(\TypeOne)$ and
$\Psi \in \textit{CF}\,(\TypeTwo)$, the definition
\begin{equation}
\label{equation:define-formula}
(\phi \mapsto \Psi) ~ := ~ \bigwedge \{(V \mapsto \Psi) \mid V \in \textit{Val}(\TypeOne),\, V \models \phi\} \enspace 
\end{equation}
defines a formula whose derived semantics is
\begin{equation}
\label{equation:derived-semantics}
W \models (\phi \mapsto \Psi) ~~~ \Leftrightarrow ~~~ \forall {V \!\in\! \textit{Val}(\TypeOne)},~ V \models \phi ~\text{implies}~ {WV} \models  \Psi \enspace  .
\end{equation}
In Section~\ref{section:pure}, we shall consider the possibility of changing the basic value formulas in $\textit{VF}\,(\TypeOne \to \TypeTwo)$
to formulas $(\phi \mapsto \Psi)$. 

It remains to explain how the basic computation formulas in $\textit{CF}\,(\TypeOne)$ are formed. For this we require a given set $\Obser$ of
\emph{modalities}, which depends on the algebraic effects  contained in the language. The basic computation formulas in $\textit{CF}\,(\TypeOne)$ then have the form $o\,\phi$, where $o \in \Obser$ is one of the available modalities, and $\phi$ is a value formula in $\textit{VF}\,(\TypeOne)$.
Thus a modality lifts properties of values of type $\TypeOne$ to properties of computations of type $\TypeOne$.

In order to give semantics to computation formulas $o\, \phi$, we need a general theory of the kind of modality under consideration.
This is one of the main contributions of the paper. Before presenting the general theory, we first consider motivating examples, using our running examples of algebraic effects.

\setcounter{myexampleB}{-1}
\begin{myexampleB}[Pure functional computation]
Define $\Obser = \{{\downarrow}\}$. Here the single modality $\downarrow$ is the \emph{termination modality}: ${\downarrow\! \phi}$ asserts that a computation terminates with a return value $V$ satisfying $\phi$. This is formalised using effect trees:
\[
M \models {\downarrow \! \phi} ~~~ \Leftrightarrow ~~~ 
\text{$|M|$ is a  leaf \,$V$ and  $V \models \phi$} \enspace .
\]
Note that, in the case of pure functional computation, all trees are leaves: either value leaves $V$, or nontermination leaves $\bot$.
\end{myexampleB}

\begin{myexampleB}[Error] Define $\Obser = \{{\downarrow}\} \cup \{\mathsf{E}_e \mid e \in E\}$.
The semantics of the termination modality $\downarrow$ is defined as above. The \emph{error modality} $\mathsf{E}_e$ flags error $e$:
\[
M \models {\mathsf{E}_e \phi} ~~~ \Leftrightarrow ~~~ 
\text{$|M|$ is a node labelled with $\textit{raise}_e$} \enspace .
\]
(Because $\textit{raise}_e$ is an operation of arity $0$, a $\textit{raise}_e$ node in a tree has $0$ children.)
Note that the semantics of $\mathsf{E}_e \phi$ makes no reference to $\phi$. Indeed it would be natural to consider $\mathsf{E}_e$ as a basic computation formula in its own right, which could be done by introducing a notion of $0$-argument modality, and considering $\mathsf{E}_e$ as such. 
In this paper, however, we keep the treatment uniform by always considering modalities as unary operations, with natural $0$-argument modalities subsumed as unary modalities with a redundant argument.
\end{myexampleB}

\begin{myexampleB}[Nondeterminism] Define $\Obser = \{\Diamond, \, \Box\}$ with:
\begin{align*}
M \models \Diamond \phi ~~~ & \Leftrightarrow ~~~
  \text{$|M|$ has some leaf $V$ such that $V \models \phi$}
\\
M \models \Box \phi ~~~ & \Leftrightarrow ~~~
  \text{$|M|$ has finite height and every leaf is a value $V$ s.t.\ $V \models \phi$ \, .}
\end{align*}
Including both modalities amounts to a neutral view of nondeterminism. In the case of angelic nondeterminism, one would  include just the $\Diamond$ modality; in that of demonic nondeterminism, just the $\Box$ modality. Because of the way the semantic definitions interact with termination, the modalities $\Box$ and $\Diamond$ are not De Morgan duals. Indeed, each of the  
three possibilities $\{\Diamond, \, \Box\}, \{\Diamond\}, \{\Box\}$ for $\Obser$ leads to a logic with a different expressivity.
\end{myexampleB}

\begin{myexampleB}[Probabilistic choice] Define $\Obser = \{ \mathsf{P}_{>q} \mid q \in \mathbb{Q},\, 0 \leq q < 1\}$ with:
\begin{align*}
M \models  \mathsf{P}_{>q} \,\phi ~~~ & \Leftrightarrow ~~~ 
\mathbf{P}(\text{$|M|$ terminates with a value in $\{V \mid V \models \phi\}$}) > q\, ,
\end{align*}
where the probability on the right is the probability that a run through the tree $|M|$, starting at the root, and making an independent fair probabilistic choice at each branching node, terminates at a value node with a value $V$ in the set $\{V \mid V \models \phi\}$. We observe that the restriction to 
rational thresholds $q$ is immaterial, as, for any real $r$ with $0 \leq r < 1$, we can define:
\[
\mathsf{P}_{>r} \,\phi ~ := ~ \bigvee \{\mathsf{P}_{>q} \,\phi \mid q \in \mathbb{Q},\, r < q < 1\} \enspace .
\]
Similarly, we can define non-strict threshold modalities, for $0 < r \leq 1$, by:
\[
\mathsf{P}_{\geq r} \,\phi ~ := ~ \bigwedge \{\mathsf{P}_{>q} \,\phi \mid q \in \mathbb{Q},\, 0 \leq q < r \} \enspace .
\]
Also, we can exploit negation to define modalities expressing strict and non-strict upper bounds on probabilities.
Notwithstanding the definability of non-strict and upper-bound thresholds, we shall see later that it is important that we include
only strict lower-bound modalities
in our set $\Obser$ of  primitive modalities.
\end{myexampleB}
\begin{myexampleB}[Global store]
	\label{exampleB:gs} 
	Given the set of locations $L$, we define the set of states by $\textit{State} := \NatNum^L$. The modalities are 
$\Obser = \{(s \rightarrowtail s') \mid s,s' \in \textit{State}\}$, where  informally:
\begin{align*}
M \models  (s \rightarrowtail s')\, \phi ~~~  \Leftrightarrow ~~~  & 
\text{the execution of $M$, starting in state $s$, terminates in} \\
&  ~~~ \text{final state $s'$ with return value $V$ such that $V \models \phi$ \enspace .}
\end{align*}
We make the above definition precise using the effect tree of $M$. Define
\[\textit{exec}: TX \times\textit{State} \rightarrow X \times \textit{State} \, , \] 
for any set $X$,
to be the least partial function satisfying:
{\small
\begin{align*}
\textit{exec}(t,s) ~ =  \begin{cases}
(x,s) & \text{if $t$ is a leaf labelled with $x \in X$}
\\
\textit{exec}(t_{s(l)}, s) & \text{if $t = \textit{lookup}_l(t_0,t_1,\cdots)$ and $\textit{exec}(t_{s(l)}, s)$ is defined}
\\
\textit{exec}(t', s[l:=n]) & \text{if $t = \textit{update}_{l,n}(t')$ and $\textit{exec}(t', s[l:=n])$ is defined\, ,}
\end{cases}
\end{align*}}
$\!\!$where $s[l:=n]$ is the evident modification of state $s$. Intuitively, $\textit{exec}(t,s)$ defines the result of ``executing'' the tree of commands in 
effect tree $t$ starting in state $s$, whenever this execution terminates. In terms of operational semantics, it can be viewed as defining a `big-step' semantics for effect trees (in the signature of global store).
We can now define the semantics of the 
$ (s \rightarrowtail s')$ modality formally:
\begin{align*}
M \models  (s \rightarrowtail s')\, \phi ~~~ & \Leftrightarrow ~~~ \textit{exec}(|M|,s) = (V,s')~\text{where}~V \models \phi \enspace.
\end{align*}

\noindent
In Section~\ref{section:reasoning}, we show an example of how to encode Hoare Logic 
in the above logic.
\end{myexampleB}

\begin{myexampleB}[Input/output]
Define an \emph{i/o-trace} to be a word $w$ over the alphabet
\[
\{?n \mid n \in \NatNum\} \cup \{!n \mid n \in \NatNum\} \, .
\]
The idea is that such a word represents an input/output sequence, where $?n$ means the number $n$ is given in response to an input prompt, and 
$!n$ means that the program outputs $n$. Define the set of modalities 
\[\Obser = \{{\langle w \rangle\!\!\downarrow},\,
\langle w\rangle_{\!\dots} \mid \text{$w$ an i/o-trace}\}\, . \]
The intuitive semantics of these modalities is as follows.
\begin{align*}
M \models {\langle w \rangle\!\!\downarrow}\, \phi ~~~ \Leftrightarrow ~~~ &
\text{$w$ is a complete i/o-trace for the execution of $M$ }
\\ 
& ~~~\text{ resulting in  termination with $V$ s.t.\ $V \models \phi$ \,.}
\\[1ex]
M \models {\langle w\rangle_{\!\dots}} \, \phi ~~~ \Leftrightarrow ~~~ &
\text{$w$ is an initial i/o-trace for the execution of $M$ \,.}
\end{align*}
In order to define the semantics of formulas precisely, we first define relations $t \models  {\langle w \rangle\!\!\downarrow}\, P$  and
$t \models  {\langle w\rangle_{\!\dots}}$, between $t \in TX$ and $P \subseteq X$, by induction on words. 
(Note that we are overloading the $\models$ symbol.)
In the following, we write $\varepsilon$ for the empty word, and we use textual juxtaposition for concatenation of words.
\begin{align*}
t \models {\langle \varepsilon \rangle\!\!\downarrow}\, P ~~~ \Leftrightarrow ~~~ &
\text{$t$ is a leaf $x$ and $x \in P$}
\\ 
t \models {\langle (?n)\,w \rangle\!\!\downarrow}\, P ~~~ \Leftrightarrow ~~~ &
\text{$t = \textit{read}(t_0,t_1, \dots)$ and $t_n \models  {\langle w \rangle\!\!\downarrow}\, P$}
\\
t \models {\langle (!n)\,w \rangle\!\!\downarrow}\, P ~~~ \Leftrightarrow ~~~ &
\text{$t = \textit{write}_n(t')$ and $t' \models  {\langle w \rangle\!\!\downarrow}\, P$}
\\[1ex]
t \models {\langle \varepsilon \rangle_{\!\dots}}  ~~~ \Leftrightarrow ~~~ &
\text{true} 
\\
t \models {\langle (?n)\,w \rangle_{\!\dots}} ~~~ \Leftrightarrow ~~~ &
\text{$t = \textit{read}(t_0,t_1, \dots)$ and $t_n \models   {\langle w \rangle_{\!\dots}}$}
\\
t \models {\langle (!n)\,w \rangle_{\!\dots}} ~~~ \Leftrightarrow ~~~ &
\text{$t = \textit{write}_n(t')$ and $t' \models   {\langle w \rangle_{\!\dots}}$}
\end{align*}
The formal semantics of modalities is now easily defined by:
\begin{align*}
M \models {\langle w \rangle\!\!\downarrow}\, \phi ~~~ \Leftrightarrow ~~~ &
|M| \models {\langle w \rangle\!\!\downarrow}\, \{V \mid V \models \phi\}
\\[1ex]
M \models {\langle w\rangle_{\!\dots}} \, \phi ~~~ \Leftrightarrow ~~~ &
|M| \models {\langle w\rangle_{\!\dots}}\, .
\end{align*}
Note that, as in Example~\ref{example:error}, the formula argument of the ${\langle w\rangle_{\!\dots}}$ modality is redundant. 
Also, note that our modalities for input/output could naturally be formed by  combining the termination modality $\downarrow$, which lifts value formulas to computation formulas, with sequences of
atomic modalities $\langle?n\rangle$ and $\langle!n\rangle$ acting directly on computation formulas. In this paper, we  do not include such modalities, acting on computation formulas, in our general theory. But this is a natural avenue for future consideration.
\end{myexampleB}

We now give a formal treatment of the logic and its semantics, in full generality. We assume a signature $\Sigma$ of effect operations, as in Section~\ref{section:language}. We also assume a given set $\Obser$, whose elements we call \emph{modalities}.

We call our main behavioural logic $\LogNeg$, where the letter $\mathcal{V}$ is chosen as a reference to the fact that the basic formula at function type specifies function behaviour on individual value arguments $V$.
\begin{definition}[The logic $\LogNeg$] \label{logic}
The classes $\textit{VF}\,(\TypeOne)$ and $\textit{CF}\,(\TypeOne)$ of \emph{value} and \emph{computation formulas}, for each type $\tau$, are mutually inductively defined by the rules in Fig. \ref{figure:logneg}.
\begin{figure}[t]
\begin{gather*}
	\frac{n \in \NatNum}{\{n\} \in \textit{VF}\,(\NatType)}(1) \quad \quad
	\frac{V : \TypeOne \quad \quad \Phi \in \textit{CF}\,(\TypeTwo)}{(V \mapsto \Phi) \in \textit{VF}\,(\TypeOne \to \TypeTwo)}(2) \quad \quad
	\frac{\phi \in \textit{VF}\,(\TypeOne) \quad \quad o \in \Obser}{o \, \phi \in \textit{CF}\,(\TypeOne)}(3)
\\[1ex]
	\frac{\phi : I \rightarrow \textit{VF}\,(\TypeOne)}{\bigvee_I \phi \in \textit{VF}\,(\TypeOne)}(4) \quad \quad 
	\frac{\phi : I \rightarrow \textit{VF}\,(\TypeOne)}{\bigwedge_I \phi \in \textit{VF}\,(\TypeOne)}(5) \quad \quad
	\frac{\phi \in \textit{VF}\,(\TypeOne)}{\neg \phi \in \textit{VF}\,(\TypeOne)}(6)
\\[1ex]
	\frac{\Phi : I \rightarrow \textit{CF}\,(\TypeOne)}{\bigvee_I \Phi \in \textit{CF}\,(\TypeOne)}(7) \quad \quad 
	\frac{\Phi : I \rightarrow \textit{CF}\,(\TypeOne)}{\bigwedge_I \Phi \in \textit{CF}\,(\TypeOne)}(8) \quad \quad
	\frac{\Phi \in \textit{CF}\,(\TypeOne)}{\neg \Phi \in \textit{CF}\,(\TypeOne)}(9)
\end{gather*}
\caption{The logic $\LogNeg$}
\label{figure:logneg}
\end{figure}
In this, $I$ can be instantiated to any set, allowing for arbitrary conjunctions and disjunctions. When $I$ is $\emptyset$, we get the special formulas $\top = \bigwedge_{\emptyset}$ and $\bot = \bigvee_{\emptyset}$.
The use of arbitrary index sets means that  
formulas, as defined, form a proper class. However, we shall see below that countable index sets suffice.
\end{definition}

In order to specify the semantics of modal formulas, we require a connection between modalities and effect trees, which is given by
an interpretation function
\[
\denote{-} : \Obser \to \mathcal{P}(T \bf{1}) \, .
\]
That is, every modality $o \in \Obser$ is mapped to a subset $\denote{o} \subseteq  T\bf{1}$  of unit-type effect trees.
Given a subset $P \subseteq X$ (e.g. given by a formula) and a tree $t \in TX$ we can define a unit-type tree $\tis{t}{P} \in T\UnitType$ as the tree created by replacing the leaves of $t$ that belong to $P$ by $*$ and the others by $\bot$. 
In the case that $P$ is the subset $\{V \mid V \models \phi\}$
 specified by a formula $\phi \in \textit{VF}\,(\TypeOne)$, we also write $t[\,\models \phi]$ for $\tis{t}{P}$.

We now define the two satisfaction relations ${\models} \subseteq \textit{Val}(\TypeOne) \times \textit{VF}\,(\TypeOne)$ and
${\models} \subseteq \textit{Com}(\TypeOne) \times \textit{CF}\,(\TypeOne)$, mutually inductively, where for the basic formulas we have:
\begin{align*}
W \models \{n\} ~~~ & \Leftrightarrow ~~~ W = \overline{n} \\
W \models (V \mapsto \Phi) ~~~ &  \Leftrightarrow  ~~~ WV \models \Phi \\
M \models o\, \phi ~~~ & \Leftrightarrow ~~~ |M|\,[\,\models \phi] \in \denote{o} \,,
\end{align*}
and for the other formulas we have:
\begin{align*}
W \models \bigvee_I  \phi ~~~ & \Leftrightarrow ~~~ \exists i \in I, W \models \phi_i\\
W \models \bigwedge_I  \phi ~~~ & \Leftrightarrow ~~~ \forall i \in I, W \models \phi_i\\
W \models \neg \phi ~~~ &  \Leftrightarrow  ~~~ \neg(W \models \phi) \,.
\end{align*}
We remark that all conjunctions and disjunctions are semantically equivalent to countable ones, because value and computation formulas are interpreted 
over sets  of terms, $\textit{Val}(\TypeOne)$ and $\textit{Com}(\TypeOne)$, which are countable.

The lemma below is standard. It states that every formula in infinitary propositional logic can be written in infinitary disjunctive normal form. (It can also be written in infinitary conjunctive normal form.)
\begin{lemma}\label{lemma:normal}
	Each formula $\phi \in \LogNeg$ (value or computation) is equivalent to a formula of the form $\bigvee_I \bigwedge_J \psi$ where for each $i \in I$ and $j \in J_i$ the formula $\psi_{i,j}$ is either a basic formula 
	or the negation of a basic formula.
\end{lemma}

\noindent
We omit the proof, which is both routine and standard.

We end this section by revisiting our running examples, and observing, in each case, that the example modalities presented above are all specified by  suitable interpretation functions $\denote{-} : \Obser \to \mathcal{P}(T \bf{1})$.

\setcounter{myexampleC}{-1}
\begin{myexampleC}[Pure functional computation]
We have  $\Obser = \{{\downarrow}\}$. Define:
\[
\denote{{\downarrow}} ~ = ~ \{\, * \, \}~~~\text{(where $*$ is the tree with single node $*$)} \enspace .
\]
\end{myexampleC}

\begin{myexampleC}[Error] We have  $\Obser = \{{\downarrow}\} \cup \{\mathsf{E}_e \mid e \in E\}$. Define:
\[
\denote{{\mathsf{E}_e}} ~ = ~ \{\, \textit{raise}_e \, \}\enspace .
\]
\end{myexampleC}

\begin{myexampleC}[Nondeterminism] We have  $\Obser = \{\Diamond, \, \Box\}$. Define:
\begin{align*}
\denote{{\Diamond}} ~ & = ~ \{ t \mid \text{$t$ has some $*$ leaf} \} 
\\
\denote{{\Box}} ~ & = ~ \{ t \mid \text{$t$ is finite and every leaf is a  $*$} \} \enspace .
\end{align*}
\end{myexampleC}

\begin{myexampleC}[Probabilistic choice] $\Obser = \{ \mathsf{P}_{>q} \mid q \in \mathbb{Q},\, 0 \leq q < 1\}$. Define:
\[
\denote{\mathsf{P}_{>q}} ~ = ~ \{ t \mid 
\mathbf{P}(\,\text{$t$ terminates with a $*$ leaf}\,) > q\} \enspace .
\]
\end{myexampleC}
\begin{myexampleC}[Global store] 
$\Obser = \{(s \rightarrowtail s') \mid s,s' \in \textit{State}\}$.
Define:
\begin{align*}
\denote{(s \rightarrowtail s')}  ~& = ~ \{ t \mid  \textit{exec}(t,s) = (*,s')\} \enspace .
\end{align*}
\end{myexampleC}

\begin{myexampleC}[Input/output]
$\Obser = \{{\langle w \rangle\!\!\downarrow},\,
\langle w\rangle_{\!\dots} \mid \text{$w$ an i/o-trace}\}$. Define:
\begin{align*}
\denote{\langle w \rangle\!\!\downarrow \,}
~ & =  ~ \{ t \mid  t \models {\langle w \rangle\!\!\downarrow}\, \{*\}\, \}
\\
\denote{\langle w\rangle_{\!\dots} \,}
~ & =  ~ \{ t \mid  t \models   \langle w\rangle_{\!\dots}\, \} \enspace .
\end{align*}
\end{myexampleC}

In this section we have defined our logic expressing behavioural properties. We next proceed to derive the induced notion of \emph{behavioural equivalence} between programs, as motivated in Section~\ref{section:introduction}.

%% file: 4_Equivalence.tex
\section{Behavioural equivalence}
\label{section:preorder}

The goal of this section is to precisely formulate our main theorem: under suitable conditions, the \emph{behavioural equivalence} determined by the logic $\LogNeg$ of Section~\ref{section:logic} is a congruence. In addition, we shall obtain a similar result for a coarser \emph{behavioural preorder} determined by a natural 
\emph{positive fragment} of $\LogNeg$,  which we call $\LogPos$. In addition to being natural in its own right, the
preorder induced by the positive fragment  $\LogPos$   turns out to be an indispensable technical tool for establishing properties of the behavioural equivalence induced by the full logic $\LogNeg$.

\begin{definition}[The logic $\LogPos$] 
The logic $\LogPos$ is the fragment of $\LogNeg$ consisting of those formulas in 
$\textit{VF}\,(\TypeOne)$ and $\textit{CF}\,(\TypeOne)$ that do not contain negation. It is inductively defined using rules 1-5, 7 and 8 from Fig.~\ref{figure:logneg}.
\end{definition}

\noindent
Whenever we have a logic $\mathcal{L}$
whose value and computation formulas are given as subcollections 
$\textit{VF}_\mathcal{L}(\TypeOne) \subseteq \textit{VF}\,(\TypeOne)$ and $\textit{CF}_\mathcal{L}(\TypeOne) \subseteq \textit{CF}\,(\TypeOne)$, 
then $\mathcal{L}$ determines a {preorder} (and hence also an {equivalence relation}) between terms of the same type and aspect.

\begin{definition}[Logical preorder and equivalence]\label{definition:preorder}
Given a fragment $\mathcal{L}$ of $\LogNeg$, we define the \emph{logical preorder} $\sqsubseteq_{\mathcal{L}}$, 
between well-typed terms of the same type and aspect, by:
\begin{align*}
V \sqsubseteq_{\mathcal{L}} W ~~~  &  \Leftrightarrow ~~~ \forall \phi \in \textit{VF}_{\mathcal{L}}(\TypeOne), ~V \models \phi~ \Rightarrow ~ W \models \phi \\
M \sqsubseteq_{\mathcal{L}} N ~~~ & \Leftrightarrow  ~~~ \forall \Phi \in \textit{CF}_{\mathcal{L}}(\TypeOne), ~ M \models \Phi~  \Rightarrow ~ N \models \Phi
\end{align*}
The \emph{logical equivalence} $\equiv_{\mathcal{L}}$ on terms is the equivalence relation induced by the preorder
(the intersection of $\sqsubseteq_{\mathcal{L}}$ and its converse).
\end{definition}

\noindent
In the case that formulas in $\mathcal{L}$ are closed under negation, the preorder  $\sqsubseteq_{\mathcal{L}}$ is already an equivalence relation, and hence coincides with $\equiv_{\mathcal{L}}$. Thus we shall only refer specifically to the preorder  $\sqsubseteq_{\mathcal{L}}$, for fragments, such as $\LogPos$, that are not closed under negation. 

The two main  relations of interest to us in this paper are the primary relations determined by $\LogNeg$ and $\LogPos$: 
full \emph{behavioural equivalence} $\equiv_{\LogNeg}$; and the \emph{positive behavioural preorder} $\PrePos$ (which induces 
\emph{positive behavioural equivalence} $\equiv_{\LogPos}$).
Since $\LogPos$ is a subset of $\LogNeg$, it is apparent that $\equiv_{\LogNeg}$ is finer than $\equiv_{\LogPos}$, as it considers more behavioural properties which could distinguish terms.
For the same reason, it holds that $\sqsubseteq_{\LogNeg} \,\subseteq \,\,\sqsubseteq_{\LogPos}$.

Before formulating the required notions to prove congruence of the behavioural equivalences, we shall make some observations about the preorders and discuss a possible simplification of the logic (Proposition~\ref{proposition:logA}).

\begin{lemma}\label{funct_clas}
	For any $V_0, V_1 \in \textit{Val}(\TypeTwo \rightarrow \TypeOne)$, we have  $V_0 \PrePos V_1$ if and only if:
	$$\forall W \in \textit{Val}(\TypeTwo),\,
	\forall \Psi \in \textit{CF}_{\LogPos}(\TypeOne), ~ 
	V_0 \models   (W \mapsto \Psi) ~ \text{implies}~ V_1 \models   (W \mapsto \Psi) \, . $$
\end{lemma}

\begin{lemma}\label{comp_clas}
	For any $M_0, M_1 \in \textit{Com}(\TypeOne)$, we have  $M_0 \PrePos M_1$ if and only if:
	$$\forall o \in \Obser,\, 
	\forall \phi \in \textit{VF}_{\LogPos}(\TypeOne), ~ 
	M_0 \models   o\,\phi ~ \text{implies}~ M_1 \models   o\,\phi \, . $$
\end{lemma}

\noindent
Both these lemmas are a consequence of the fact that satisfaction of conjunctions and disjunctions are completely determined by satisfaction of the formulas over which the connectives are taken. As such, the logical preorder is completely determined by satisfaction of basic formulas. Similar characterisations, but replacing `implies' and $\LogPos$ with `if and only if' and $\LogNeg$ respectively, hold for  the behavioural equivalence $\equiv_{\LogNeg}$.

\begin{proposition}\label{proposition:logA}
	Let \ $\mathcal{K}$ be the fragment of \ $\LogNeg$ inductively defined by rules 1 to 6 from Fig.~\ref{figure:logneg} (so computation formulas are not closed under propositional connectives), then the induced logical equivalence
	$\equiv_{\mathcal{K}}$ is the same as $\equiv_{\LogNeg}$. 
\end{proposition}
\begin{proof}
	We prove that any value formula $\phi \in \LogNeg$ is equivalent to a value formula $\phi'$ from $\mathcal{K}$. We do this by induction on types.
	For value formulas of natural numbers type, note that $\textit{VF}_{\LogNeg}(\NatType) = \textit{VF}_{\mathcal{K}}(\NatType)$, so the statement is trivially true by taking $\phi' := \phi$.

	For value formulas of function type, assume $\phi$ is a basic formula $V \mapsto \Phi$, where by Lemma \ref{lemma:normal} we may assume w.l.o.g. that $\Phi$ is a disjunction over conjunctions over formulas $\Psi$ of the form $o\,\psi$ or $\neg o\,\psi$. We can now use the equivalences $(V \mapsto \bigvee_I \bigwedge_J \Psi) \equiv \bigvee_I \bigwedge_J (V \mapsto \Psi)$ and $(V \mapsto \neg o \, \psi) \equiv \neg (V \mapsto o \, \psi)$ to construct a formula $\phi' \in \mathcal{K}$ equivalent to $\phi$, using the induction hypothesis to replace each occurrence of $o \, \psi$ with $o \, \psi'$ where $\psi' \in \mathcal{K}$ is equivalent to $\psi$. In the case that $\phi$ is not a basic formula, we can do an induction on its structure to find the desired $\phi'$, where $(\bigvee_I \psi)' := \bigvee_I \psi'$,  $(\bigwedge_I \psi)' := \bigwedge_I \psi'$, $(\neg \psi)' := \neg (\psi')$, and basic formulas are handled as above.
	
	So every value formula has an equivalent  value formula in $\mathcal{K}$, so the logical preorder on value terms remains unchanged. To see that the logical equivalence on computation terms remains unchanged, simply use Lemma \ref{comp_clas}. 
\end{proof}

Altering the proof slightly, we can derive a similar result for the positive logic.
\begin{proposition}\label{proposition:logB}
	Let \ $\mathcal{K}^+$ be the fragment of \ $\LogPos$ defined by rules 1 to 5 from Fig.~\ref{figure:logneg} (so computation formulas are not closed under propositional connectives), then the induced logical preorder \
	$\sqsubseteq_{\mathcal{K}^+}$ is the same as \ $\PrePos$. 
\end{proposition}

We next formulate the appropriate technical notion of (pre)congruence\footnote{A \emph{precongruence} is a preorder enjoying the properties of a congruence relation apart from symmetry.} to apply to the relations $\equiv_{\LogNeg}$ and
$\PrePos$. 
These two preorders are examples of \emph{well-typed relations} on closed terms. Any such relation can be extended to a relation on open terms in the following way. Given a well-typed relation $\mathcal{R}$ on closed terms, we define the \emph{open extension} $\mathcal{R}^{\circ}$ where $\Gamma \vdash M \,\mathcal{R}^{\circ} N : \TypeOne$ precisely when, for every well-typed vector of closed values $\overrightarrow{V} : \Gamma$, it holds that  $M[\overrightarrow{V}] \,\mathcal{R}\, N[\overrightarrow{V}]$.  
The correct notion of precongruence for a well-typed preorder on closed terms, is to ask for its open extension to be  \emph{compatible} in the sense of the definition below; see, e.g., Lassen and Pitts~\cite{Lassen,Pitts00} for further explanation.

\begin{definition}[Compatibility]\label{open}
	A well-typed open relation $\mathcal{R}$ is said to be \emph{compatible} if it is closed under the rules in Fig.~\ref{fig:compat}.
\end{definition}

\begin{figure}
	\begin{gather*}
	\frac{}{\Gamma, x: \TypeOne \vdash x \, \mathcal{R} \, x : \TypeOne} \quad \quad
	\frac{}{\Gamma \vdash Z \, \mathcal{R} \, Z : \NatType} \quad \quad
	\frac{\Gamma \vdash V \, \mathcal{R} \, V' : \NatType}{\Gamma \vdash S(V) \, \mathcal{R} \, S(V') : \NatType}  
	\\[1ex]
	\frac{\Gamma \vdash V\, \mathcal{R} \, V' : \TypeOne}{\Gamma \vdash \textbf{return}(V) \, \mathcal{R} \, \textbf{return}(V') : \TypeOne} \quad \quad
	\frac{\Gamma, x : \TypeOne \vdash M \, \mathcal{R} \, M' : \TypeTwo}{\Gamma \vdash (\lambda x:\TypeOne.M)  \, \mathcal{R} \, (\lambda x:\TypeOne.M') : \TypeOne \to \TypeTwo}
	\\[1ex]
	\frac{\Gamma \vdash V \, \mathcal{R} \, V' : \TypeOne \to \TypeTwo \quad \quad \Gamma \vdash W \, \mathcal{R} \, W' : \TypeOne}{\Gamma \vdash (V W)  \, \mathcal{R} \, (V' W') : \TypeTwo} \quad \quad
	\frac{\Gamma \vdash V \, \mathcal{R} \, V' : (\TypeOne \to \TypeTwo) \to (\TypeOne \to \TypeTwo)}{\Gamma \vdash \text{\textbf{fix}}(V) \, \mathcal{R} \, \text{\textbf{fix}}(V') : \TypeOne \to \TypeTwo}
	\\[1ex]
	\frac{\Gamma \vdash V \, \mathcal{R} \, V' : \NatType \quad \quad \Gamma \vdash M \, \mathcal{R} \, M' : \TypeOne \quad \quad \Gamma, x:\NatType \vdash N \, \mathcal{R} \, N' : \TypeOne}
	{\Gamma \vdash \text{ \textbf{case} } V \text{ \textbf{of} } \{Z \Rightarrow M; S(x) \Rightarrow N\} \, \mathcal{R} \, \text{ \textbf{case} } V' \text{ \textbf{of} } \{Z \Rightarrow M'; S(x) \Rightarrow N'\} : \TypeOne}
	\\[1ex]
	\frac{\Gamma \vdash M \, \mathcal{R} \, M' : \TypeOne  \quad \quad \Gamma, x:\TypeOne \vdash N \, \mathcal{R} \, N' : \TypeTwo}{\Gamma \vdash \textbf{let } M \Rightarrow x \textbf{ in } N \, \mathcal{R} \, \textbf{let } M' \Rightarrow x \textbf{ in } N' : \TypeTwo}	
	\\[1ex]
	\frac{\Gamma \vdash M_i \, \mathcal{R} \, M_i' : \TypeOne}{\Gamma \vdash \sigma(M_0,M_1,...) \, \mathcal{R} \, \sigma(M_0',M_1',...) : \TypeOne} \quad
	\frac{\Gamma \vdash V \, \mathcal{R} \, V' : \NatType \quad \quad \Gamma \vdash M_i \, \mathcal{R} \, M_i' : \TypeOne}{\Gamma \vdash \sigma(V;M_0,M_1,...) \, \mathcal{R} \, \sigma(V';M_0',M_1',...) : \TypeOne}
	\\[1ex]
	\frac{\Gamma \vdash V \, \mathcal{R} \, V' : \NatType \to \TypeOne}{\Gamma \vdash \sigma(V) \, \mathcal{R} \, \sigma(V') : \TypeOne} \quad \quad 
	\frac{\Gamma \vdash V \, \mathcal{R} \, V' : \NatType \quad \quad \Gamma \vdash W \, \mathcal{R} \, W' : \NatType \to \TypeOne}{\Gamma \vdash \sigma(V;W) \, \mathcal{R} \, \sigma(V';W') : \TypeOne}
	\end{gather*}
	\caption{Rules for compatibility}
	\label{fig:compat}
\end{figure}

We now state our main congruence result, although we have not yet defined the conditions it depends upon.

\begin{theoremm}\label{Main_Theorem}
	If \ $\Obser$ is a decomposable set of Scott-open modalities then the open extensions of \
	 $\equiv_{\LogNeg}$ and \ $\PrePos$ are both compatible. (It is an immediate consequence that the open extension of \
	  $\equiv_{\LogPos}$ is also compatible.)
\end{theoremm}

The Scott-openness condition refers to the \emph{Scott topology} on $T\UnitType$, which can be divided into two properties.
\begin{definition}
We say that $o \in \Obser$ is \emph{upwards closed} if $\denote{o}$ is an upwards closed subset of $T\UnitType$; i.e.,
if $t \in \denote{o}$ implies $t'  \in \denote{o}$ whenever $t  \leq t'$.
\end{definition}
\begin{definition}
The modality $o \in \Obser$ is \emph{Scott-open} if $\denote{o}$ is an open set in the Scott topology on  $T\UnitType$; i.e.,
$\denote{o}$ is upwards closed and, whenever $t_1 \leq t_2 \leq \dots$ is an ascending chain in  $T\UnitType$ with supremum $\sqcup_i t_i \in \denote{o}$, we have  $t_n \in  \denote{o}$ for some $n$.
\end{definition}

We now turn to the \emph{decomposability} condition in Theorem~\ref{Main_Theorem}, which will
eventually be given in Definition~\ref{definition:decomposable} below, following the introduction of certain auxiliary relations that are required for the definition. Because the formulation is unavoidably technical, we first give
some motivation for where we are heading.

The main purpose of the decomposability property is to enable us to prove that the positive behavioural preorder is preserved by the $\textbf{let}$ term constructor. (This motivation will eventually manifest itself in case 5 in the proof of Lemma \ref{lemma_key} in Appendix \ref{section:proofs}.) The particular technical challenge presented by $\textbf{let}$ is that it sequences effects. Semantically, the operation of effect sequencing is distilled into the 
monad multiplication map  $\mu : TTX \to TX$ defined in Section \ref{section:language}. Accordingly, we formulate the required notion of decomposability as a property of $\mu$. It turns out that we need only consider $\mu$ in the case that $X$ is the singleton set $\{\ast\}$. Thus decomposability will involve  trees of unit type, that is trees in $T\UnitType$, as well as \emph{double trees}, that is trees in $\Dtrees$.
We shall define it as the property that the monad multiplication $\mu$ is order-preserving with respect to 
preorders $\modbel$  and $\prebas$ defined on double and single trees respectively.

The relation $\prebas$ on $T\UnitType$ is a natural extension of  the positive behavioural preorder $\PrePos$, at unit type, 
from a relation on computation terms to a relation $\prebas$ on arbitrary effect trees.
To accommodate this perspective, we introduce a new notation, allowing us to apply a modality $o$ to an arbitrary subset of a set $X$ of values: for $A \subseteq X$ and $o \in \Obser$ we define
 $$o_X(A) ~ = ~ \{t \in TX \mid t[{\in A}] \in \denote{o}\} ~  \subseteq  ~ TX \enspace . $$
We will often write $o(A)$ instead of $o_X(A)$ when $X$ is clear from the context. For instance, if $t \in TX$ then $t \in o(A)$ means $t \in o_X(A)$.
One case of particular interest is when $A = X = \{\ast\}$, for which we note that $o(\{\ast\}) = \denote{o}$. 
Using the extended interpretation of modalities, we similarly extend the satisfaction relation of type $\UnitType$ positive computation formulas, from computation terms of type $\UnitType$ to arbitrary $T\UnitType$ computation trees.\footnote{In general, there are uncountably many trees, whereas there are only countably many computation terms.}
For each $\Phi \in \Treef$, we define its denotation $\denote{\Phi} \subseteq T\UnitType$ over unit type trees inductively, according to the following rules:\footnote{For a set $X$ of formulas, we denote $\bigvee X$ and $\bigwedge X$ for respectively the disjunction and conjunction over X.}
\begin{align*}
	& \denote{o(\top)} := o(\{\ast\}), & \quad \quad &
	\denote{o(\bot)} := o(\emptyset),\\
	& \Denote{\bigvee X} := \bigcup \{\denote{\Phi} \mid \Phi \in X\}, & \quad \quad & \Denote{\bigwedge X} := \bigcap \{\denote{\Phi} \mid \Phi \in X\}.
\end{align*}

\noindent
The above definition can be seen as an extension of the satisfaction relation on unit type computation terms, since for any computation term $M : \UnitType$ and formula $\Phi \in \textit{CF}_{\LogPos}(\UnitType)$ it holds that $M \models \Phi \Leftrightarrow |M| \in \denote{\Phi}$.
 
We can now define the required  preorder between arbitrary trees of type  $\UnitType$.

\begin{definition}
	We define the preorder $\prebas$ on $T\UnitType$ by: \,for any two trees $t,t' \in T\UnitType$, 
	$$t \prebas t' \quad:\iff\quad \forall \Phi \in \Treef,~ t \in \denote{\Phi} \Rightarrow t' \in \denote{\Phi}.$$
\end{definition} 

It is immediate that $\prebas$ is a conservative extension of the positive behavioural preorder on  computation terms of type $\UnitType$ in the following sense:

\begin{proposition}\label{prop:preceq}
	For any $M, N \in \textit{Com}(\UnitType)$, it holds that  $|M| \prebas |N|$ if and only if
	$M \PrePos N$.
\end{proposition}

We give some alternative characterisations of $\prebas$.

\begin{proposition}\label{prop:base_relation}
	For any $t,t' \in T\UnitType$, the following three statements are equivalent:
	\begin{enumerate}
		\item $t \prebas t'$.
		\item \label{base:two} $(t \prebas t') \wedge (t[\in \emptyset] \prebas t'[\in \emptyset])$.
		\item \label{base:three} $\forall o \in \Obser, ~ (t \in o(\{\ast\}) \Rightarrow t' \in o(\{\ast\}))\,  \wedge\,  (t \in o(\emptyset)  \Rightarrow t' \in o(\emptyset))$.
	\end{enumerate}
\end{proposition}
\begin{proof}
	The equivalence (1) $\Leftrightarrow$ (3)
	follows from a straightforward induction on the structure of positive unit type computation formulas. 
	Using the previous equivalence, the equivalence (1) $\Leftrightarrow$ (2) follows from the fact that $t[\in \emptyset] \in o(\{\ast\}) \Leftrightarrow t \in o(\emptyset)$.
\end{proof}

In the examples given in this paper, the preorder can be characterised in a simpler way, where 
\begin{equation}
\label{equation:prebas_reduce}
t \prebas t' \quad\iff\quad \forall o \in \Obser, ~ t \in \denote{o} \Rightarrow t' \in \denote{o}.
\end{equation}
This is a consequence of the fact that each modality $o$ given in the examples satisfy one of the two following properties: 
\begin{itemize}
	\item[(i)] $\forall t \in T\UnitType, t \notin o(\emptyset)$. The modalities with this property are: $\downarrow, \Diamond, \Box, \mathsf{P}_{>q}, (s \rightarrowtail s'),$ and $\langle w \rangle\!\downarrow$.
	\item[(ii)] $\forall t \in T\UnitType, t \in o(\{\ast\}) \Leftrightarrow t \in o(\emptyset)$. The modalities with this property are $\mathsf{E}_e$ and $\langle w \rangle_{\dots}$.
\end{itemize} 

\noindent
There do however exist sets of Scott open modalities for which the characterisation of $\prebas$ via (\ref{equation:prebas_reduce}) does not hold.
For example, for $\Sigma = \{\textit{raise}: \alpha^0 \rightarrow \alpha\}$ and $\Obser = \{o\}$ where $\denote{o} = \{\ast, \textit{raise}\}$. 

We next define the relation $\modbel$ on double trees $TT\UnitType$.

\begin{definition}
	We define the preorder $\modbel$ on $TT\UnitType$ by: \,for any two double trees $r,r' \in TT\UnitType$,
	$$r \modbel\, r' \quad:\iff\quad \forall o \in \Obser, \forall \Phi \in \Treef, ~ r \in o(\denote{\Phi}) \Rightarrow r' \in o(\denote{\Phi}).$$
\end{definition} 

We give some characterisations of $\modbel$. These use the notion of \emph{right-set}, that is; for any relation $\mathcal{R} \subseteq X \times Y$ and subset $A \subseteq X$, we write $\Rise{\mathcal{R}}{A} := \{y \in Y \, | \, \exists x \in A, x \,\mathcal{R}\, y \}$ for the right-set of $A$ under the relation $\mathcal{R}$.

\begin{lemma}\label{lem:decomp_equ}
	If $\Obser$ is a set of upwards-closed modalities, then for all $r, r' \in \Dtrees$, the following are equivalent:
	\begin{enumerate}
		\item $r \modbel\, r'$.
		\item $\forall A \subseteq T\UnitType, ~ \tis{r}{A} \prebas \tis{r'}{(\Rise{\prebas}{A})}$.
		\item $\forall o \in \Obser, \forall A \subseteq T\UnitType, ~ r \in o(A) \Rightarrow r' \in o(\Rise{\prebas}{A})$.
	\end{enumerate}
\end{lemma}
\begin{proof}
	(1) $\Rightarrow$ (2) ~ Let $\tis{r}{A} \in o(\{\ast\})$, then since $A \subseteq (\Rise{\prebas}{A})$ and $o$ is upwards closed, $\tis{r}{(\Rise{\prebas}{A})} \in o(\{\ast\})$, which means $r \in o(\Rise{\prebas}{A})$.
	Let $\Phi_A := \bigvee_{t \in A} ((\bigwedge_{o \in \Obser, t \in o(\{\ast\})} o(\top)) \wedge (\bigwedge_{o \in \Obser, t \in o(\emptyset)} o(\bot)))$, then $\denote{\Phi_A} = (\Rise{\prebas}{A})$ as $\Phi_A$ perfectly replicates the condition of membership in $\Rise{\prebas}{A}$.
	So by (1) it holds that $r' \in o(\Rise{\prebas}{A})$, hence $\tis{r'}{(\Rise{\prebas}{A})} \in o(\{\ast\})$.
	If $\tis{r}{A} \in o(\emptyset)$, then $r \in o(\emptyset)$ hence $r' \in o(\emptyset)$ (since $(\bigvee \emptyset) \in \Treef$ and $\denote{\bigvee \emptyset} = \emptyset$), so $\tis{r'}{(\Rise{\prebas}{A})} \in o(\emptyset)$.
	We can conclude that $\tis{r}{A} \prebas \tis{r'}{(\Rise{\prebas}{A})}$.
	
	(2) $\Rightarrow$ (3) ~ If $\tis{r}{A} \prebas \tis{r'}{(\Rise{\prebas}{A})}$, then for any $o \in \Obser$ it holds that $\tis{r}{A} \in o(\{\ast\}) \Rightarrow \tis{r'}{\Rise{\prebas}{A}} \in o(\{\ast\})$, which is identical to the statement $r \in o(A) \Rightarrow r' \in o(\Rise{\prebas}{A})$.
	
	(3) $\Rightarrow$ (1) ~ If $r \in o(\denote{\Phi})$ with $\Phi \in \Treef$, then by (3) it holds that $r' \in o(\Rise{\prebas}{\denote{\Phi}})$.
	By the definition of $\prebas$ it holds that $(\Rise{\prebas}{\denote{\Phi}}) \subseteq \denote{\Phi}$, and since $\prebas$ is reflexive $\denote{\Phi} \subseteq (\Rise{\prebas}{\denote{\Phi}})$.
	Hence $(\Rise{\prebas}{\denote{\Phi}}) = \denote{\Phi}$ and we conclude that $r' \in o(\denote{\Phi})$.
\end{proof}

We can now finally define the promised notion of decomposability. 

\begin{definition}[Decomposability]
	\label{definition:decomposable}
	We say that  $\Obser$ is \textit{decomposable} if, for all double trees $r,r' \in \Dtrees$,
	$r \modbel\, r'$ implies $\mu r \prebas \mu r'$.
\end{definition}

Since decomposability is such a technical notion, we give 
two lemmas providing alternative characterisations of it.
The first gives a reformulation that is immediate in the case of our examples, where the statement $t \prebas t'$ can be simplified via (\ref{equation:prebas_reduce}). The general case, however, requires a rather technical proof.

\begin{lemma}\label{lem:decomp_reduct}
	A set of upwards-closed modalities $\Obser$ is decomposable if and only if for all $r,r' \in \Dtrees$, if $r \modbel\, r'$, then ~$\forall o \in \Obser, ~ \mu r \in o(\{\ast\}) \Rightarrow \mu r' \in o(\{\ast\})$.
\end{lemma}
\begin{proof}
	We use the equivalences from Proposition~\ref{prop:base_relation}.
	
	($\Rightarrow$) ~
	The result follows by observing that $\mu r \prebas \mu r'$ implies $\forall o \in \Obser, ~ \mu r \in o(\{\ast\}) \Rightarrow {\mu r' \in o(\{\ast\})}$.
	
	($\Leftarrow$) ~
	Assume: ~
	(I) $\forall r,r' \in \Dtrees, ~ r \modbel\, r' ~ \Rightarrow ~ (\forall o \in \Obser, ~ \mu r \in o(\{\ast\}) \Rightarrow \mu r' \in o(\{\ast\}))$.
	
	\noindent
	Take some $r, r' \in \Dtrees$, and  
	suppose that $r \modbel\, r'$, hence with Lemma \ref{lem:decomp_equ} it holds that: 
	
	(II) $\forall o \in \Obser, A \subseteq T\UnitType, ~ r \in o(A) \Rightarrow r' \in o(\Rise{\prebas}{A})$.
	
	\noindent
	We need to prove that $\mu r \prebas \mu r'$.
	By (I) we derive that $\forall o \in \Obser, ~ \mu r \in o(\{\ast\}) \Rightarrow \mu r' \in o(\{\ast\})$.
	To prove $\mu r \prebas \mu r'$ we need only prove $\forall o \in \Obser, ~ \mu r \in o(\emptyset) \Rightarrow \mu r' \in o(\emptyset)$.
	
	Assume $\mu r \in o'(\emptyset)$ for $o' \in \Obser$, then $\mu(r[t \mapsto \tis{t}{\emptyset}]) = \tis{(\mu r)}{\emptyset} \in o'(\{\ast\})$.
	We prove that ${r[t \mapsto \tis{t}{\emptyset}]} \modbel~ r'[t' \mapsto \tis{t'}{\emptyset}]$.
	Suppose for some $o \in \Obser$ and $A \subseteq T\UnitType$, it holds that ${r[t \mapsto \tis{t}{\emptyset}]} \in o(A)$.
	Let $B := \{t \in T\UnitType \mid \tis{t}{\emptyset} \in A\}$, then $r \in o(B)$.
	By (II) it holds that $r' \in o(\Rise{\prebas}{B})$.
	
	For $t' \in \,(\Rise{\prebas}{B})$, there is a $t \in B$ such that $t \prebas t'$.
	Since $t \in B$, $\tis{t}{\emptyset} \in A$ and hence $\tis{t'}{\emptyset} \in (\Rise{\prebas}{A})$.
	So $(\Rise{\prebas}{B}) \subseteq \{t' \in T\UnitType \mid \tis{t'}{\emptyset} \in (\Rise{\prebas}{A})\}$, and we can use upwards closure of $o'$ to derive ${r'[t' \mapsto \tis{t'}{\emptyset}]} \in o(\Rise{\prebas}{A})$.
	Hence by Lemma \ref{lem:decomp_equ}, $r[t \mapsto \tis{t}{\emptyset}] \modbel\, r'[t' \mapsto \tis{t'}{\emptyset}]$.
	
	We can apply (I) to derive $\mu(r[t \mapsto \tis{t}{\emptyset}]) \prebas \mu(r'[t' \mapsto \tis{t'}{\emptyset}])$.
	So since $\mu(r[t \mapsto \tis{t}{\emptyset}]) = \tis{(\mu r)}{\emptyset} \in o'(\{\ast\})$, it holds that and $\mu(r'[t' \mapsto \tis{t'}{\emptyset}]) \in o'(\{\ast\})$ and hence $\mu r' \in o'(\emptyset)$. 
	We conclude that $\mu r \prebas \mu r'$, so we are finished.
\end{proof}

 \noindent
The second reformulation of decomposability shows that 
it is equivalent to being able to `decompose' statements of the 
form $\mu r \in o(\{\ast\})$ into a collection of modal properties of $r$.

\begin{lemma}\label{lem:true_decomp}
	A set of upwards-closed modalities $\Obser$ is decomposable if and only if for any $r \in \Dtrees$ and $o \in \Obser$ such that $\mu r \in o(\{\ast\})$, there is a collection of pairs $\{(o_i,\Phi_i)\}_{i \in I}$, with each $o_i \in \Obser$ and $\Phi_i \in \Treef$, satisfying the following two properties:
	\begin{enumerate}
		\item $\forall i \in I, ~ r \in o_i(\denote{\Phi_i})$.
		\item $\forall r' \in \Dtrees, ~ (\forall i \in I, ~ r' \in o_i(\denote{\Phi_i})) ~ \Rightarrow ~ \mu r' \in o(\{\ast\})$.
	\end{enumerate}
\end{lemma}

\begin{proof}
	We use the equivalent statement for decomposability established in Lemma \ref{lem:decomp_reduct}.
	
	($\Rightarrow$) Assume decomposability, and that for some $r \in \Dtrees$ and $o \in \Obser$, it holds that $\mu r \in o(\{\ast\})$.
	For any $A \subseteq TT\UnitType$ such that $r \in A$, choose a pair $(o'_A,\Phi'_A)$ where $o'_A \in \Obser$ and $\Phi'_A \in \Treef$ such that $o'_A(\denote{\Phi'_A}) = A$, if such a pair exists. 
	All these chosen pairs together form our desired collection.
	Since for any such $A$, $r \in A = o'_A(\denote{\Phi'_A})$, this collection satisfies condition (1).
	For any $r' \in  \Dtrees$, if $r' \in o'_A(\denote{\Phi'_A}) = A$ for any pair in the collection, then for any $r \in A \subseteq TT\UnitType$ for which there is a pair $o'', \Phi''$ such that $o''(\denote{\Phi''}) = A$, it holds that $r' \in A$. So $r \modbel\, r'$.
	Hence by decomposability, $\forall o'' \in \Obser, \mu r \in o''(\{\ast\}) \Rightarrow \mu r' \in o''(\{\ast\})$, so in particular $r' \in o(\{\ast\})$, and thus condition (2) holds.
	
	($\Leftarrow$) Assume the statement given in the lemma, we need to prove decomposability.
	For some  $r,r' \in \Dtrees$, suppose that $r \modbel\, r'$.
	Now let $o \in \Obser$ such that $\mu r \in o(\{\ast\})$,
	then there is a collection of pairs $\{(o_i,\Phi_i)\}_{i \in I}$ satisfying the properties given above.
	So by property (1), $\forall i \in I, ~ r \in o_i(\denote{\Phi_i})$, and since $r \modbel\, r'$ it holds that $r' \in o_i(\denote{\Phi_i})$.
	By property (2) we conclude that $\mu r' \in o(\{\ast\})$ which is what we needed to prove.
\end{proof}

Below, we shall show that all our running examples satisfy the decomposability property.
In all cases we can do this by establishing a stronger property that  is easier to verify.
This strengthened notion of decomposability is obtained by simplifying the property given in Lemma \ref{lem:true_decomp}. 

\begin{definition}[Strong decomposability]
	\label{def:strong}
	We say that $\Obser$ is \emph{strongly decomposable} if, for every $r  \in \Dtrees$ and $o \in \Obser$ 
	for which $\mu r \in o(\{\ast\}) $, there exists a collection  $\{(o_i,o_i')\}_{i \in  I}$ of pairs of modalities  such that:
	\begin{enumerate}
		\item $\forall i \in I, ~ r \in o_i(o_i'(\{\ast\})) \,$; and
		\item for every  $r' \in \Dtrees$, if for all $i \in I$, 
		$r' \in o_i(o_i'(\{\ast\}))$ then $\mu r' \in o(\{\ast\})\,$.
	\end{enumerate}
\end{definition}

\begin{proposition}
	If $\Obser$ is strongly decomposable set of upwards closed modalities, 
	then $\Obser$ is decomposable.
\end{proposition}

\begin{proof}
	Using Lemma \ref{lem:true_decomp}, this result is a simple consequence of the fact that for any $o_i \in \Obser$, $o_i(\top) \in \Treef$ and $\denote{o_i(\top)} = o_i(\{\ast\})$.
\end{proof}

Not all decomposable sets of Scott open modalities are strongly decomposable. Take for instance $\Obser := \{\downarrow, \widehat{\Box}\}$ for the signature $\Sigma = \{\textit{or}: \alpha^2 \rightarrow \alpha\}$, where $\denote{{\downarrow}} := \{\, * \, \}$ 
and $\denote{\widehat{\Box}} := (\denote{\Box} - \{\, * \, \})$. 

We end this section by again looking at our running examples, and showing, in each case, that the identified collection $\Obser$ only has Scott-open (hence upwards closed) modalities and is strongly decomposable (hence decomposable). For any of the examples, upwards closure is easily established, so we will not show it here.

\setcounter{myexampleD}{-1}
\begin{myexampleD}[Pure functional computation]
We have  $\Obser = \{{\downarrow}\}$ and $
\denote{{\downarrow}} = \{\, * \, \}$. 
Scott openness holds since if $\sqcup_i t_i = *$ then for some $i$ we must already have $t_i = *$.
It is strongly decomposable since: $\mu r \in  \denote{\downarrow} \Leftrightarrow r \in \,\downarrow\!(\downarrow\!(\{\ast\}))$, which means $r$ returns a tree $t$ which is a leaf $*$.
\end{myexampleD}

\begin{myexampleD}[Error] We have  $\Obser = \{{\downarrow}\} \cup \{\mathsf{E}_e \mid e \in E\}$ and 
$ \denote{{\mathsf{E}_e}} = \{\, \textit{raise}_e \, \}$. Scott-openness holds for both modalities for the same reason as in the previous example, and strong decomposability holds since:
$\mu r \in  \denote{\downarrow} \Leftrightarrow r \in \,\downarrow\!(\downarrow\!(\{\ast\}))$ like in the previous example, and:

\medskip
$\mu r \in \denote{\mathsf{E}_e} \quad \Leftrightarrow \quad r \in \mathsf{E}_e(\mathsf{E}_e(\{\ast\})) ~ \vee ~ r \in \,\downarrow\!(\mathsf{E}_e(\{\ast\}))$.
\medskip

\noindent Which means $r$ raises an error, or returns a tree that raises an error.
\end{myexampleD}

\begin{myexampleD}[Nondeterminism]
\label{example:nonD}
We have  $\Obser = \{\Diamond\}$ for angelic nondeterminism and $\Obser = \{\Box\}$ for demonic nondeterminism.
The Scott-openness of $\denote{{\Diamond}}  = \{ t \mid \text{$t$ has some $*$ leaf} \}$ holds because if $\sqcup_i t_i$ has a $*$ leaf, then that leaf must already be contained in $t_i$ for some $i$. Similarly, if  $\sqcup_i t_i \in \denote{{\Box}}$ then, because $\denote{{\Box}} = \{ t \mid \text{$t$ has finite height and every leaf is a  $*$}\}$, the tree $\sqcup_i t_i$ has finitely many leaves and all must be contained in $t_i$ for some $i$. Hence $t_i \in \denote{{\Box}}$ for that $i$. Strong decomposability holds because:

\medskip
$\mu r \in  \denote{\Diamond} ~ \Leftrightarrow ~ r \in \Diamond(\Diamond(\{\ast\}))
\quad \text{and} \quad
\mu r \in  \denote{\Box} ~ \Leftrightarrow ~ r \in \Box(\Box(\{\ast\})
\enspace .$
\medskip

\noindent
The former states that  $r$ has as a leaf a tree $t$, which itself has a leaf $*$.
The latter states that  $r$ is finite and all leaves are finite trees $t$ that  have only $*$ leaves. We can conclude that $\{\Diamond\}$ and $\{\Box\}$ are both strongly decomposable sets of Scott open modalities. Moreover, it is obvious that their union $\{\Diamond,\Box\}$, for neutral nondeterminism, is a strongly decomposable set of Scott open modalities too.
\end{myexampleD}

\begin{myexampleD}[Probabilistic choice] $\Obser = \{ \mathsf{P}_{>q} \mid q \in \mathbb{Q},\, 0 \leq q < 1\}$. For the Scott-openness of $\denote{\mathsf{P}_{>q}} = \{ t \mid 
\mathbf{P}(\,\text{$t$ terminates with a $*$ leaf}\,) > q\}$, note that $\mathbf{P}(\,\text{$\sqcup_i t_i$ terminates with a $*$ leaf}\,)$ is determined by some countable sum over the leaves of $t_i$. If this sum is greater than a rational $q$, then some finite approximation of the sum must already be above $q$. The finite sum is over finitely many leaves
from $\sqcup_i t_i$, all of which will be present in $t_i$ for some $i$. Hence $t_i \in \denote{\mathsf{P}_{>q}}$.

For strong decomposability, suppose $\mu r \in \mathsf{P}_{>q}(\{\ast\})\,$. However,
$\mathbf{P}(\,\text{$\mu r$ terminates with a $*$ leaf}\,)$ equals the integral 
of the (monotone decreasing) function $f_r(x) = \textit{sup}\{y \in [0,1] \mid r \in \mathsf{P}_{>y}(\mathsf{P}_{>x}(\{\ast\}))\}$
from $[0,1]$ to $[0,1]$. Informally, 
$f_r(x)$ is the probability that $r$ returns a tree in the set $\denote{\mathsf{P}_{>x}}$. 
Since $\int_0^1 f_r(x)>q$, we can find a monotone decreasing rational step function below $f_r$
whose integral is also greater than $q$. This rational step function can be specified by rational numbers
$0 < a_0 < \dots < a_n < 1$ and $1 > b_0 > \dots > b_n > 0$ satisfying $b_i < f_r(a_i)$ and 
$a_0 b_0 + \sum_{i=1}^n (a_i-a_{i-1})b_i)  > q\,$. Then the collection of pairs of modalities
$\{(\mathsf{P}_{>a_i}, \mathsf{P}_{>b_i})\}_{i=1}^n$ satisfies the properties required by strong decomposability.
\end{myexampleD}

\begin{myexampleD}[Global store] We have 
$\Obser = \{(s \rightarrowtail s') \mid s,s' \in \textit{State}\}$. For the Scott-openness of 
$\denote{(s \rightarrowtail s')}  = \{ t \mid  \textit{exec}(t,s) = (*,s')\}$,
note that if $\textit{exec}(\sqcup_i t_i,s) = (*,s')$, there is a single finite branch of $t$ that follows the path the recursive function $\textit{exec}$ took. This branch must already be contained in $t_i$ for some $i$. We also have strong decomposability since:

\medskip
$\mu r \in \denote{s \rightarrowtail s'} \quad \Leftrightarrow \quad \exists s'' \in \textit{State}, \,r \in (s \rightarrowtail s'')((s'' \rightarrowtail s')(\{\ast\}))$.
\medskip

\noindent
Which just means that for some $s''$, it holds that $\textit{exec}(r,s) = (t,s'')$ and $\textit{exec}(t,s'') = (*,s')$.
\end{myexampleD}

\begin{myexampleD}[Input/output] We have 
$\Obser = \{{\langle w \rangle\!\!\downarrow},\,
\langle w\rangle_{\!\dots} \mid \text{$w$ an i/o-trace}\}$. For the Scott-openness of 
$\denote{\langle w \rangle\!\!\downarrow \,} =  \{ t \mid  t \models {\langle w \rangle\!\!\downarrow}\, \{*\}\, \}$, note that the i/o-trace $\langle w \rangle\!\!\downarrow$ is given by some finite branch, which if in $\sqcup_i t_i$ must be in $t_i$ for some $i$. The Scott-openness of $\denote{\langle w\rangle_{\!\dots} \,} = \{ t \mid  t \models   \langle w\rangle_{\!\dots}\, \}$ holds for similar reasons. We have strong decomposability because of the implications:

$\mu r \in \denote{\langle w \rangle\!\!\downarrow} \quad \Leftrightarrow \quad \exists v,u$ i/o-traces$, ~ vu = w ~ \wedge ~ r \in \langle v \rangle\!\!\downarrow\!(\langle u \rangle\!\!\downarrow\!(\{\ast\}))$.

\noindent Which means $r$ follows trace $v$ returning $t$, and $t$ follows trace $u$ returning $*$.

$\mu r \in \denote{\langle w \rangle_{\!\dots}} \; \Leftrightarrow \quad 
r \in \langle w \rangle_{\!\dots}(\langle \rangle_{\!\dots}(\{\ast\})) ~ \vee	~ (\exists v,u, ~ vu = w ~ \wedge ~ r \in \langle v \rangle\!\!\downarrow(\langle u \rangle_{\!\dots}(\{\ast\})))$.

\noindent Which means either $r$ follows trace $w$ immediately, or it follows $v$ returning a tree that follows $u$.
\end{myexampleD}

In Example \ref{example:nonD} above, we used the obvious fact that the property of being strongly decomposable is preserved under union. Similarly, if all modalities are upwards-closed, then by Lemma \ref{lem:true_decomp},  the property of being decomposable is preserved under union.

\begin{lemma}
	If $\,\Obser_0$ and $\Obser_1$ are decomposable sets of upwards closed modalities for the same signature $\Sigma$, then $\Obser_0 \cup \Obser_1$ is decomposable.
\end{lemma}

In order to prove Theorem \ref{Main_Theorem}, we need to relate the behavioural equivalence $\equiv_{\LogNeg}$ with a modality-defined notion of applicative bisimilarity.
As motivated in the introduction, the connection between these two equivalences is interesting on its own. 
We will spend the next two sections defining this notion of bisimilarity, establishing the connection with behavioural equivalence and proving that the two relations are compatible.
From Section \ref{section:contextual} and onwards, we will return to only studying the behavioural equivalence.

%% file: 5_Similarity.tex
\section{Applicative \texorpdfstring{$\Obser$}{TEXT}-(bi)similarity}\label{section:similarity}

In this section, we define this notion of applicative bisimilarity based on a set of modalities $\Obser$, and establish that it is equal to the behavioural equivalence $\equiv_{\LogNeg}$.
Central to such a definition lies the concept of a \emph{relator}~\cite{Thijs,Levy11}, which we use to lift a relation on value terms to a relation on computation terms. 
Later on in this section, we will prove that this relator, if based on a decomposable set of Scott open modalities $\Obser$, satisfies the right properties in order to prove the compatibility of bisimilarity in Section~\ref{section:Howe}.

Let $\Obser$ be some set of modalities for an effect signature $\Sigma$.
We define a relator based on $\Obser$.

\begin{definition}[$\Obser$-relator] \label{definition:relator-action}
The $\Obser$-\emph{relator}  lifts for each two sets $X$ and $Y$ a relation $\mathcal{R} \subseteq X \times Y$ to a relation $\Obser(\mathcal{R}) \subseteq TX \times TY$, such that
$$t \, \Obser(\mathcal{R}) \, t' \quad \iff \quad \forall A \subseteq X, \forall o \in \Obser, \, t \in o(A) \Rightarrow t' \in o(\Rise{\mathcal{R}}{A}) \enspace .$$
\end{definition}

\noindent
Remember that $\Rise{\mathcal{R}}{A} := \{y \in Y \mid \exists x \in A, x \,\mathcal{R}\, y\}$, and $t \in o(A)$ means $t[\in A] \in \denote{o}$.
By Proposition \ref{prop:base_relation} it holds that $\Obser(\textit{id}_{\UnitType}) = (\prebas)$, and by Lemma \ref{lem:decomp_equ} we know that $\Obser(\prebas) = \, \modbel$. 

We will use this relator to define an $\Obser$-tailored  variant of Abramsky's applicative similarity~\cite{Abramsky91}. First however, we describe the specific $\Obser$-relators that arise for each of our running examples.

\setcounter{myexampleE}{-1}
\begin{myexampleE}[Pure functional computation]
	The statement $t \,\{{\downarrow}\}(\mathcal{R})\, r$ holds if and only if:
	
	\noindent
	When $t$ evaluates to $x \in X$ then $r$ evaluates to a $y \in Y$ such that $x \,\mathcal{R}\, y$.
\end{myexampleE}

\begin{myexampleE}[Error] The statement $t ~ (\{{\downarrow}\} \cup \{\mathsf{E}_e \mid e \in E\})(\mathcal{R}) ~ r$ holds precisely when the following two statements hold:
	\begin{enumerate}
	\item[-] When $t$ evaluates to $x \in X$ then $r$ evaluates to a $y \in Y$ such that $x \,\mathcal{R}\, y$.
	
	\item[-] When $t$ raises an error $e \in E$ (meaning it is node $\textit{raise}_e$), then $r$ raises the same error $e$.
	\end{enumerate}
\end{myexampleE}

\begin{myexampleE}[Nondeterminism] $t ~ \{\Diamond, \, \Box\}(\mathcal{R}) ~ r$ holds precisely when the following statements hold:
	\begin{enumerate}
	\item[-] If $x \in X$ is a leaf of $t$, then $r$ has a leaf $y \in Y$ such that $x \,\mathcal{R}\, y$.
	
	\item[-] If $t$ is finite and has no $\bot$ leaf, then $r$ is finite and has no $\bot$-leaf.
	
	\item[-] If $t$ is finite, has no $\bot$ leaf, and $y \in Y$ is a leaf of $r$, then $t$ has a leaf $x \in X$ such that $x \,\mathcal{R}\, y$.
	\end{enumerate}
\end{myexampleE}

\begin{myexampleE}[Probabilistic choice] It holds that $t ~ \{ \mathsf{P}_{>q} \mid q \in \mathbb{Q},\, 0 \leq q < 1\}(\mathcal{R}) ~ r$ if and only if:
	
	\noindent
	For any $A \subseteq X$, the probability of $t$ terminating with an element of $A$ is at most the probability of $r$ terminating with a $y$ related to some element of $A$ (meaning there is an $x \in A$ such that $x \, \mathcal{R} \, y$).
\end{myexampleE}
\begin{myexampleE}[Global store] The statement $t ~ \{(s \rightarrowtail s') \mid s,s' \in \textit{State}\}(\mathcal{R}) ~ r$ holds if:
	
	\noindent
	For any $s \in \textit{State}$, if $\textit{exec}(t,s) = (x,s')$ then there is a $y \in Y$ such that $\textit{exec}(r,s) = (y,s')$ and $x \,\mathcal{R}\, y$.
\end{myexampleE}

\begin{myexampleE}[Input/output] $t ~ \{{\langle w \rangle\!\!\downarrow},\,
	\langle w\rangle_{\!\dots} \mid w \text{ an i/o-trace}\}(\mathcal{R}) ~ r$ holds precisely when the following two statements hold:
	\begin{enumerate}
	\item[-] If $t$ has $w$ as initial execution trace, then $r$ has $w$ as an initial execution trace.
	
	\item[-] If $t$ has execution trace $w$ terminating with $x \in X$, then $r$ has trace $w$ terminating with some $y \in Y$ such that $x \,\mathcal{R}\, y$.
	\end{enumerate}
\end{myexampleE}

In each of the examples above,  we obtain an $\Obser$-relator that acts in line with expectations, 
based on the explicit definitions of relators for different effects in the paper by Dal Lago \emph{et al.}~\cite{Relational}. In our case, these relators have been 
obtained in a uniform way from the corresponding sets of modalities  $\Obser$. 

A direct comparison with definitions in the paper~\cite{Relational} is slightly involved, because our relators act on the set of effect trees, given by applying the effect-tree monad $T$, whereas those in 
the paper~\cite{Relational} are defined using effect-dependent monads $M$. For example, for probabilistic choice the distribution monad is used.
Despite this difference, there is still a tight relationship between the two approaches. 
The tree monad $T$ is the free-continuous-algebra monad over the effect signature $\Sigma$. There is thus an induced
$\rho_X : TX \to MX$, mapping trees to elements of the effect-specific monad $MX$. It then holds, for each effect example, 
that
$$t \,\Obser(\mathcal{R})\, r \iff \rho(t) \,\Gamma(\mathcal{R}) \, \rho(r) \,,$$
where $\Gamma$ is the relevant relator from the paper~\cite{Relational} and $\Obser$ is the corresponding set of modalities from this paper.

Following the paper~\cite{Relational}, we use the relation-lifting operation of Definition~\ref{definition:relator-action} to define notions
of applicative similarity and bisimilarity. We assume that all modalities of $\Obser$ are upwards closed.

\begin{definition}[Similarity]\label{definition:sim}
	An \emph{applicative $\Obser$-simulation} is a pair of relations $\mathcal{R}^v_{\TypeOne}$ and $\mathcal{R}^c_{\TypeOne}$ for each type $\TypeOne$, where $\mathcal{R}^v_{\TypeOne} \subseteq \textit{Val}(\TypeOne)^2$ and $\mathcal{R}^c_{\TypeOne} \subseteq \textit{Com}(\TypeOne)^2$, such that:
	\begin{enumerate}
		\item $V \,\mathcal{R}^v_{\NatType}\, W ~ \implies ~ (V = W)$ .
		\item $M \,\mathcal{R}^c_{\TypeOne}\, N~ \implies ~|M|$ $\Obser(\mathcal{R}^v_{\TypeOne})$ $|N|$ .
		\item $V \,\mathcal{R}^v_{\TypeTwo \rightarrow \TypeOne}\, W ~ \implies ~\forall U \in \textit{Val}(\TypeTwo), \; VU \, \mathcal{R}^c_{\TypeOne} \, WU$ .
	\end{enumerate}
	\emph{Applicative $\Obser$-similarity} is the largest applicative $\Obser$-simulation, which is equal to the union of all applicative $\Obser$-simulations.
\end{definition}

\begin{definition}[Bisimilarity]\label{definition:bisim}
	An \emph{applicative $\Obser$-bisimulation} is a symmetric $\Obser$-simulation. The relation of \emph{$\Obser$-bisimilarity} is the largest applicative $\Obser$-bisimulation.
\end{definition}

\noindent
Applicative $\Obser$-similarity and $\Obser$-bisimilarity may not exist for all sets of modalities $\Obser$. But, if all modalities of $\Obser$ are upwards closed (as asserted above), their existence is guaranteed. 
Though this fact can be proven directly from relational properties established in Lemma \ref{Rel_prop1}, we can also see it as a consequence of Theorem \ref{log_is_sim}.

\begin{lemma}
	Applicative $\Obser$-bisimilarity is identical to the relation of applicative $(\Obser \cap \Obser^{\mathsf{op}})$-similarity, where $\Obser^{\mathsf{op}}(\mathcal{R}) := (\Obser(\mathcal{R}^{\mathsf{op}}))^{\mathsf{op}}$, so $t \,(\Obser \cap \Obser^{\mathsf{op}})(\mathcal{R})\, r \Leftrightarrow t \,\Obser (\mathcal{R})\, r \wedge r \,\Obser (\mathcal{R}^{\mathsf{op}})\, t$.
\end{lemma}

\begin{proof}
	Let $\mathcal{R}$ be the $\Obser$-bisimilarity, then by symmetry we have $\mathcal{R}^{\mathsf{op}} = \mathcal{R}$. So if $M \,\mathcal{R}\, N$ we have $N \,\mathcal{R}\, M$, and by the simulation rules we derive $|M| \,\Obser (\mathcal{R})\, |N|$ and $|N| \,\Obser(\mathcal{R})\, |M|$, so $\mathcal{R}$ is an $\Obser \cap \Obser^{\mathsf{op}}$-simulation.
	
	Let $\mathcal{R}$ be the $\Obser \cap \Obser^{\mathsf{op}}$-similarity. If $M \,\mathcal{R}^{\mathsf{op}}\, N$ then $|N| \,(\Obser \cap \Obser^{\mathsf{op}})(\mathcal{R})\, |M|$ so $|N| \,\Obser(\mathcal{R})\, |M| \wedge |M| \,\Obser(\mathcal{R}^{\mathsf{op}})\, |N|$ which results in $|M| \,(\Obser \cap \Obser^{\mathsf{op}})(\mathcal{R}^{\mathsf{op}})\, |N|$. Verifying the other simulation conditions as well, we can conclude that the symmetric closure $\mathcal{R} \cup \mathcal{R}^{\mathsf{op}}$ is an $\Obser \cap \Obser^{\mathsf{op}}$-simulation. So $\mathcal{R} = \mathcal{R} \cup \mathcal{R}^{\mathsf{op}}$ and hence $\mathcal{R}$ is a symmetric $\Obser$-simulation.
\end{proof}

For brevity, we will leave out the word ``applicative'' from here on out.
The key result now is that the maximal relation, the $\Obser$-similarity is, given our assertion that all modalities are upwards closed, the same object as our logical preorder. We first give a short Lemma.

\begin{lemma}[Characteristic formulas]\label{pre_form}
	For any fragment $\mathcal{L}$ of $\,\mathcal{V}$ closed under countable conjunction, it holds that for each value $V$ there is a formula $\chi^\mathcal{L}_V \in \mathcal{L}$ such that $W \models_{\mathcal{L}} \chi^\mathcal{L}_V \Leftrightarrow V \sqsubseteq_{\mathcal{L}} W$. Similarly, for every computation $M$, there is a computation formula 
	$\Chi^\mathcal{L}_M \in \mathcal{L}$ such that $N \models_{\mathcal{L}} \Chi^\mathcal{L}_M \Leftrightarrow M\sqsubseteq_{\mathcal{L}} N$. 
\end{lemma}
\begin{proof}
	We prove the statement for values $V$.
	For each $U$ such that $(V \not\sqsubseteq_{\mathcal{L}} U)$, choose a formula $\phi^U \in \mathcal{L}$ such that $V \models_{\mathcal{L}} \phi^U$ and $(U \not\models \phi^U)$. Then if we define $\chi_V := \bigwedge_{\{U \,\mid\, V \not\sqsubseteq_{\mathcal{L}} U\}} (U \mapsto \phi^U)$  it holds that $V \not\sqsubseteq_{\mathcal{L}} U \Leftrightarrow U \not\models \chi_V$, which is what we want.
\end{proof}

\noindent
As in the proof above, we usually omit the superscript $\mathcal{L}$ when clear from the context. Note that when 
$\mathcal{L}$ is the full logic $\mathcal{V}$, it holds that 
\begin{equation}
\label{equation:characteristic}
N \models_{\mathcal{V}} \Chi^\mathcal{L}_M ~ \Leftrightarrow ~ M\equiv_{\mathcal{V}} N \enspace ,
\end{equation}
and similarly for value formulas.

\begin{theoremm}[A]\label{log_is_sim}
	For any family of upwards closed modalities $\Obser$, we have that the logical preorder $\sqsubseteq_{\LogPos}$ is identical to $\Obser$-similarity.
\end{theoremm}
\begin{proof}
	We write $\sqsubseteq$ instead of $\sqsubseteq_{\LogPos}$ to make room for other annotations.
	We first prove that our logical preorder $\sqsubseteq$ is an  $\Obser$-simulation by induction on types.
	\begin{enumerate}
		\item Values of $\NatType$. If $\overline{n} \sqsubseteq^v_{\NatType} \overline{m}$, then since $\overline{n} \models \{n\}$ we have that $\overline{m} \models \{n\}$, hence $m=n$.
	
		\item Computations of $\TypeOne$. Assume $M \sqsubseteq^c_{\TypeOne} N$, we prove that $|M| \,\Obser(\sqsubseteq^v_{\TypeOne})\, |N|$. Take $A \subseteq \textit{Val}(\TypeOne)$ and $o \in \Obser$ such that $|M| \in o(A)$. Taking the following formula $\phi_A := \bigvee_{a \in A} \chi_a$ (where $\chi_a$ as in Lemma \ref{pre_form}), then $b \models \phi_A \Leftrightarrow \exists a \in A, a \sqsubseteq^v_{\TypeOne} b \Leftrightarrow b \in \Rise{(\sqsubseteq^v_{\TypeOne})}{A}$ and $\forall a \in A, a \models \phi_A$. So $|M|[\models \phi_A] \geq \tis{|M|}{A}$, hence since $o$ is upwards closed, $|M|[\models \phi_A] \in \denote{o}$. 
		By $M \sqsubseteq^c_{\TypeOne} N$ we have $|N|[\models \phi_A] \in \denote{o}$ and hence $|N| \in o(\Rise{(\sqsubseteq^v_{\TypeOne})}{A})$. We conclude that $|M| \,\Obser(\sqsubseteq^v_{\TypeOne})\, |N|$.
	
		\item Function values of $\TypeTwo \rightarrow \TypeOne$. This step follows from Lemma \ref{funct_clas} and the induction hypothesis.
	\end{enumerate}
	
	\noindent
	We can conclude that $\sqsubseteq$ is an $\Obser$-simulation. Now take an arbitrary $\Obser$-simulation $\mathcal{R}$. We prove by induction on types that $\mathcal{R} \subseteq (\sqsubseteq)$.
	
	\begin{enumerate}
		\item Values of $\NatType$. If $V \,\mathcal{R}^v_{\NatType}\, W$ then $V=W$, hence by reflexivity we get $V \sqsubseteq^v_{\NatType} W$.
		
		\item Computations of $\TypeOne$. Assume $M \,\mathcal{R}^c_{\TypeOne}\, N$, we prove that $M \sqsubseteq^c_{\TypeOne} N$ using the characterisation from Lemma \ref{comp_clas}. Say for $o \in \Obser$ and $\phi \in \textit{VF}\,(\TypeOne)$ we have $M \models o \, \phi$. 
		Let $A_{\phi} := {\{a \in \textit{Val}(\TypeOne) \,\mid\, a \models \phi \}} \subseteq \textit{Val}(\TypeOne)$, then $|M| \in o(A_{\phi})$ hence by $M \,\mathcal{R}^c_{\TypeOne}\, N$ and the simulation property, we derive $|N| \in o(\Rise{(\mathcal{R}^v_{\TypeOne})}{A_{\phi}})$. 
		By the induction hypothesis on values of $\TypeOne$, we know that $\mathcal{R}^v_{\TypeOne} \subseteq (\sqsubseteq^v_{\TypeOne})$, hence  `$\exists a \in A_{\phi}, a \,\mathcal{R}^v_{\TypeOne}\, b$' implies $b \models \phi$. 
		We get that $|N|[\models \phi] \geq \tis{|N|}{\Rise{(\mathcal{R}^v_{\TypeOne})}{A_{\phi}}}$, so by upwards closure of $o$ we have $|N|[\models \phi] \in \denote{o}$ meaning $N \models o \, \phi$. We conclude that $M \sqsubseteq^c_{\TypeOne} N$.
		
		\item Function values of $\TypeTwo \rightarrow \TypeOne$. Assume $V \,\mathcal{R}^v_{\TypeTwo \rightarrow \TypeOne}\, W$. We prove $V \sqsubseteq^v_{\TypeTwo \rightarrow \TypeOne} W$ using the characterisation from Lemma \ref{funct_clas}. Assume $V \models (U \mapsto \Phi)$ where $U \in \textit{Val}(\TypeTwo)$ and $\Phi \in \textit{CF}\,(\TypeOne)$, so $VU \models \Phi$. By $V \,\mathcal{R}^v_{\TypeTwo \rightarrow \TypeOne}\, W$ we have $VU \; \mathcal{R}^c_{\TypeOne} \; WU$ and by the induction hypothesis we have $\mathcal{R}^c_{\TypeOne} \subseteq (\sqsubseteq^c_{\TypeOne})$, so $VU \sqsubseteq^c_{\TypeOne} WU$. Hence $WU \models \Phi$ meaning $W \models (U \mapsto \Phi)$. We can conclude that $V \sqsubseteq^v_{\TypeTwo \rightarrow \TypeOne} W$.
		
		\item Values of $\bf{1}$. If $V \,\mathcal{R}^v_{\bf{1}}\, W$ then $V=*=W$ hence $V \sqsubseteq^v_{\bf{1}} W$.
	\end{enumerate}
	In conclusion: any $\Obser$-simulation $\mathcal{R}$ is a subset of the $\Obser$-simulation $\sqsubseteq_{\LogPos}$. So $\sqsubseteq_{\LogPos}$ is $\Obser$-similarity.
\end{proof}

Adapting the proof slightly, we can find the connection between full behavioural equivalence and $\Obser$-bisimilarity.

\addtocounter{theoremm}{-1}
\begin{theoremm}[B]
	For any family of upwards closed modalities $\Obser$, we have that the logical equivalence $\equiv_{\LogNeg}$ is identical to  $\Obser$-bisimilarity.
\end{theoremm}
\begin{proof}
	Note first that $\equiv_{\LogNeg}$ is symmetric. 
	Secondly, note that since $\equiv_{\LogNeg} \,=\, \sqsubseteq_{\LogNeg}$ we know by Lemma \ref{pre_form}, that for any $V$, there is a formula $\chi_V$ such that $W \models \chi_V \Leftrightarrow V \equiv_{\LogNeg} W$. 
	Using these observations and the appropriate adjustments to Lemma \ref{funct_clas} and \ref{comp_clas}, the proof of this result is as in Theorem \ref{log_is_sim}(A), proving $\equiv_{\LogNeg}$ is a symmetric $\Obser$-simulation and contains any other symmetric $\Obser$-simulation.
\end{proof}

Given Theorem~\ref{log_is_sim}(b), it is straightforward to observe an expressive completeness result with respect to the notion of \emph{behavioural property} induced by $\Obser$-bisimilarity (as adumbrated in Section~\ref{section:introduction}). A subset $S \subseteq  \textit{Com}(\TypeOne)$  (respectively $S \subseteq \textit{Val}(\TypeOne)$)
may be called \emph{behavioural} if it respects bisimilarity; that is, in the case of computations, if $M \in S$ and 
$M$ is bisimilar to $M'$ then  $M' \in S$.
It follows from Theorem~\ref{log_is_sim}(b) that the behavioural subsets coincide with the logically definable subsets.
Namely, $S \subseteq  \textit{Com}(\TypeOne)$ is behavioural if and only if there exists a formula $\Phi \in \textit{CF}\,(\TypeOne)$ such that $S = \{M \mid M \models \Phi\}$;
and similarly, any $S \subseteq  \textit{Val}(\TypeOne)$ is behavioural if and only if there exists a formula $\phi \in \textit{VF}\,(\TypeOne)$ such that $S = \{V \mid V \models \phi\}\,$.
The proof is straightforward, given the observation that every behavioural $S$ is a union of equivalence classes under bisimilarity. Hence, for example, $\Phi$ may be defined as an infinite disjunction
\[
\Phi ~ := ~ \bigvee_{M \in S} \, \Chi_M \enspace ,
\]
using the characteristic formulas obtained in~(\ref{equation:characteristic}) above.

\subsection{Relator properties}\label{section:relprop}

In this subsection, 
we identify abstract properties of our relation lifting $\Obser(\mathcal{R})$, which will be used in our
application of Howe's method, in Section \ref{section:Howe}. The necessary properties were identified in the paper~\cite{Relational}. The contribution of this paper is  that all the required properties follow 
from our modality-based definition of $\Obser(\mathcal{R})$.

The first set of properties tell us that $\Obser(-)$ is indeed a relator in the sense of Levy~\cite{Levy11}.

\begin{lemma}\label{Rel_prop1}
	If the modalities from $\Obser$ are upwards closed, then $\Obser(-)$ is a relator, meaning that:
	\begin{enumerate}
		\item If \,$\mathcal{R} \subseteq X \times X$ is reflexive, then so is $\Obser(\mathcal{R})$.
		\item $\forall \mathcal{R} \subseteq X \times Y, \forall \mathcal{S} \subseteq Y \times Z, \,\, \Obser(\mathcal{R})\, \Obser(\mathcal{S}) \subseteq \Obser(\mathcal{R} \mathcal{S})$, where $\mathcal{R} \mathcal{S} \subseteq X \times Z$ is relation composition.
		\item $\forall \mathcal{R} \subseteq X \times Y, \forall \mathcal{S} \subseteq X \times Y, \,\, \mathcal{R} \subseteq \mathcal{S} \Rightarrow \Obser(\mathcal{R}) \subseteq \Obser(\mathcal{S})$.
		\item $\forall f: X \to Z, g: Y \to W, \mathcal{R} \subseteq Z \times W, \Obser((f \times g)^{-1}\mathcal{R}) = (Tf \times Tg)^{-1} \Obser(\mathcal{R})$
		
		where $(f \times g)^{-1}(\mathcal{R}) = \{(x,y) \in X \times Y \,\mid\, f(x) \,\mathcal{R}\, g(y)\}$
		\footnote{In general, given $h: X \to Y$, $A \subseteq X$, and $B \subseteq Y$, we write $h(A) := \{h(x) \mid x \in X\}$ and $h^{-1}(B) := \{x \in X \mid {h(x) \in B}\}$.}.
	\end{enumerate}
\end{lemma}

\begin{proof} We prove each property separately.
	\begin{enumerate}
		\item For any set $A \subseteq X$ we have $A \subseteq (\Rise{\mathcal{R}}{A})$ by reflexivity of $\mathcal{R}$. 
		So for any $t \in TX$, if $t \in o(A)$ then by upwards closure of $o$ we have $t \in o(\Rise{\mathcal{R}}{A})$. 
		We conclude that $t \,\Obser(\mathcal{R})\, t$.
		
		\item We can see this by observing that $\Rise{(\mathcal{R}\mathcal{S})}{A} = \{z \in Z \mid \exists x \in X, y \in Y, x \,\mathcal{R}\,y\,\mathcal{S}\,z\} = \Rise{\mathcal{S}}{(\Rise{\mathcal{R}}{A})}$, and with $t \,\Obser(\mathcal{R})\, r \,\Obser(\mathcal{S})\, l$ it holds that $t \in o(A) \Rightarrow r \in o(\Rise{\mathcal{R}}{A}) \Rightarrow l \in o(\Rise{\mathcal{S}}{(\Rise{\mathcal{R}}{A})})$.
		
		\item With $\mathcal{R} \subseteq \mathcal{S}$ it holds that $\Rise{\mathcal{R}}{A} \subseteq \Rise{\mathcal{S}}{A}$ for any set $A \subseteq X$. 
		Assume $t  \,\Obser(\mathcal{R})\, r$ and $t \in o(A)$, then $r \in o(\Rise{\mathcal{R}}{A})$. 
		Hence with $\tis{r}{\Rise{\mathcal{S}}{A}} \geq \tis{r}{\Rise{\mathcal{R}}{A}}$ and by upwards closure of $o$, $r \in o(\Rise{\mathcal{S}}{A})$.
		
		\item If $t \,\Obser((f \times g)^{-1}\mathcal{R})\, r$ then for all $A \subseteq X$ and $o \in \Obser$, $t \in o(A) \Rightarrow r \in o(\Rise{((f \times g)^{-1} \mathcal{R})}{A}) \Rightarrow r \in o(g^{-1}(\Rise{\mathcal{R}}{f(A)})$. 
		Assume that for some $B$ we have $Tf(t) \in o(B)$, then $t \in o(f^{-1}(B))$ so $r \in o(g^{-1}(\Rise{\mathcal{R}}{f(f^{-1}(B))}))$. 
		Considering that $f(f^{-1}(B)) \subseteq B$ and $o$ is upwards closed, we derive that $r \in o(g^{-1}(\Rise{\mathcal{R}}{B}))$ and hence $Tg(r) \in o(\Rise{\mathcal{R}}{B})$. 
		This is for all such $B$ and $o$, so we can conclude that $Tf(t) \,\Obser(\mathcal{R})\, Tg(r)$.
		
		Now assume $Tf(t) \,\Obser(\mathcal{R})\, Tg(r)$ and $t \in o(A)$, then since $Tf(t) \in o(f(A))$ we have $Tg(r) \in o(\Rise{\mathcal{R}}{f(A)})$ so $r \in o(g^{-1}(\Rise{\mathcal{R}}{f(A)}))$. 
		This is for all such $A$ and $o$, hence $t \,\Obser((f \times g)^{-1}\mathcal{R})\, r$.
	\end{enumerate}
\end{proof}

\noindent
The next property
together with the previous lemma establishes that $\Obser(-)$ is a \emph{monotone relator} in the sense of Thijs~\cite{Thijs}.

\begin{lemma}\label{Rel_prop1a}
	If the modalities from $\Obser$ are upwards closed, then $\Obser(-)$ is {monotone}, meaning for any $f: X \to Z$, $g: Y \to W$, $\mathcal{R} \subseteq X \times Y$ and $\mathcal{S} \subseteq Z \times W$:
	$$t \,\Obser(\mathcal{R})\, r \wedge (\forall x,y, x \,\mathcal{R}\, y \Rightarrow f(x) \, \mathcal{S} \, g(y)) \quad \implies \quad t[x \mapsto f(x)] \, \Obser(\mathcal{S}) \, r[y \mapsto g(y)]$$
\end{lemma}

\begin{proof}
	Let $\mathcal{R} \subseteq X \times Y$ and $\mathcal{S} \subseteq A \times B$. Assume: (i) $\forall x,y, x \,\mathcal{R}\, y \Rightarrow f(x) \,\mathcal{S}\, g(y)$ and (ii) $t \,\Obser(\mathcal{R})\, r$. 
	
	Take $o \in \Obser$ and $K \subseteq TA$ such that $t[x \mapsto f(x)] \in o(K)$. 
	Take $L = f^{-1}(K)$, then $t \in o(L)$. 
	So by (ii), $r \in o(\Rise{\mathcal{R}}{L})$. 
	For $y \in (\Rise{\mathcal{R}}{L})$, there is an $x \in L$ such that $x \,\mathcal{R}\, y$, hence by (i) we have $f(x) \,\mathcal{S}\, g(y)$. 
	Since $x \in L$ means $f(x) \in K$, it holds that $g(y) \in (\Rise{\mathcal{S}}{K})$, hence $(\Rise{\mathcal{R}}{L}) \subseteq g^{-1}(\Rise{\mathcal{S}}{K})$. 
	Using upwards closure of $o$, $r \in o(g^{-1}(\Rise{\mathcal{S}}{K}))$, and hence $r[y \mapsto g(y)] \in o(\Rise{\mathcal{S}}{K})$.
	Since this is for all $o \in \Obser$ and $K \subseteq TA$ with $t[x \mapsto f(x)] \in o(K)$, we can conclude that $t[x \mapsto f(x)] \,\Obser(\mathcal{S})\, r[y \mapsto g(y)]$.
\end{proof}	

\noindent
The relator also interacts  well with the monad structure on $T$.
\begin{lemma}\label{Rel_prop3}
	If $\,\Obser$ is a decomposable set of upwards closed modalities, then:
	\begin{enumerate}
		\item $x \,\mathcal{R}\, y~ \Rightarrow ~\eta(x) \,\Obser(\mathcal{R})\, \eta(y)$.
		\item $t \,\Obser(\Obser(\mathcal{R}))\, r~ \Rightarrow~ \mu t \,\Obser(\mathcal{R})\, \mu r$.
	\end{enumerate}
\end{lemma}

\begin{proof} We prove the properties separately.
	\begin{enumerate}
		\item Note that $\eta(x) \in o(A)$ either means $x \in A$ and $\ast \in \denote{o}$, or $\bot \in \denote{o}$. 
		By upwards closure, if $\bot \in \denote{o}$ then $\denote{o} = T\UnitType$. 
		If $x \in A$ then $y \in \Rise{\mathcal{R}}{A}$. 
		Either way, $\eta(y) \in o(\Rise{\mathcal{R}}{A})$.
		
		\item Take $t \,\Obser(\Obser(\mathcal{R}))\, r$ and $(\mu t) \in o(K)$ where $o \in \Obser$.
		Take $A \subseteq T\UnitType$ and $\gamma \in \Obser$ such that ${t[a \mapsto \tis{a}{K}]} \in \gamma(A)$. Let $S := \{a \,\mid\, \tis{a}{K} \in A\}$, then $t \in \gamma(S)$. 
		By $t \,\Obser(\Obser(\mathcal{R}))\, r$ we get $r \in \gamma(\Rise{\Obser(\mathcal{R})}{S})$. 
		
		In the following paragraph we prove that for $b \in (\Rise{\Obser(\mathcal{R})}{S})$, $\tis{b}{\Rise{\mathcal{R}}{K}} \in (\Rise{\prebas}{A})$, using characterisation \ref{base:three} of $\prebas$ from Proposition \ref{prop:base_relation}.
		
		For $b \in (\Rise{\Obser(\mathcal{R})}{S})$, there is an $a \in S$ such that $\tis{a}{K} \in A$ and $a \,\Obser(\mathcal{R})\, b$. 
		If $\tis{a}{K} \in \delta(\{\ast\})$ for some $\delta \in \Obser$, then $a \in \delta(K)$, so $b \in \delta(\Rise{\mathcal{R}}{K})$ and hence $\tis{b}{\Rise{\mathcal{R}}{K}} \in \delta(\{\ast\})$. 
		If $\tis{a}{K} \in \delta(\emptyset)$ for some $\delta \in \Obser$, then $a \in \delta(\emptyset)$, so we have $b \in \delta(\Rise{\mathcal{R}}{\emptyset})$. 
		Since $\Rise{\mathcal{R}}{\emptyset} = \emptyset$, $b \in \delta(\emptyset)$ and hence $\tis{b}{\Rise{\mathcal{R}}{K}} \in \delta(\emptyset)$. 
		With the previous two derivations, we see that $\tis{a}{K} \prebas \tis{b}{\Rise{\mathcal{R}}{K}}$, hence $\tis{b}{\Rise{\mathcal{R}}{K}} \in (\Rise{\prebas}{A})$ since $\tis{a}{K} \in A$. So the following inclusion holds:
		$$(\Rise{\Obser(\mathcal{R})}{S}) \,\,\subseteq\,\, \{b \,\mid\, \tis{b}{\Rise{\mathcal{R}}{K}} \in (\Rise{\prebas}{A})\} \,.$$
		
		By upwards closure of $\gamma$ and $r \in \gamma(\Rise{\Obser(\mathcal{R})}{S})$ we get $r \in \gamma(\{b \,\mid\, \tis{b}{\Rise{\mathcal{R}}{K}} \in (\Rise{\prebas}{A})\})$. 
		So we have derived that $(t[a \mapsto \tis{a}{K}] \in \gamma(A)) \Rightarrow (r[b \mapsto \tis{b}{\Rise{\mathcal{R}}{K}}] \in \gamma(\Rise{\prebas}{A}))$ for all $\gamma$ and $A$. 
		By Lemma \ref{lem:decomp_equ}, $t[a \mapsto \tis{a}{K}] \modbel r[b \mapsto \tis{b}{\Rise{\mathcal{R}}{K}}]$, hence by decomposability we get $\tis{(\mu t)}{K} = {\mu (t[a \mapsto \tis{a}{K}])} \prebas \mu (r[b \mapsto \tis{b}{\Rise{\mathcal{R}}{K}}]) = \tis{(\mu r)}{\Rise{\mathcal{R}}{K}}$. 
		Hence with $\mu t \in o(K)$ assumed at the beginning, we get $\mu t \in o(\Rise{\mathcal{R}}{K})$, and we conclude $\mu t \,\Obser(\mathcal{R})\, \mu r$.
	\end{enumerate}
\end{proof}

\noindent
Finally, the following properties show that relators behave well with respect to the order on trees.

\begin{lemma}\label{Rel_prop2}
	If $\,\Obser$ only contains Scott open modalities, then for $\mathcal{R} \subseteq X \times Y$:
	\begin{enumerate}
		\item If $\,\mathcal{R}$ is reflexive, then for any $t \in TX, r \in TY$ it holds that $t \leq r \Rightarrow t \,\Obser(\mathcal{R})\, r$.
		\item For any two sequences $u_0 \leq u_1 \leq u_2 \leq \dots$ from $TX$ and $v_0 \leq v_1 \leq v_2 \leq \dots$ from $TY$: 
		
		$\forall n, (u_n \,\Obser(\mathcal{R})\, v_n) ~ \Rightarrow~ (\sqcup_n u_n) \,\Obser(\mathcal{R})\, (\sqcup_n v_n)$
	\end{enumerate}
\end{lemma}

\begin{proof} We separate the proofs by property.
	\begin{enumerate}
		\item If $\mathcal{R}$ is reflexive then $\Obser(\mathcal{R})$ is reflexive by Lemma \ref{Rel_prop1}.  Now for $t \leq r$, if $t \in o(A)$ then since $\tis{t}{A} \leq \tis{r}{A}$ we have $r \in o(A)$. 
		By reflexivity $r \, \Obser(\mathcal{R}) \, r$, so $r \in o(\Rise{\mathcal{R}}{A})$.
		
		\item Take $A \subseteq X$ and $o \in \Obser$ such that $(\sqcup_n u_n) \in o(A)$. 
		Now $\tis{(\sqcup_n u_n)}{A} = \sqcup_n (\tis{u_n}{A})$ so by Scott openness, there is an $m$ such that $u_m \in o(A)$. 
		Using the assumption we derive $v_m \in o(\Rise{\mathcal{R}}{A})$. 
		Note that $v_m \leq \sqcup_n v_n$, hence $\tis{v_m}{\Rise{\mathcal{R}}{A}} \leq \sqcup_n \tis{v_n}{\Rise{\mathcal{R}}{A}}$, so since $o$ is Scott open we conclude that $(\sqcup_n v_n) \in o(\Rise{\mathcal{R}}{A})$.
	\end{enumerate}
\end{proof}

The lemmas above list the core properties of the relator, which are satisfied when our set $\Obser$ is decomposable and contains only Scott open modalities. 
The results below follow from those above, specifically the next result combines Lemma \ref{Rel_prop3} with the fact that $\Obser(\text{id}_{\UnitType}) = \,\prebas$ and $\Obser(\prebas) = \modbel$.

\begin{corollary}\label{dec_equ2}
		If $\,\Obser$ contains only upwards closed modalities, then:
		$$\Obser \textit{ is decomposable} \quad \iff \quad \forall \mathcal{R} \subseteq X \times Y, \forall t \in TTX, r \in TTY, (t \,\Obser(\Obser(\mathcal{R}))\, r \Rightarrow \mu t \, \Obser(\mathcal{R}) \, \mu r) \enspace .$$
\end{corollary}

Lastly, we verify that sequencing and the algebraic effect operations interact in the appropriate way with the relator action.

\begin{corollary}\label{Rel_prop4}
	If $\,\Obser$ is a decomposable set of upwards closed modalities, then lifted relations are preserved by Kleisli lifting and effect operations:
	\begin{enumerate}
		\item Given $f: X \to TZ$, $g: Y \to TW$, $\mathcal{R} \subseteq X \times Y$ and $\mathcal{S} \subseteq Z \times W$, if for all $x \in X$ and $y \in Y$ we have $\,x \,\mathcal{R}\, y \Rightarrow f(x) \, \Obser(\mathcal{S}) \, g(y)$ and if $\,t \,\Obser(\mathcal{R})\, r$ then $\mu (t[x \mapsto f(x)]) \, \Obser(\mathcal{S}) \, \mu (r[y \mapsto g(y)])$.
		\item $(\forall k, u_k \,\Obser(\mathcal{S})\, v_k) ~ \Rightarrow~ \sigma(u_0,u_1,\dots) \,\Obser(\mathcal{S})\, \sigma(v_0,v_1,\dots)$
	\end{enumerate}
\end{corollary}

\begin{proof} We prove the properties separately.
	\begin{enumerate}
		\item Using Lemma \ref{Rel_prop1a} on the assumptions we get $t[x \mapsto f(x)] \, \Obser(\Obser(\mathcal{S})) \, r[y \mapsto g(y)]$. We can then apply property 2 of Lemma \ref{Rel_prop3} to get the correct result.
		\item We apply the previous property to the following data; $t = r = \sigma(0,1,2,\dots) \in T\NatNum$, $f(n) = u_n$, $g(n) = v_n$ for all $n \in \NatNum$, and $\mathcal{R} = \textit{id}_{\NatNum}$. The conclusion follows directly.
	\end{enumerate}
\end{proof}

\noindent
Point 2 of Corollary \ref{Rel_prop4} has been stated in such a way that it contains both the infinite arity case $\alpha^{\NatType} \rightarrow \alpha$ and the finite arity case $\alpha^{n} \rightarrow \alpha$ of effect operations $\sigma$. So it states that any lifted relation is preserved under any of the predefined algebraic effects (assuming equality between any possible natural numbers arguments given by the arity).
The next section focusses on proving Theorem \ref{Main_Theorem} using a generalisation of Howe's method.

%% file: 6_Howe.tex
\section{Howe's method}\label{section:Howe}

In this section, we apply Howe's method~\cite{How89,How} to establish the compatibility of
applicative $\Obser$-similarity and $\Obser$-bisimilarity, and hence compatibility of the positive behavioural preorder and full behavioural equivalence. Given a relation $\mathcal{R}$ on terms, one defines its \emph{Howe closure} $\mathcal{R}^{\bullet}$, which is compatible and contains the open extension $\mathcal{R}^{\circ}$. 
Our proof makes use of the 
relator properties from Section~\ref{section:similarity}, closely following the approach of Dal Lago \emph{et al.}~\cite{Relational}. We will only give an  outline of the proof, focussing on the main steps. Detailed proofs can be found in Appendix \ref{section:proofs}.

Recall from Section~\ref{section:preorder} that, for any closed relation $\mathcal{R}$, we can define the open extension $\mathcal{R}^{\circ}$ as $\Gamma \vdash M \,\mathcal{R}^{\circ}\, N : \TypeOne \Leftrightarrow \forall \overrightarrow{V} : \Gamma, M[\overrightarrow{V}] \,\mathcal{R}\, N[\overrightarrow{V}]$. 
We define two more closure operations.

\begin{definition}\label{def:How}
	Taking an \emph{open} relation $\mathcal{R}$, we define the \emph{compatible refinement} $\widehat{\mathcal{R}}$ using the derivation rules in Fig. \ref{fig:clos}. For a closed relation $\mathcal{R}$ we define the \emph{Howe closure} $\mathcal{R}^{\bullet}$ as the smallest open relation $\mathcal{S}$ closed under the rules:
	\[
	\frac{\Gamma \vdash V \,\widehat{\mathcal{S}}_v\, W \quad \quad \Gamma \vdash W \,\mathcal{R}^{\circ}_v\, L}{\Gamma \vdash V  \,\mathcal{S}_v\, L} \textbf{(HV)} \quad \quad \frac{\Gamma \vdash M \,\widehat{\mathcal{S}}_c\, N \quad \quad \Gamma \vdash N \,\mathcal{R}^{\circ}_c\, K}{\Gamma \vdash M  \mathcal{S}_c K} \textbf{(HC)}
	\]
\end{definition}
\begin{figure}
	\[
	\frac{}{\Gamma, x : \TypeOne \vdash x  \,\,\widehat{\mathcal{R}}_v\,  x } \textbf{(C1)} \quad \quad
	\frac{}{\Gamma \vdash Z  \,\,\widehat{\mathcal{R}}_v\,  Z} \textbf{(C2)} \quad \quad
	\frac{\Gamma \vdash V  \,\mathcal{R}_v\,  V'}{\Gamma \vdash S(V)  \,\,\widehat{\mathcal{R}}_v\,  S(V')} \textbf{(C3)}  
	\]
	\[
	\frac{\Gamma \vdash V \,\mathcal{R}_v\,  V'}{\Gamma \vdash \textbf{return}(V)  \,\,\widehat{\mathcal{R}}_c\,  \textbf{return}(V') } \textbf{(C4)} \quad \quad\frac{\Gamma, x : \TypeOne \vdash M  \,\mathcal{R}_c\,  M'}{\Gamma \vdash (\lambda x:\TypeOne.M)   \,\,\widehat{\mathcal{R}}_v\,  (\lambda x:\TypeOne.M')} \textbf{(C5)}
	\]
	\[
	\frac{\Gamma \vdash V  \,\mathcal{R}_v\,  V' \quad \quad \Gamma \vdash W  \,\mathcal{R}_v\,  W'}{\Gamma \vdash (V W)   \,\,\widehat{\mathcal{R}}_c\,  (V' W')} \textbf{(C6)} \quad \quad
	\frac{\Gamma \vdash V  \,\mathcal{R}_v\,  V'}{\Gamma \vdash \text{\textbf{fix}}(V)  \,\,\widehat{\mathcal{R}}_v\,  \text{\textbf{fix}}(V')} \textbf{(C7)}
	\]
	\[
	\frac{\Gamma \vdash V  \,\mathcal{R}_v\,  V' \quad \quad \Gamma \vdash M  \,\mathcal{R}_c\,  M' \quad \quad \Gamma, x:\NatType \vdash N  \,\mathcal{R}_c\,  N'}
	{\Gamma \vdash (\text{ \textbf{case} } V \text{ \textbf{of} } \{Z \Rightarrow M; S(x) \Rightarrow N\})  \,\,\widehat{\mathcal{R}}_c\,  (\text{\textbf{case} } V' \text{ \textbf{of} } \{Z \Rightarrow M'; S(x) \Rightarrow N'\})} \textbf{(C8)}
	\]
	\[
	\frac{\Gamma \vdash M  \,\mathcal{R}_c\,  M'  \quad \quad \Gamma, x:\TypeOne \vdash N  \,\mathcal{R}_c\,  N'}{\Gamma \vdash (\textbf{let } M \Rightarrow x \textbf{ in } N)  \,\,\widehat{\mathcal{R}}_c\,  (\textbf{let } M' \Rightarrow x \textbf{ in } N')} \textbf{(C9)}	
	\]
	\[
	\frac{\Gamma \vdash M_i \, \mathcal{R}_c \, M_i'}{\Gamma \vdash \sigma(M_0,M_1,\dots) \,\, \widehat{\mathcal{R}}_c \, \sigma(M_0',M_1',\dots)}\textbf{(CA)} 
	\quad \quad
	\frac{\Gamma \vdash V \,  \,\mathcal{R}_v\,  \, V'}{\Gamma \vdash \sigma(V) \,\, \widehat{\mathcal{R}}_c \, \sigma(V')} \textbf{(CB)}
	\]
	\[
	\frac{\Gamma \vdash V \,  \,\mathcal{R}_v\,  \, V' \quad \quad \Gamma \vdash M_i \, \mathcal{R}_c \, M_i'}{\Gamma \vdash \sigma(V;M_0,M_1,\dots) \,\, \widehat{\mathcal{R}}_c \, \sigma(V';M_0',M_1',\dots)}\textbf{(CC)} \quad \quad
	\frac{\Gamma \vdash V \,  \,\mathcal{R}_v\,  \, V' \quad \quad \Gamma \vdash W \,  \,\mathcal{R}_v\,  \, W'}{\Gamma \vdash \sigma(V;W) \,\, \widehat{\mathcal{R}}_c \, \sigma(V';W')}\textbf{(CD)}
	\]
	\caption{Compatible refinement rules}
	\label{fig:clos}
\end{figure}

Note that $\mathcal{R}$ is compatible if and only if $\,\widehat{\mathcal{R}} \subseteq \mathcal{R}$. The Howe closure can also be expressed as the least solution for $\mathcal{S}$ of the equation $\mathcal{S} = \widehat{\mathcal{S}} \mathcal{R}^{\circ}$, or of the inclusion  $\widehat{\mathcal{S}} \mathcal{R}^{\circ} \subseteq \mathcal{S}$. 

The following result contains the main properties of the Howe closure we are interested in, e.g. from Lassen~\cite{Lassen}.

\begin{restatable}[]{lemma}{LemmaHowe}
	\label{Hprop1}
	If \,$\mathcal{R}$ is reflexive, then:
	\begin{enumerate}
		\item $\mathcal{R}^{\bullet}$ is compatible, hence reflexive.
		\item $\mathcal{R}^{\circ} \subseteq \mathcal{R}^{\bullet}$.
		\item If $x:\TypeTwo \vdash A \, \mathcal{R}^{\bullet} B$ and $V, W:\TypeTwo$ such that $V \,\mathcal{R}^{\bullet}\, W$, then $A[V] \,\mathcal{R}^{\bullet}\, B[W]$
	\end{enumerate}
\end{restatable}

The main step of the Howe's method proof is establishing the following result.

\begin{restatable}[]{proposition}{MainProp}
	\label{How_sim}
	If $\,\Obser$ is a decomposable set of Scott open modalities, then for any preorder applicative $\,\Obser$-simulation $\,\sqsubseteq$, the Howe closure $\,\sqsubseteq^{\bullet}$ limited to closed terms is an applicative $\,\Obser$-simulation.
\end{restatable}

This is proven by checking the simulation properties of Definition \ref{definition:sim} one by one. The most difficulty arises when trying to prove property 2, which is done inductively using the following lemma, and uses most of the relator properties from Subsection \ref{section:relprop}.

\begin{restatable}[Key Lemma]{lemma}{KeyLemma}
	\label{lemma_key}
	For any $n \in \NatNum$, given two closed computation terms $A$ and $B$, if $A \sqsubseteq^{\bullet} B$ then $|A|_n \,\Obser(\sqsubseteq^{\bullet})\, |B|$.
\end{restatable}

Having established the Key Lemma, we can use Lemma \ref{Rel_prop2}, implying $(\forall n \in \NatNum, \, |A|_n \,{\Obser(\sqsubseteq^{\bullet})}\, |B|) \Rightarrow |A| \,\Obser(\sqsubseteq^{\bullet})\, |B|$, to derive property 2 of $\Obser$-simulations and finalise the proof of Proposition \ref{How_sim}. The complete proofs of the previous lemmas and Proposition \ref{How_sim} can be found in Appendix \ref{section:proofs}.

Using Proposition \ref{How_sim}, we can derive the compatibility of  applicative $\Obser$-similarity and $\Obser$-bisimilarity. 

\begin{theoremm}[A]\label{theorem:sim_com}
	\label{theorem:similarity-compatible}
	If $~\Obser$ is a decomposable set of Scott open modalities, then the open extension of the relation of applicative
	$\,\Obser$-similarity is compatible.
\end{theoremm}

\begin{proof}
	Before we start the proof, we observe the following general result. If a compatible relation $\mathcal{R}$ is contained in a closed relation $\mathcal{S}$ when limited to closed terms, then $\mathcal{R} \subseteq \mathcal{S}^{\circ}$ since:
	$$\Gamma \vdash M \,\mathcal{R}\, N \Rightarrow \forall \overrightarrow{V}:\Gamma, M[\overrightarrow{V}/\overrightarrow{x}] \,\mathcal{R}\, N[\overrightarrow{V}/\overrightarrow{x}] \Rightarrow \forall \overrightarrow{V}:\Gamma, M[\overrightarrow{V}/\overrightarrow{x}] \,\mathcal{S}\, N[\overrightarrow{V}/\overrightarrow{x}] \Rightarrow \Gamma \vdash M \,\mathcal{S}^{\circ}\, N \enspace .$$
	
	We write $\simil$ for the relation of $\Obser$-similarity. Since $\simil$ is an $\Obser$-simulation, we know by Proposition \ref{How_sim} that $\simil^{\bullet}$ limited to closed terms is one as well, and hence is contained in the largest $\Obser$-simulation $\simil$. By Lemma \ref{Hprop1} it holds that $\simil^{\bullet}$ is compatible, and by the observation before we know it is contained in the open extension $\simil^{\circ}$.
	Again by Lemma \ref{Hprop1}, we also know that $\simil^{\circ}$ is contained in $\simil^{\bullet}$.
	We can conclude that $\simil^{\circ}$ is equal to the Howe closure $\simil^{\bullet}$, which is compatible.
\end{proof}

To prove that $\Obser$-bisimilarity is compatible, we need another result from Lassen~\cite{Lassen} (a proof is also given in Appendix \ref{section:proofs}).

\begin{restatable}[]{lemma}{LemmaSym}
	\label{lemma:symmetry}
	If \,$\mathcal{R}^{\circ}$ is symmetric and reflexive, then $\mathcal{R}^{\bullet *}$ (the transitive closure of $\mathcal{R}^{\bullet}$) is symmetric.
\end{restatable}

Given these facts, we can derive the following.

\addtocounter{theoremm}{-1}
\begin{theoremm}[B]\label{theorem:bis_com}
	If $\,\Obser$ is a decomposable set of Scott open modalities, then the open extension of the relation of applicative $\,\Obser$-bisimilarity is compatible.
\end{theoremm}

\begin{proof}
	We write $\Obser$-bisimilarity as $\bisim$. From Proposition \ref{How_sim} we know that $\bisim^{\bullet}$ on closed terms is an $\Obser$-simulation.
	Using property 2 of Lemma \ref{Rel_prop1}, we can verify that the transitive closure $\bisim^{\bullet *}$ is an $\Obser$-simulation as well. 
	Since $\bisim$ is reflexive and symmetric, we know by the previous lemma that $\bisim^{\bullet *}$ is symmetric. Hence $\bisim^{\bullet *}$ is an $\Obser$-bisimulation, implying $(\bisim^{\bullet *}) \subseteq (\bisim^{\circ})$ by compatibility of $\bisim^{\bullet *}$ (due to Lemma \ref{Hprop1} and the same reasoning from the previous proof). 
	Again by Lemma \ref{Hprop1}, it holds that $(\bisim^{\circ}) \subseteq (\bisim^{\bullet}) \subseteq (\bisim^{\bullet *})$, hence $(\bisim^{\bullet *}) = (\bisim^{\circ})$. We can conclude that $\bisim^{\circ}$ is compatible.
\end{proof}

\noindent
Note that Theorem~\ref{Main_Theorem} is an immediate consequence of Theorems~\ref{log_is_sim} and~\ref{theorem:similarity-compatible}.
As such, we have finished the proof of compatibility of the full behavioural equivalence and positive behavioural preorder.

%% file: 7_Contextual.tex
\section{Comparing equivalences}\label{section:contextual}

Our development thus far has focused on two main relations for comparing programs. The first is the behavioural equivalence
${\equiv_\LogNeg}$ defined via the full logic with negation and characterised as 
$\mathcal{O}$-bisimilarity. The second is the positive behavioural  preorder 
${\sqsubseteq_{\LogPos}}$, defined via the negation-free logic and characterised as
$\mathcal{O}$-similarity. The latter  preorder induces the equivalence relation ${\equiv_{\LogPos}}$ of mutual
$\mathcal{O}$-similarity. 

As discussed after Definition~\ref{definition:preorder}, there is an inclusion
${\equiv_\LogNeg} \subseteq {\equiv_{\LogPos}}$. That is, $\mathcal{O}$-bisimilarity implies mutual $\mathcal{O}$-similarity. It is a standard fact in concurrency theory that ordinary mutual similarity is in general a strictly coarser relation than bisimilarity. A similar fact holds in the context of the present paper: it is possible that
the inclusion ${\equiv_\LogNeg} \subseteq {\equiv_{\LogPos}}$ is strict. We illustrate this with an example, adapted from Lassen's PhD thesis \cite[page 92]{Lassen}, for which the two relations do not coincide. (Ong gives another example of a similar phenomenon~\cite{Ong}.)

\begin{example}[Separating the full and positive behavioural equivalence, cf.\ \cite{Lassen}]
\label{example:separation:A}
We take $\Sigma := \{\textit{or}: \alpha^2 \rightarrow \alpha\}$ and $\Obser := \{\Diamond, \Box\}$ from Example \ref{example:non}.
Let $\Omega : \UnitType$ be some always diverging computation term. We consider the following two computation terms of type $(\UnitType \to \UnitType)$:
\begin{align*}
M ~ :=  ~ & \textit{or}\,(~\textbf{return}(\lambda x: \UnitType.\, (\textbf{return } \ast)), ~~ \textbf{return}(\lambda x: \UnitType. \Omega)) \\
N  ~ := ~  & \textit{or}\,(M, ~ \textbf{return}(\lambda x: \UnitType. \, (\textit{or}\,(~\textbf{return } \ast,~~ \Omega)))
\end{align*}
We make use of our logics $\LogNeg$ and $\LogPos$ to give simple proofs that 
$M \equiv_{\LogPos} N$  but $M \not\equiv_\LogNeg N$.

By definition, $M \equiv_{\LogPos} N$ if the two terms satisfy the same positive computation formulas of type $(\UnitType \to \UnitType)$.
Note that $\Diamond(\bot)$ and $\Box(\bot)$ are not satisfied by any terms, $\Diamond(\phi) \vee \Box(\phi)$ is satisfied precisely if $\Diamond(\phi)$ is, and $\Diamond(\phi) \wedge \Box(\phi)$ is satisfied precisely if $\Box(\phi)$ is.
Moreover, $(V \mapsto \Phi) \wedge (V \mapsto \Psi)$ and $(V \mapsto \Phi) \vee (V \mapsto \Psi)$ are equivalent to $(V \mapsto (\Phi \wedge \Psi))$ and $(V \mapsto (\Phi \vee \Psi))$ respectively.
By these observations and Lemma \ref{comp_clas}, we can reduce the set of formulas we need to check to establish equivalence between two computation terms of type $(\UnitType \to \UnitType)$ to:
$$\{\Diamond(\top),\,\, \Box(\top),\,\, \Diamond(\ast \mapsto \Diamond(\top)),\,\, \Diamond(\ast \mapsto \Box(\top)),\,\, \Box(\ast \mapsto \Diamond(\top)),\,\, \Box(\ast \mapsto \Box(\top))\}$$
The terms $M$ and $N$ both satisfy the same formulas from this set, namely $\Diamond(\top), \Box(\top), \Diamond(\ast \mapsto \Diamond(\top)),$ and $\Diamond(\ast \mapsto \Box(\top))$.
We conclude that $M \equiv_{\LogPos} N$.

The terms $M$ and $N$ are not however equivalent under full behavioural equivalence,
since $N$ satisfies the formula $\Diamond(\ast \mapsto (\Diamond(\top) \wedge \neg \Box(\top)))$, but $M$ does not. (This formula says of a computation  of type $(\UnitType \to \UnitType)$ that it may return a lambda term which, when applied to argument $*$, results in a nondeterministic computation that may diverge and may also terminate.)
Thus $M \not\equiv_\LogNeg N$.
\end{example}

It is also instructive to compare ${\equiv_\LogNeg}$   and $\equiv_{\LogPos}$ to 
\emph{contextual equivalence}, which is 
often taken to be the default equivalence for applicative programming languages. 
For the language of this paper, the  collection $\mathcal{O}$ of modalities used to define our logics serves a further purpose: via the  preorder $\prebas$ on unit-type computation trees from Section~\ref{section:preorder}, it gives rise to a natural definition of 
\emph{contextual preorder} between terms of the same type and aspect (value or computation). The idea is that two terms $M,N$ of the same type $\tau$ and aspect are related by the  {contextual preorder $\sqsubseteq_\text{ctxt}$} if for any context $C(-)$ of unit type, accepting terms of type $\tau$ of the relevant aspect, we have $|C(M)| \prebas |C(N)|$.\footnote{If one expands the use of $\prebas$ in this definition as in statement~\ref{base:three} of Proposition~\ref{prop:base_relation}, one obtains a definition of 
contextual preorder in terms of basic unit-type `observations' 
of the form $o (\{*\})$ and $o (\emptyset)$.}
The same definition can  be given more elegantly using compatibility.

\begin{definition}[Contextual preorder/equivalence]
	The \emph{contextual preorder} relation $\sqsubseteq_\text{ctxt}$ is the largest compatible preorder satisfying: 
 	\begin{equation}
 	\label{equation:contextual}
 	\text{for all $M,N \in \textit{Com}(\UnitType)$, it holds that $M \sqsubseteq_\text{ctxt} N$ implies
 	$|M|  \prebas |N|$.}
 	\end{equation}
 	The relation $\equiv_\text{ctxt}$ of \emph{contextual equivalence} is the equivalence relation determined by
 	$\sqsubseteq_\text{ctxt}$ (i.e., the intersection of $\sqsubseteq_\text{ctxt}$ and its converse).
\end{definition}
 
If $\mathcal{O}$ is a decomposable set of Scott-open modalities then, by Theorem~\ref{Main_Theorem} and Proposition~\ref{prop:preceq}, ${\sqsubseteq_{\LogPos}}$ is a compatible preorder satisfying condition~(\ref{equation:contextual}) above. Hence,  there is an inclusion of relations 
${\sqsubseteq_{\LogPos}} \subseteq {\sqsubseteq_\text{ctxt}}$.  
We thus have a chain of inclusions between equivalence relations.
$${\equiv_\LogNeg}\quad \subseteq \quad {\equiv_{\LogPos}} \quad \subseteq  \quad {\equiv_\text{ctxt}} \enspace .$$
\noindent
Having already shown the left-hand inclusion to be strict in general, we now show that the right-hand inclusion can also be strict. This is again based on an example from Lassen~\cite{Lassen}. 

\begin{example}[Separating contextual equivalence and positive behavioural equivalence, c.f.\ \cite{Lassen}]
\label{example:separation:B}
We take $\Sigma := \{\textit{or}: \alpha^2 \rightarrow \alpha\}$ and $\Obser := \{\Diamond\}$ from Example \ref{example:non}, that is we consider angelic nondeterminism.

The example we give differs marginally  from Lassen~\cite[page 90]{Lassen}, 
in that  one of the terms involved does not exhibit any effectful behaviour.
Let $\textit{id}: \NatType \to \NatType$ be the value term given by the identity function $\lambda x : \NatType.(\textbf{return } x)$. 
We can define approximations of this identity function, using a term $\textit{apid}: \NatType \to (\NatType \to \NatType)$ such that for any two natural numbers $n, m \in \mathbb{N}$:
$$|\textbf{let } (\textit{apid}\,\,\overline{n}) \Rightarrow y \textbf{ in } (y\,\,\overline{m})| \quad = \quad \begin{cases} \overline{m}\quad & \text{if } m < n\\ \bot & \text{otherwise.} \end{cases}$$
So for any natural number $n \in \mathbb{N}$, `$\textit{apid}\,\,\overline{n}$' returns the $n$-th approximation of the identity function.

Given the computation term $?N$ of type $\NatType$ from Example \ref{example:nonA} at the end of Section \ref{section:language}, which nondeterministically returns any natural number, we define the nondeterministic approximation of the identity function:
$$M \quad := \quad \textbf{let } ?N \Rightarrow x \textbf{ in } (\textbf{return } \lambda y. (\textbf{let } (\textit{apid}\,\,x) \Rightarrow z \textbf{ in } (z\,\,y)))$$
This term nondeterministically returns any approximation of the identity function. It holds that $M$ is contextually equivalent to $\textbf{return}(\textit{id})$ under the angelic interpretation of nondeterminism, since if a context holding the function terminates, it will have in the process of computation only fed the function a finite number of arguments. As such, in any context, there is always a large enough approximation of the identity which cannot be distinguished from the identity function. This outline argument that  $M \equiv_\text{ctxt} \textbf{return}(\textit{id})$ 
can be made precise 
along the lines of Lassen~\cite[Example 6.4.4]{Lassen}.

However, $M$  and $\textbf{return}(\textit{id})$ are distinguished by the positive logic, since $M$ cannot return the identity function itself. Specifically, the formula $\Diamond(\bigwedge_{n \in \mathbb{N}} \overline{n} \mapsto \Diamond(\{n\}))$ is satisfied by  $\textbf{return}(\textit{id})$ but not by $M$.
Thus $M \not\equiv_{\LogPos} \textbf{return}(\textit{id})$.
\end{example}

It is an interesting question whether our modal behavioural logic can be used to better understand the relationship between applicative (bi)similarity and contextual equivalence/preorder.  For example, in the case of effects for which the relations differ (such as nondeterminism), can the logic be restricted to 
characterise contextual equivalence/preorder? 

%% file: 8_Pure.tex
\section{Pure behavioural logic}\label{section:pure}

In this section, we explore an alternative formulation of our logic. 
This has  both conceptual and practical  motivations.
Our  very approach to behavioural logic, fits into the category of \emph{endogenous} logics in the sense of Pnueli~\cite{pnueli77}.
Formulas ($\phi $ and $\Phi$) express properties of individual programs, through satisfaction relations
($V \models \phi$ and  $M \models \Phi$). Programs are thus considered as `models' of the logic, with the satisfaction relation being defined via program behaviour. 

It is conceptually appealing to push the separation between program and logic to its natural conclusion, and ask for the 
syntax of the logic to be independent of the syntax of the programming language. Indeed, it seems natural that 
it should be possible to express properties of program behaviour without knowledge of the syntax of the programming language. Under our formulation of the logic $\LogNeg$, this desideratum is violated by the value formula
$(V \mapsto \Psi)$ at function type, which mentions the programming language value $V$.

This issue can be addressed, by replacing the basic value formula $(V \mapsto \Psi)$ with the alternative
$(\phi \mapsto \Psi)$, already mentioned in Section~\ref{section:logic}. Such a change also has a practical motivation.
The formula $(\phi \mapsto \Psi)$ declares a precondition and postcondition for function application, supporting a useful specification style. It also supports 
the expression of various other properties that are natural from the point of view of program specification. For example, 
using negation, we can define the dual formula $\neg (\phi \mapsto \neg \Phi)$, with the resulting semantics\footnote{Given compatibility of $\sqsubseteq_{\LogPos}$, this  formula is equivalent to $\bigvee_{\{V \mid V \models \phi\}} \chi_V \mapsto \Psi$, where $\chi_V$ is from Lemma \ref{pre_form}. As such, it may also be formulated without the use of negation in the positive behavioural logic $\LogPos$.}:
\[
	W \models \neg (\phi \mapsto \neg \Psi) ~~~ \Leftrightarrow ~~~ \exists {V \!\in\! \textit{Val}(\TypeOne)}.~ V \models \phi ~\text{and}~ {WV} \models  \Psi \enspace  .
\]
Further examples of the usefulness of the pure behavioural logic will be given in Section~\ref{section:reasoning}.

\begin{definition}\label{definition:pure}
The \emph{pure behavioural logic} \ $\PureLogNeg$ is defined by  replacing rule (2) in Fig. \ref{figure:logneg} with the alternative:
\[
	\frac{\phi \in \textit{VF}\,(\TypeTwo) \quad \quad \Psi \in \textit{CF}\,(\TypeOne)}{(\phi \mapsto \Psi) \in \textit{VF}\,(\TypeTwo \rightarrow \TypeOne)}(2^*)
\]
The semantics is modified by defining $W \models (\phi \mapsto \Psi)$ using formula (\ref{equation:derived-semantics}) of Section~\ref{section:logic}. 
\[
	W \models (\phi \mapsto \Psi) ~~~ \Leftrightarrow ~~~ \forall {V \!\in\! \textit{Val}(\TypeOne)}.~ V \models \phi ~\text{implies}~ {WV} \models  \Psi \enspace  .
\]
\end{definition}

With the definition of  $(\phi \mapsto \Psi)$ within $\LogNeg$, given in (\ref{equation:define-formula}) of Section~\ref{section:logic}, we can see $\PureLogNeg$ as a fragment of $\LogNeg$ (and $\PureLogPos$ a fragment of  $\LogPos$), resulting in the following fact.
\begin{lemma}\label{FinV}
	Any formula of \ $\PureLogNeg$ and \ $\PureLogPos$ is equivalent to some formula of \ $\LogNeg$ and \ $\LogPos$ respectively.
\end{lemma}
Moreover, once we have established compatibility, we have the following result.
\begin{proposition}
\label{proposition:equi-expressive}
	If the open extension of $\equiv_\LogNeg$ is compatible then
	the logics $\LogNeg$ and $\PureLogNeg$ are equi-expressive. Similarly, if the open extension of $\sqsubseteq_{\LogPos}$ is compatible then 
	the logics
	$\LogPos$ and $\PureLogPos$ are equi-expressive.
\end{proposition}
\begin{proof}
Lemma \ref{FinV} tells us we can translate any formula from $\PureLogNeg$ into an equivalent formula from $\LogNeg$. For the  reverse translation, whose correctness proof is more interesting, we give more detail. Every value (respectively computation) formula,
$\phi$ (respectively $\Phi$), of $\LogNeg$ is inductively translated to a corresponding formula $\widehat{\phi}$ (respectively $\widehat{\Phi}$) of 
$\PureLogNeg$. We do this by induction on the structure of the formula, where we define the following translation:

\begin{gather*} 
	\widehat{\{n\}} := \{n\}, \quad \widehat{o \, \phi} := o \, \widehat{\phi}, \quad \widehat{\bigvee_I \phi} := \bigvee_I \widehat{\phi}, \quad \widehat{\bigwedge_I \phi} := \bigwedge_I \widehat{\phi}, \quad \widehat{\neg \phi} := \neg \widehat{\phi}
	\\
	\widehat{(V \mapsto \Phi)} ~ := ~ (\, \psi_V  \,  \mapsto\,  \widehat{\Phi} \, ) \, ,
\end{gather*}

\noindent
where, in the last case, $\psi_V$ is the formula $\chi_V^{\PureLogNeg}$ from Lemma \ref{pre_form}. Hence: $V \models_{\PureLogNeg} \psi_V$; and, for any $\psi$, if
$V  \models_{\PureLogNeg} \psi$ then $\psi_V \to \psi$ (meaning that $V' \models_{\PureLogNeg} \psi_V$ implies 
$V' \models_{\PureLogNeg} \psi$, for all $V'$). 

One now proves, by induction on types, followed by an induction on formulas $\phi$ (or $\Phi$) of this type, that the $\PureLogNeg$-semantics of $\widehat{\phi}$ (resp.\ $\widehat{\Phi}$) coincides
with the $\LogNeg$-semantics of ${\phi}$ (resp.\ ${\Phi}$). 
This induction is obvious in all cases except for formula
$\widehat{(V \mapsto \Phi)}$ of a function type, for which we need to show that: $W  \models_{\PureLogNeg} \psi_V  \,  \mapsto\,  \widehat{\Phi}$
if and only if $WV  \models_{\PureLogNeg} \widehat{\Phi}$.
 
For the interesting right-to-left implication, suppose that $WV  \models_{\PureLogNeg} \widehat{\Phi}$, and consider any $V'$ satisfying $V' \models_{\PureLogNeg}  \psi_V$.
By the induction hypothesis, using the defining property of $\psi_V$, we have that  $V' \equiv_{\LogNeg} V$. 
It then follows from the compatibility of $\equiv_{\LogNeg}$ 
that $WV' \equiv_{\LogNeg} WV$, whence $WV' \equiv_{\PureLogNeg} WV$, again by the induction hypothesis.
Thus it follows from $WV \models_{\PureLogNeg} \widehat{\Phi}$ that 
$WV'  \models_{\PureLogNeg} \widehat{\Phi}$, as required. 
 
The proof in the case of the positive logics is similar.
\end{proof}

Combining the above proposition with Theorem \ref{Main_Theorem} we obtain the following.
\begin{corollary}\label{corollary:pure_compat}
	Suppose $\mathcal{O}$ is a decomposable family of Scott-open modalities. Then  $\equiv_{\PureLogNeg}$
	coincides with $\equiv_{\LogNeg}$, and $\sqsubseteq_{\PureLogPos}$ coincides with
	$\sqsubseteq_{\LogPos}$. Hence the open extensions of  $\equiv_{\PureLogNeg}$, $\sqsubseteq_{\PureLogPos}$, and $\equiv_{\PureLogPos}$
	are compatible.
\end{corollary}

\noindent
We do not know any proof of the compatibility of the $\equiv_{\PureLogNeg}$ and $\sqsubseteq_{\PureLogPos}$ relations that does not go via the logic $\LogNeg$. In particular, it is the presence of general recursion in the programming language 
(via the $\textbf{fix}$ operator) that has thwarted attempts to give direct proofs of compatibility for the pure-logic-induced preorder and equivalence.

%% file: 9_Reasoning.tex
\section{Reasoning with the logic}\label{section:reasoning}

The behavioural logics considered in this paper are designed for the purpose of formalising the notion of 
`behavioural property', and for defining behavioural equivalence. As infinitary propositional logics, they are not directly suited to practical applications such as  specification and verification. 
Nevertheless, they serve as low-level logics into which more practical finitary logics can be translated. 
For this, the closure of the logics  under infinitary propositional logic is important.
For example, there are standard translations of quantifiers and least and greatest fixed points into infinitary propositional logic. Furthermore, practical idioms for expressing effect-related properties can also be translated into logical combinations of modal formulas.

We illustrate this last point, by showing how Hoare triples, the standard formalism for specifying programs with global store, compile into our propositional logic with $(s \rightarrowtail s')$ modalities, for the language with lookup and update operations (Example~\ref{exampleB:gs}).
Take for example the statement $\{l = z\}M\{l' = z!\}$, where $l, l' \in L$ are two locations, and $z$ is an
auxiliary variable ranging over the natural numbers, and $M$ is a unit-type computation term.
This statement asserts that after executing $M$, the location $l'$ will contain the factorial of whatever number was stored at $l$ before starting the computation.
In the case of the \emph{total-correctness} version of Hoare logic,
in which $M$ is required to terminate, the property stated by the Hoare triple is equivalent to:\footnote{For $P(x)$ some property on $X$, we write $\bigvee_{a \in X, P(a)} \phi_a$ for $\bigvee\{\phi_a \mid a \in X, P(a)\}$. If the predicate is always true, we write $\bigvee_{a \in X} \phi_a$. We use a similar notation for $\bigwedge$.}
$$
M ~ \models ~ \bigwedge_{n \in \NatNum} ~ \bigwedge_{s \in \NatNum^L, s(l) = n} ~ \bigvee_{s' \in \NatNum^L, s'(l') = n!} (s \rightarrowtail s')\top \enspace , $$
where the formula on the right  expresses the combined effect of the precondition and postcondition.
Similarly, the \emph{partial correctness} interpretation, under which termination is not guaranteed, is captured by:
$$M ~ \models ~ \bigwedge_{n \in \NatNum} ~ \bigwedge_{s \in \NatNum^L, s(l) = n} \left(
   \,  \neg \left(\bigvee_{s' \in \NatNum^L} (s \rightarrowtail s') \top\right) \, \vee \, 
   \left(\bigvee_{s' \in \NatNum^L, s'(l') = n!} (s \rightarrowtail s')\top\right) \right) \enspace .$$

Such compilations of high-level formulas (for example, Hoare triples) into low-level formulas (the infinitary propositional modal logic) are cumbersome. Nevertheless, they serve two purposes.
The first is that translatability into the infinitary logic serves as a test for any proposed high-level specification syntax, as the existence of  such a translation guarantees that the syntax expresses only behaviourally meaningful properties. The second purpose is that the propositional modal logic can potentially serve as a guideline for the extraction of compositional proof rules for the high-level logic.
This is based on an interesting fact: such compositional proof rules are available for the 
infinitary propositional modal logic itself. Interestingly, this property is intimately related to the 
notion of \emph{strong decomposability} defined in Section~\ref{section:preorder}.
We now expand on this point. 

We consider a format for a style of compositional proof rule for programs 
whose outermost constructor is a sequential composition:
$\textbf{let } M \Rightarrow x \textbf{ in } N\,$. The proof rule we consider is specified by a 
modality $o \in \Obser$ together with a family of pairs of modalities: $\{(o_i,o'_i)\}_{i \in I}\,$.
This data determines the \emph{\textbf{let} proof rule} below, which makes use of the pure behavioural logic $\PureLogNeg$ from Section~\ref{section:pure}.
\begin{equation}
\label{equation:let-rule}
\frac{~\{ \, M \models o_i \,\psi_i  \quad \lambda x \colon \! \rho. \, N \models \psi_i \mapsto o'_i \, \phi \, \}_{i \in I}~}
{\textbf{let } M \Rightarrow x \textbf{ in } N \models o\, \phi}
\end{equation}
In this rule, the modalities $o$, $o_i$ and $o'_i$ and their index set $I$ have been determined by the specifying data. The other components of the rule may be instantiated arbitrarily, thus the rule is parametric with respect to  terms $M$ and $N$ and the formulas $\phi$, $\phi_i$ and $\psi_i$.
The rule is compositional in the sense that its premises require properties to be established of $M$ and $N$ separately.

\begin{definition}
A \textbf{let} proof rule~(\ref{equation:let-rule}) is said to be \emph{sound} if,
for all types $\rho, \, \tau$, computation terms $\vdash M \colon  \rho$ and $x\colon\rho \vdash N \colon  \tau$, and value formulas $\phi \in \textit{VF}(\tau)$ and $\{\psi_i \in   \textit{VF}(\rho)\}_{i \in I}$, the following implication holds:
	$$ \forall i \in I  \, (\, M \models o_i\, \psi_i ~ \wedge ~ \lambda x:\!\rho. \,N \models \psi_i \mapsto o'_i\, \phi\,) \quad \implies \quad \textbf{let } M \!\Rightarrow \!x \textbf{ in } N \models o\, \phi \enspace .$$
\end{definition}

The modalities $o$ and $\{(o_i,o'_i)\}_{i \in I}\,$
required to specify a \textbf{let} proof rule are reminiscent of the definition of strong decomposability~\ref{def:strong}. In fact, there is a precise connection. Whenever $\Obser$
 is a strongly decomposable set of Scott open modalities, the set of sound \textbf{let} proof rules is 
 complete in the sense of the proposition below, which informally states that every
 true modal property of  a \textbf{let} term can be proved from properties of its component terms by means of a
 suitably chosen  sound \textbf{let} proof rule.

\begin{proposition}\label{prop:let_proof}
	Suppose that $\Obser$ is a strongly decomposable set of Scott open modalities. Suppose also that  $\textbf{let } M' \Rightarrow x \textbf{ in } N' \models o\, \phi'$ holds (where 
	$\vdash M' \colon \rho'$ and $x \colon \rho' \vdash N' \colon \tau'$).
	Then there is  a sound \textbf{let} proof rule~(\ref{equation:let-rule}) for $o$, with premise modalities $\{(o_i,o'_i)\}_{i \in I}\,$, such that 
	$$\forall i \in I \, (\, M'  \models o_i\, \psi'_i  ~   \wedge ~ \lambda x\colon\rho'. \, N' \models \psi'_i \mapsto o'_i\, \phi' \, ) \enspace , $$
	for some choice of  value formulas $\{\psi'_i  \in  \textit{VF}(\rho')\}_{i \in I}\,$.
\end{proposition}
\begin{proof}
	Assume that $\Obser$ is strongly decomposable and $\textbf{let } M' \Rightarrow x \textbf{ in } N' \models o\, \phi'$.
	By the operational semantics, $|\textbf{let } M' \Rightarrow x \textbf{ in } N'| = \mu(|M'|[V \mapsto |N'[V/x]|])$, so 
	$$\mu(|M'|[V \mapsto |N'[V/x]|[\models \phi']]) ~ = ~ \mu(|M'|[V \mapsto |N'[V/x]|])[\models \phi'] ~  \in~  o(\{\ast\}) \enspace .$$
	By strong decomposability, there is a collection $\{(o_i, o'_i)\}_{i \in I}$ of pairs of modalities such that: (1) $\forall i \in I, |M'|[V \mapsto |N'[V/x]|[\models \phi']] \in o_i(o'_i(\{\ast\}))$, and (2) $\forall r \in TT\UnitType$, if $\forall i \in I, r \in o_i(o'_i(\{\ast\}))$ then $\mu r \in o(\{\ast\})$.
	We prove that this collection $\{(o_i, o'_i)\}_{i \in I}$ specifies the desired \textbf{let} proof rule for $o$.
	
	\emph{Soundness:} To prove that the proof rule is sound,
	consider types  $\rho$ and $\tau$, terms $\vdash M \colon \rho$ and $x \colon \rho \vdash N \colon \tau$, and value formulas $\phi \in \textit{VF}(\tau)$ and $\{\psi_i \in \textit{VF}(\rho)\}_{i \in I}$
	such that: 
	$$ \forall i \in I \, ( M\models o_i\, \psi_i ~ \text{and}~  \lambda x\colon\rho. \,N \models \psi_i \mapsto o'_i\, \phi\,)
	\enspace . $$
	Let $r = |M|[V \mapsto |N[V/x]|[\models \phi]] \in TT\UnitType$. For any $i \in I$ and $V$ such that 
 	 $V \models \psi_i$, we have $|N[V/x]|[\models \phi] \in o'_i(\{\ast\})$. We have by 
	monotonicity and because $|M|[\models \psi_i] \in o_i(\{\ast\})$, that $r \in o_i(o'_i(\{\ast\}))\,$.
	This holds for all $i \in I$, so by property (2) of strong decomposability (given above), $\mu r \in o(\{\ast\})$, hence $|\textbf{let } M \Rightarrow x \textbf{ in } N|[\models \phi] \in o(\{\ast\})$ and we conclude that $\textbf{let } M \Rightarrow x \textbf{ in } N \models o\,\phi$.
	
	\emph{Completeness:} We now prove that  suitable value formulas $\{\psi'_i  \in  \textit{VF}(\rho')\}_{i \in I}\,$ exist, allowing us to use the rule as a method for proving that $\textbf{let } M' \Rightarrow x \textbf{ in } N' \models o\, \phi'$.
	For any $i \in I$, define $$\psi'_i := \bigvee\{\chi_V \mid N'[V/x] \models o'_i\, \phi'\} \enspace , $$ where $\chi_V$ is the characteristic function for $V$, as given by Lemma \ref{pre_form}, but translated into the pure logic $\PureLogNeg$ using Proposition~\ref{proposition:equi-expressive}.
	So if $W \models \chi_V$ then $V \PreGen W$, and since $N'[V/x] \models o'_i\, \phi'$ and $\PreGen$ is compatible, $N'[W/x] \models o'_i\, \phi'$.
	Hence, whenever $W \models \psi'_i$, we have that $N'[W/x] \models o'_i\, \phi'$, so $\lambda x\colon \rho'. \, N \models \psi'_i \mapsto o'_i\, \phi'$.
	
	By the definition of $\psi'_i$ and the observation above, it holds that, for any closed value term $W\colon \rho'$,  $W \models \psi$ if and only if $|N'[W/x]|[\models \phi'] \in o'_i(\{\ast\})$.
	By property (1) above of strong decomposability,  $|M'|[V \mapsto |N'[V/x]|[\models \phi']] \in o_i(o'_i(\{\ast\}))$, so $|M'|[\models \psi'] \in o_i(\{\ast\})$, and hence $M' \models o_i(\psi')$.
\end{proof}

\noindent
We remark that Proposition \ref{prop:let_proof} also works if we restrict to the logic $\PureLogPos$ of positive pure formulas.

We illustrate Proposition \ref{prop:let_proof}, by giving examples of complete collections of sound \textbf{let} proof rules in the case of two of our example effects: global store, and probabilistic choice. For global store, such a collection is given by \textbf{let} proof rules of the form:
$$
\frac{~M \models  (s \rightarrowtail s'')\,\psi  \quad \lambda x \colon \! \rho. \, N \models \psi \mapsto (s'' \rightarrowtail s') \, \phi ~}
{\textbf{let } M \Rightarrow x \textbf{ in } N \models (s \rightarrowtail s')\, \phi}
$$
For probabilistic choice, such a collection of rules is:
$$
\frac{~\{ \, M \models \mathsf{P}_{> a_i} \,\psi_i  \quad \lambda x \colon \! \rho. \, N \models \psi_i \mapsto \mathsf{P}_{> b_i}\, \phi \, \}_{i = 1}^{n}~}
{\textbf{let } M \Rightarrow x \textbf{ in } N \models \mathsf{P}_{> q}\, \phi}
$$
where each rule in the collection is specified by some $n \geq 1$ and rationals $0 < a_0 < \dots < a_n < 1$ and $1 > b_0 > \dots > b_n > 0$ satisfying
$a_0 b_0 + \sum_{i=1}^n (a_i-a_{i-1})b_i  \geq q\,$. Completeness can be derived from
Proposition \ref{prop:let_proof}, using the fact that such collections $\{(\mathsf{P}_{> a_i}, \mathsf{P}_{> b_i})\}_{i = 1}^{n}$ arise in the proof of strong decomposability for the probability modalities in Section~\ref{section:preorder}.

The above discussion concerns compositional proof rules for the \textbf{let} construction of the programming language. Compositional proof rules are available for the other program constructors too. For example, it turns out to be a consequence of Proposition \ref{prop:let_proof} that compositional proof rules are available for the algebraic effect operations. Indeed such rules can be derived from the  \textbf{let} proof rules by recasting the effect operations in terms of 
their \emph{generic effects} in the sense of Plotkin and Power \cite{effect}; for example for a binary operation:
\[
\sigma(M,N) ~~ \equiv ~ ~ \textbf{let } \sigma(\overline{0},\overline{1}) \Rightarrow  z \textbf{ in }  \textbf{case } z \textbf{ of } 
\{Z \Rightarrow M;\, S(x) \Rightarrow N\} \enspace .
\]
Carrying this out in full generality would be excessively technical for this late point in the paper. So instead we illustrate the end result by giving examples of the proof rules one obtains for a selection of the effect operations in our running examples. 

In the case of nondeterminism,  the only effect operation is nondeterministic choice.
For this, one obtains separate rules for the $\Diamond$ and $\Box$ modalities.
\[
\frac{ M \models \Diamond \phi}{\textit{or}(M, N) \models \Diamond \phi} 
\qquad
\frac{ N \models \Diamond \phi}{\textit{or}(M, N) \models \Diamond \phi} 
\qquad
\frac{ M \models \Box \phi  \quad N \models \Box \phi }{\textit{or}(M, N) \models \Box \phi} 
\]
For a second example, in the case of global store, one obtains the following rules for the
$\textit{lookup}_l$ and  $\textit{update}_l$ operations.
\[
\frac{\lambda  x\colon \NatType.\, M \models \{s(l)\} \mapsto (s \rightarrowtail s')\, \phi}{\textit{lookup}_l(\lambda x\colon \NatType.\, M) \models (s \rightarrowtail s')\, \phi}
\qquad
\frac{V \models \{n\} \quad M \models (s[l:=n] \rightarrowtail s')\, \phi}{\textit{update}_l(V; M) \models (s \rightarrowtail s')\, \phi}
\]
In all cases the rules are sound as expected, and also complete in the following sense.  Every valid assertion stating that a program whose outermost constructor is an effect operation satisfies a modal property can be derived by one of the compositional proof rules. 
For example, if it is true that $\textit{update}_l(V; M) \models (s \rightarrowtail s')\, \phi$ then there exists $n$ such that $V \models \{n\}$ and $M \models (s[l:=n] \rightarrowtail s')\, \phi$. (Indeed, let $n$ be such that $V$ is is the  numeral $\overline{n}$.)

There is much that is preliminary in the above account of compositional reasoning principles. As we have already discussed, for reasoning to be practical one needs an expressive finitary program logic, and one needs the compositional principles in the  proof rules to also be expressible using the finitary language. For example, in the case of global store, one would not use low-level modalities of the form $(s \rightarrowtail s')$ where $s,s'$ are individual states, but rather of the form $(P \rightarrowtail Q)$ where $P$ and $Q$ are  formulas in an appropriate logic for expressing general \emph{preconditions} and \emph{postconditions} constraining states. 
Also, one would reformulate the proof rules above to directly deal with such a practical language for preconditions and postconditions.
Nonetheless, the general point remains that such more practical constructions can be compiled down into 
our infinitary propositional logic and its primitive modalities. In the end, in the case of global store, one will end up with a logic looking very much like an extension of Hoare logic to a higher-order language (e.g., there will be some similarities with
Hoare Type Theory~\cite{HoareTT}).

Another issue we have not addressed in the discussion above is that, in general, the verification of interesting programs will involve proving properties of open terms (with free variables) rather than just the closed terms considered above. For this reason, one needs a proof system that deals with sequents of the form
\newcommand{\Colon}{{\,::\,}}
\[
x_1\Colon \phi_1, \, \dots \, , x_n \Colon \phi_n ~ \vdash ~ V \Colon \psi
\quad \text{and} \quad 
x_1\Colon \phi_1, \, \dots\, , \, x_n \Colon \phi_n ~ \vdash ~ M \Colon \Psi \enspace ,
\]
where we write  $V \Colon \phi$ and $M \Colon \Phi$ as syntactic assertions representing the statements
$V \models \phi$ and $M \models \Phi$. The sequents above state that properties $\psi$ and $\Psi$ hold of $V$ and $M$, conditional on properties $\phi_1, \dots, \phi_n$ of the  free variables. 
For example, with such sequents, there is then a natural proof rule for $\lambda$-abstraction:
\[
\frac  
  {x_1\Colon \phi_1, \, \dots \, ,\,  x_n \Colon \phi_n \, , \,  x\Colon \phi ~ \vdash ~ M  \Colon \Psi}
  {x_1\Colon \phi_1, \, \dots \, , \,x_n \Colon \phi_n ~ \vdash ~ \lambda x\colon\!\rho.\, M  \Colon \phi \mapsto \Psi} \enspace .
\]

Summarising the above considerations, while the infinitary modal propositional logic of the present paper is not itself a practical vehicle for specification and verification, it may potentially inform the design of expressive finitary logics suitable for practical specification and reasoning with effects.  In this regard, the compositional rules for the infinitary logic may provide templates for such rules in the practical logics. Moreover, the propositional logic suggests a general sequent-based format for verification, whose mature realisation should be something like an expressive  refinement-type system for algebraic effects, supporting compositional verification. The development of such a system is a goal for future research.

The discussion thus far has concerned reasoning to establish assertions of the form $M \models \Phi$ stating that an individual program satisfies a property of programs. Another important side of reasoning is establishing \emph{equivalences} $M \equiv N$ between programs. Although a main motivation for introducing the behavioural logic was as a means for \emph{defining} equivalence between programs, we do not primarily view it as a tool for establishing or reasoning about equivalences. Indeed, in regard to equivalences, the logic is most easily used as a 
means for establishing \emph{non-equivalences} between programs. In order to show that $M \not\equiv N$, one just needs to exhibit a property $\Phi$ that is satisfied by one of $M,N$ but not the other. 
Examples of such arguments were given in Examples~\ref{example:separation:A} and~\ref{example:separation:B}. Using the logic as a means for proving that two terms are equivalent is trickier since it requires showing that the terms satisfy the same set of logical formulas. Again, one example of such an argument has been given in Example~\ref{example:separation:A}, where two terms were shown to be equivalent with respect to mutual similarity, by showing that they satisfy the same positive formulas. We now give one further example of such an argument, as a means of illustrating a general point.

The general point is that the logical properties of programs, as captured by the modalities and their interaction with the effect operations, suffice to establish the usual equalities between algebraic terms comprising of effect operations, which are central to Plotkin and Power's approach to algebraic effects \cite{effect}.
For the present paper, we illustrate this using one of the equations for global store as an example.
For any value term $\emptyset \vdash \lambda x \colon\! \NatType.\, M  : \NatType \to \tau$, we prove the following equivalence between computation terms of type $\tau$:
\[
  \textit{update}_l(\overline{n}; \textit{lookup}_l(\lambda x \colon\! \NatType.\, M )) ~  \equiv ~  \textit{update}_l(\overline{n}; M\, [\overline{n}/X])\\
\]
By Lemma~\ref{comp_clas}, we need only check that both sides of the equation satisfy the same formulas of the form $(s \rightarrowtail s')\, \phi\,$. To show this, we argue using the proof principles discussed above for the $\textit{lookup}_l$ and  $\textit{update}_l$ operations. 
\begin{align*}
 \textit{update}_l &  (\overline{n}; \textit{lookup}_l(\lambda x \colon\! \NatType.\, M )) ~ \models ~
(s \rightarrowtail s')\, \phi
\\
\text{iff} ~~~  & \textit{lookup}_l(\lambda x \colon\! \NatType.\, M ) ~ \models ~ (s[l:=n] \rightarrowtail s') \, \phi
\\
 ~  \text{iff} ~~~ &   \lambda x \colon\! \NatType.\, M  ~ \models ~ \{n\} \mapsto (s[l:=n] \rightarrowtail s')\, \phi
\\
~  \text{iff} ~~~ &   M[\overline{n}/x]   ~ \models ~ (s[l:=n] \rightarrowtail s')\, \phi
\\
~ \text{iff} ~~~ & \textit{update}_l(\overline{n}; M\, [\overline{n}/X])~ \models ~ (s \rightarrowtail s')\, \phi
\end{align*}

The above derivation and others like it illustrate an alternative to taking equations between algebraic terms as the basic data defining the behaviour of algebraic effects, which is the mainstream approach in the theory of algebraic effects. Instead one can take the modal behavioural properties of effects and their associated proof principles as primitive, and then derive the equations. Indeed the latter is the natural approach if one views the equations as expressing empirical truths about computational behaviour rather than as \emph{a priori} axiomatic facts.

%% file: 10_Conclusion.tex
\section{Discussion and related work}
\label{section:conclusion}

In this paper we have introduced an infinitary propositional logic for expressing properties of programs with algebraic effects, with the principal aim of using the logic to \emph{define} behavioural equivalence for effectful programs. 
Although our logics exhibit certain similarities in form with the endogenous logic developed in
Abramsky's \emph{domain theory in logical form}~\cite{Abramsky91}, our motivation and approach are quite different. 
Whereas Abramsky shows the 
usefulness of an axiomatic approach to a finitary logic as a way of characterising denotational equality, the present paper shows that there is a similar utility in considering an infinitary logic from a semantic perspective (based on operational semantics) as a method for defining behavioural equivalence.

In the literature, logics for effects have primarily been considered for the purpose of reasoning about programs with effects. 
We review the main approaches to such logics.

Pitts' \emph{evaluation logic} was an early logic for general computational effects~\cite{Pitts91}. 
Evaluation logic has built-in $\Box$ and $\Diamond$ modalities, which can be given natural interpretations for some effects, but whose interpretation in the case of other effects is rather forced.
In the light of the general theory of modalities in the present paper, it would be natural to reconsider 
evaluation logic with 
the built-in $\Box$ and $\Diamond$ modalities replaced with 
effect-specific modalities, as in Section~\ref{section:logic}. It is left for future work to assess the extent to which this provides a useful basis for a program logic for effects.

The \emph{logic for algebraic effects}, of Plotkin and Pretnar~\cite{PlotkinPretnar08}, axiomatises effectful behaviour by means of an equational theory over a signature of effect operations, following the algebraic approach to effects advocated by Plotkin and Power~\cite{PlotkinPower02}. 
Although we also have effect operations in the present paper,  we do not take equational axioms as primitive. Instead, it is the choice of the set $\mathcal{O}$ of modalities that determines the equational properties that the effect operations satisfy, as discussed in Section \ref{section:reasoning}.
A different kind of modality is also present in the logic of Plotkin and Pretnar~\cite{PlotkinPretnar08}. 
Their
\emph{operation modalities} are generated by the signature $\Sigma$, with
each modality being associated with a single effect operation.
It is our view that the {behavioural} modalities of the present paper 
provide a more natural 
 vocabulary for
 specifying and reasoning about program behaviour, as discussed in Section \ref{section:reasoning}.
It is possible that a  reformulation of Plotkin and Pretnar's  logic for algebraic effects, using behavioural modalities instead of operation modalities, might provide a more usable framework for specification and verification of programs with algebraic effects.

A different approach to logics for effects has been proposed by 
Goncharov, Mossakowski and Schr{\"o}der~\cite{MossakowskiEA09,GoncharovSchroder13b}.
They assume 
a semantic setting in which the programming language is   rich enough to contain a \emph{pure fragment} that itself acts as a program logic. This approach is  very powerful for certain effects. For example, Hoare logic can be derived in the case of global store. However, it appears not to be as adaptable across as wide a range of effects as the approach of the present paper.

The $F^*$ project \cite{Swamy:2016} takes a complementary approach to reasoning about effects, under which an effectful programming language is embedded within a powerful general purpose theorem proving environment based on dependent type theory. The framework supports a generalisation of Hoare-logic-style reasoning, based on preconditions and postconditions, which is applicable to those effects that can be given associated \emph{Dijkstra monads} \cite{Swamy:2013, Swamy:2016, Dijkstra_Free, Malecha:2011}.
Indeed, an underlying philosophy of the $F^*$ project is to identify, for each eligible effect, the relevant form of precondition, postcondition and composition principle related to them. In contrast,
in the logic of the present paper, the precondition-postcondition style arises naturally for global store, as it
is built into the format of the store modality $(s  \rightarrowtail s')$; whereas other modalities provide different specification formats that nonetheless support compositional  proof methods, as
discussed in Section~\ref{section:reasoning}. In principle, this latter view is more general, as 
there is no \emph{a priori} reason to force specifications into a precondition-postcondition
format. In particular, for effects, such as input/output, in which there is observable mid-execution behaviour, it is not clear that such a format is the most useful (although it can be achieved; see, for example  Penninckx \textit{et al.} \cite{Penninckx:2015} and Maillard \textit{et al.} \cite{Dijkstra_All}).

Examples in Sections~\ref{section:contextual} and~\ref{section:reasoning} show how our logic can be used directly  for reasoning about program equivalence. Its main practical use in this regard is as a method of establishing \emph{non-equivalence}: the non-equivalence of $M$ and $N$ can be shown by exhibiting a single logical property that distinguishes between them. In this regard, there is a useful complementarity underpinning the coincidence of logical equivalence and bisimilarity. For, in contrast, a principal virtue of bisimilarity is that it provides practical reasoning techniques
for establishing equivalences between programs that do hold. (Such techniques include exhibiting bisimilarities and `up-to' methods \cite{Sangiorgi_book, Sangiorgi_book_2}.) 
The relationship with other techniques for reasoning about equivalences between applicative and effectful programs, such as \emph{logical relations} \cite{PittsAM:typor, Pitts00, Benton:2014, op_meta, Katsumata:2011, Hofmann2015}, which are often directed towards contextual equivalence, is a question for future investigation. 

In addition to the question of reasoning about equivalences and non-equivalences, 
the formulation of our logic as a language of formulas expressing behavioural properties,  raises the question 
of its potential usefulness as a vehicle for reasoning about \emph{properties of programs}. 
As discussed in Section~\ref{section:reasoning}, we view the  infinitary propositional logic of this paper as 
providing a low-level language, into which practical high-level finitary logics for expressing program {properties} can potentially be compiled. In order to be sufficiently expressive, such a logic may well incorporate full first-order classical logic (including negation), as well as logical constructs such as the definition of relations as least or greatest fixed points, all of which are translatable into the logic of the present paper. For such an expressive  logic, bisimilarity is likely to be the only program equivalence that supports the interchangeability property 
($M \models \phi$ and $M \equiv N$ imply $N \models \phi$) discussed in Section~\ref{section:introduction}.
We view the development of such high-level  logics and their compositional reasoning principles, aimed at practical specification and verification, as a 
particularly promising topic for future research. 

We end the paper with a brief discussion of some of the choices made in the present paper,  and of their extensions and limitations.
The work in this paper has been carried out for fine-grained call-by-value~\cite{CBV}, which is equivalent to call-by-value.
The definitions can, however,  be adapted to work for call-by-name, and even call-by-push-value~\cite{CBPV}. 
Similarly, they can be adapted to other type constructions such as sums, products, recursive types, and polymorphism, all of which will be addressed in the second authors forthcoming PhD thesis.  A more challenging direction for development, is to generalise the theory to permit
algebraic operations with more general arities than the possibilities allowed in the present paper.
Such a generalisation is needed to include effects such as local and/or higher-order store.
This is nontrivial at several levels. For example, the notion of effect tree requires nontrivial modification. Also,
the relevant notion of bisimilarity (so-called \emph{environmental bisimilarity}  \cite{Sangiorgi:2011, KOUTAVAS2011}) is much more subtle and probably
requires a logic whose formulas  express relations rather than properties.

The central 
notion of modality, in the present paper, was 
adapted from the notion of \emph{observation} in Johann \emph{et al.}~\cite{op_meta}. The crucial change is 
from sets of trees of type $\NatType$
(observations) to a set  of unit-type trees (modalities). This change allows value formulas to be lifted to computation formulas, 
analogously to  \emph{predicate lifting} in coalgebra~(cf., e.g., Jacobs~\cite{Jacobs2016}),
which is a key characteristic of our modalities. 
The properties of \emph{Scott-openness} and \emph{decomposability} play a similar role in 
the present paper to the role they play in Johann \emph{et al.}~\cite{op_meta}. However, the notion of decomposability for modalities (Definition~\ref{definition:decomposable})
is substantially more subtle than the corresponding notion for observations in Johann \emph{et al.}~\cite{op_meta}.

There are certain limitations to the theory of modalities in the present paper. For example, for the combination of probability and nondeterminism, one might naturally consider modalities $\Diamond \mathsf{P}_{>r}$ 
and $\Box \mathsf{P}_{>r}$  asserting the possibility and necessity of the termination probability exceeding $r$.
However, the decomposability property fails. (A related observation appears in Lopez and Simpson~\cite{Lopez}.)
In recent work, the second author has shown that this situation can be rescued by changing to a
 quantitative logic, with a corresponding
notion of quantitative modality \cite{Quantitative}. 

In Hasuo~\cite{HasuoGeneric}, a general framework for quantitative effect-specific modalities is proposed, where 
modalities are given by structure maps $TA \to A$ for algebras for a monad $T$, with $A$ being an object of (potentially quantitative) truth values. From a semantic perspective, the Scott-open modalities of the present paper 
can be  equivalently described as continuous functions from
$T\mathbb{S} \to \mathbb{S}$, where $\mathbb{S}$ is the two element Sierpinski-space domain $\{\bot,\top\}$.
In general, they are not, however, algebras for the tree monad $T$. We leave a thorough comparison of the
approach of the present paper (and its quantitative generalisation)  
with the proposed framework of Hasuo~\cite{HasuoGeneric} for future work.

One topic we have not touched upon in the paper is extending the programming language with 
\emph{handlers}, which are powerful control operators associated with algebraic effects \cite{Handlers2}. If added to the language in unrestricted form, handlers violate the usual equations associated with algebraic effects. Nevertheless, it is interesting to consider 
if our modalities can be used, on the one hand, to reason about handlers, and, on the other, to help constrain handlers to `safe' usages that respect the expected program equivalences. 
For example, various operators from concurrency can be defined using handlers \cite{Abadi_Plotkin, Glabbeek:2010}, and it would be interesting to see if this allows the usual Hennessy-Milner logic to
be subsumed in our setting.

\paragraph{Acknowledgements.~}
We thank Francesco Gavazzo, Aliaume Lopez and the anonymous conference and journal reviewers for helpful discussions and comments.

%% file: A_Proofs.tex
\appendix
\section{Howe's method proof}\label{section:proofs}

In this appendix, we give the details that were left out from the Howe's method proof in Section \ref{section:Howe}. 
The results in this appendix are indexed according to whether they were previously stated in Section \ref{section:Howe} or are only stated here.
Remember the definitions of the Howe closure given in Definition \ref{def:How}.
We first look at some preliminary results, mostly from Lassen \cite{Lassen}.

\LemmaHowe*

\begin{proof} This proof is adapted from \cite[Lemma 3.8.2]{Lassen}. We prove the properties separately.
	\begin{enumerate}
		\item Since $\mathcal{R}$ is reflexive, so is $\mathcal{R}^{\circ}$. Hence $\widehat{\mathcal{R}}^{\bullet} = \widehat{\mathcal{R}}^{\bullet} \textit{id} \subseteq \widehat{\mathcal{R}}^{\bullet} \mathcal{R}^{\circ} = \mathcal{R}^{\bullet}$.
		
		\item Note that the compatible refinement of a reflexive relation is reflexive (since the compatible refinement rules line up with the typing rules). Hence $\widehat{\mathcal{R}^{\bullet}}$ is reflexive, since $\mathcal{R}^{\bullet}$ is. So $\mathcal{R}^{\circ} = \textit{id} ~ \mathcal{R}^{\circ} \subseteq \widehat{\mathcal{R}^{\bullet}} \mathcal{R}^{\circ} = \mathcal{R}^{\bullet}$.
		
		\item This requires an induction on the shape of $A$ (which may be a value or a computation). If $\{x\} \vdash A \,\mathcal{R}^{\bullet}\, B$ then by \textbf{HC} or \textbf{HV} we have $\{x\} \vdash A \,\widehat{\mathcal{R}^{\bullet}}\, C$ and $\{x\} \vdash C \,\mathcal{R}^{\circ}\, B$. So we know $C[W] \,\mathcal{R}^{\circ}\, B[W]$. We need to prove, $A[V] \,\widehat{\mathcal{R}^{\bullet}}\, C[W]$. In each of the cases of $A$, $\{x\} \vdash A \,\widehat{\mathcal{R}^{\bullet}}\, C$ is derived from rule \textbf{Cn} for some number $n$. This rule has as its premise some sequence of relations $A_i \,\mathcal{R}^{\bullet}\, C_i$. By induction hypothesis we have $A_i[V] \,\mathcal{R}^{\bullet}\, C_i[W]$, this is also trivially true in the base cases $n \in \{1,2\}$ since then the sequence is empty. Using \textbf{Cn} we can then derive that $A[V] \,\widehat{\mathcal{R}^{\bullet}}\, C[W]$. Hence by \textbf{HV} or \textbf{HC} we get $A[V] \,\mathcal{R}^{\bullet}\, B[W]$. One can verify that this argument works for each of the cases of \textbf{Cn}.
	\end{enumerate}
\end{proof}

We can also say something about composing the Howe closure with the original relation, which by Lemma \ref{Rel_prop1} works well with our relator.
\begin{lemma}\label{Hprop2}
	If \,$\mathcal{R}$ is a preorder relation on closed terms, then we have:
	\begin{enumerate}
		\item If $A \,\mathcal{R}^{\bullet}\, B$ and $B \,\mathcal{R}^{\circ}\, C$, then $A \,\mathcal{R}^{\bullet}\, C$.
		\item For closed terms $M,N,K$ such that $|M| \,\mathcal{O}(\mathcal{R}^{\bullet})\, |N|$ and $|N| \,\mathcal{O}(\mathcal{R})\, |K|$ we have $|M| \,\mathcal{O}(\mathcal{R}^{\bullet})\, |K|$.
	\end{enumerate}
\end{lemma}
\begin{proof} We proof the properties individually. 
	\begin{enumerate}
		\item We use that $\mathcal{R}$ is transitive, hence $\mathcal{R}^{\circ}$ is transitive meaning $\mathcal{R}^{\circ} \mathcal{R}^{\circ} \subseteq \mathcal{R}^{\circ}$. Hence with $\mathcal{R}^{\bullet} = \widehat{\mathcal{R}}^{\bullet} \mathcal{R}^{\circ}$ we have $\mathcal{R}^{\bullet} \mathcal{R}^{\circ} = (\widehat{\mathcal{R}}^{\bullet} \mathcal{R}^{\circ}) \mathcal{R}^{\circ} \subseteq \widehat{\mathcal{R}}^{\bullet} \mathcal{R}^{\circ} = \mathcal{R}^{\bullet}$.
		
		\item This follows from applying Lemma \ref{Rel_prop1} to the previous statement.
	\end{enumerate}
\end{proof}

Lastly, we prove the result needed to deal with the $\mathcal{O}$-bisimilarity.

\LemmaSym*

\begin{proof}
	This proof is taken from \cite[Lemma 3.8.2(4)]{Lassen}.
	Looking at the compatible refinement rules, it is not difficult to see $\widehat{\mathcal{S}^{op}} = \widehat{\mathcal{S}}^{op}$ for any relation $\mathcal{S}$. From Lemma \ref{Hprop1} we know that $\mathcal{R}^{\circ} \subseteq \mathcal{R}^{\bullet}$, and $\mathcal{R}^{\bullet}$ is compatible hence $\widehat{\mathcal{R}^{\bullet *}} \subseteq \mathcal{R}^{\bullet *}$. So we get that:
	$$\widehat{\mathcal{R}^{\bullet * op}}\mathcal{R}^{\circ} = \widehat{\mathcal{R}^{\bullet * op}}\mathcal{R}^{\circ op} = \widehat{\mathcal{R}^{\bullet *}}^{op}\mathcal{R}^{\circ op} \subseteq$$
	$$\mathcal{R}^{\bullet * op}\mathcal{R}^{\circ op} \subseteq \mathcal{R}^{\bullet * op}\mathcal{R}^{\bullet op} = \mathcal{R}^{\bullet op *}\mathcal{R}^{\bullet op} =  \mathcal{R}^{\bullet op *} =  \mathcal{R}^{\bullet * op}$$
	Hence $\mathcal{R}^{\bullet * op}$ is a solution for $\mathcal{S}$ to the inclusion $\widehat{\mathcal{S}}\mathcal{R}^{\circ} \subseteq \mathcal{S}$. Since $\mathcal{R}^{\bullet}$ is the least solution, we have $\mathcal{R}^{\bullet} \subseteq \mathcal{R}^{\bullet * op}$. So if $A \,\mathcal{R}^{\bullet *}\, B$, then $A = C_0 \,\mathcal{R}^{\bullet}\, C_1  \,\mathcal{R}^{\bullet} \dots  \mathcal{R}^{\bullet}\, C_{n-1} = B$ for some choice of sequence $\{C_i\}_{i \in I}$, so we can derive that $A=C_0 \,\mathcal{R}^{\bullet * op}\, C_1  \,\mathcal{R}^{\bullet * op}\, \dots  \,\mathcal{R}^{\bullet * op}\, C_{n-1} = B$ meaning $A \,\mathcal{R}^{\bullet * op}\, B$. We conclude that $B \,\mathcal{R}^{\bullet *}\, A$, so $\mathcal{R}^{\bullet *}$ is symmetric.
\end{proof}

Now we specifically take $\mathcal{R}$ to be a pre-order $\mathcal{O}$-simulation $\sqsubseteq$. We assume $\mathcal{O}$ to be a decomposable set of Scott open modalities. The lemmas stated before are satisfied, hence we know that $(\sqsubseteq^{\circ}) \subseteq (\sqsubseteq^{\bullet})$ by Lemma \ref{Hprop1}. We prove that $\sqsubseteq^{\bullet}$ is an $\mathcal{O}$-simulation, starting with condition 1 of Definition \ref{definition:sim}:

\begin{lemma}\label{Nequ}
	If for $V, W:\NatType$ we have $V \sqsubseteq^{\bullet} W$, then $V = W$.
\end{lemma}

\begin{proof}
	For a value term $V$ of type $\NatType$, we have that $V = \overline{n}$ for some natural number $n$. Assume $\overline{n} ~ \sqsubseteq^{\bullet} W$, then for some $L$ we have $\overline{n} ~ \widehat{\sqsubseteq^{\bullet}} L$ and $L \sqsubseteq W$. The latter means $L = W$ by property 1 of $\mathcal{O}$-simulations, so $\overline{n} ~ \sqsubseteq^{\bullet} W \Rightarrow \overline{n} ~ \widehat{\sqsubseteq^{\bullet}} W$. If $n=0$ and hence $\overline{n} = Z$, then $\overline{n} ~ \widehat{\sqsubseteq^{\bullet}} W$ must have been derived from \textbf{C2}, so $W = Z = \overline{n}$. As an induction step, we assume $\overline{n} ~ \sqsubseteq^{\bullet} W \Rightarrow \overline{n} = W$ and $\overline{n+1} ~ \sqsubseteq^{\bullet} W'$. We derive $S(\overline{n}) ~ \widehat{\sqsubseteq^{\bullet}} W'$ which must be from \textbf{C3} hence $W'= S(K)$ where $\overline{n} ~ \sqsubseteq^{\bullet} K$. So by assumption $\overline{n} = K$ which means $\overline{n+1} = W'$. We can conclude by induction that $V = \overline{n} ~ \sqsubseteq^{\bullet} W \Rightarrow V = W$.
\end{proof}

So condition 1 of Definition \ref{definition:sim} holds. Condition 3 is satisfied by compatibility of $\sqsubseteq^{\bullet}$ (shown in Lemma \ref{Hprop1}), specifically by rule \textbf{C6} of compatibility. Condition 2 is the most difficult to prove and requires an induction on the reduction relation of terms. In the order of trees, $|M|_n \leq |M|_{n+1}$ and $|M| = \bigsqcup_n |M|_n$. So by the properties in Lemma \ref{Rel_prop2} we have:
$$|M| \,\mathcal{O}(\sqsubseteq)\, |N| \Leftrightarrow \forall n, |M|_n \,\mathcal{O}(\sqsubseteq)\, |N|$$

Now for the important result, following the techniques used in \cite{Relational}.

\KeyLemma*

\begin{proof}
	We do an induction on $n$. 
	
	\textbf{Base case}, where $n = 0$, which means $|A|_n = \bot$. So we have $|A|_0 = \bot \leq |B|$. By Lemma \ref{Hprop1} we know $\sqsubseteq^{\bullet}$ is reflexive, so using Lemma \ref{Rel_prop2} we get $|A|_0 \,\mathcal{O}(\sqsubseteq^{\bullet})\, |B|$.
	
	\textbf{Induction step} $(n+1)$. We assume as the \emph{Induction Hypothesis} that for any two terms $A'$ and $B'$, if $A' \sqsubseteq^{\bullet} B'$ and $k \leq n$ we have $|A'|_k \,\mathcal{O}(\sqsubseteq^{\bullet})\, |B'|$. We want to prove $|A|_{n+1} \,\mathcal{O}(\sqsubseteq^{\bullet})\, |B|$. 
	
	\noindent
	By \textbf{HC} in Definition \ref{def:How}, there is a $C$ such that $A \,\widehat{\sqsubseteq^{\bullet}}\, C$ and $C \sqsubseteq B$. Since $\sqsubseteq$ is an $\mathcal{O}$-simulation, we have $|C| \,\mathcal{O}(\sqsubseteq)\, |B|$. So we aim to prove that $|A|_{n+1} \,\mathcal{O}(\sqsubseteq^{\bullet})\, |C|$ such that by Lemma \ref{Hprop2} we get $|A|_{n+1} \,\mathcal{O}(\sqsubseteq^{\bullet})\, |B|$. To prove this, we do a case distinction on the shape of $A$ to derive: $\forall A, A \,\widehat{\sqsubseteq^{\bullet}}\, C \Rightarrow |A|_{n+1} \,\mathcal{O}(\sqsubseteq^{\bullet})\, |C|$.
	
	\begin{enumerate}
		\item $A = \textbf{return}(V) : \TypeOne$, we have $|A|_{n+1} = \eta(V)$ and $\textbf{return}(V) \,\widehat{\sqsubseteq^{\bullet}}\, C$. This is only possible from rule \textbf{C4}, meaning $C= \textbf{return}(W)$ and $V \sqsubseteq^{\bullet} W$. By Lemma \ref{Rel_prop3} we have $\eta(V) \,\mathcal{O}(\sqsubseteq^{\bullet})\, \eta(W)$, hence $|A|_{n+1} = |\textbf{return}(V)|_{n+1} = \eta(V) \,\mathcal{O}(\sqsubseteq^{\bullet})\, \eta(W) = |\textbf{return}(W)| = |C| $. That concludes this case.
		
		\item $A = (\lambda x : \TypeOne.M)\,V : \TypeTwo$, hence $|A|_{n+1} = |M[V/x]|_n$. We have $(\lambda x : \TypeOne.M)\,V \,\widehat{\sqsubseteq^{\bullet}}\, C$ which can only be concluded from \textbf{C6}. Hence $C = WU$ for some value terms $W$ and $U$, with $(\lambda x : \TypeOne.M) \sqsubseteq^{\bullet} W$ and $V \sqsubseteq^{\bullet} U$. Now $(\lambda x : \TypeOne.M) \sqsubseteq^{\bullet} W$ can only be from a combination of rule \textbf{HV} with \textbf{C5}, hence we have $\{x\} \vdash M \sqsubseteq^{\bullet} L$ and $\lambda x : \TypeOne.L \sqsubseteq W$. From $\{x\} \vdash M \sqsubseteq^{\bullet} L$, $V \sqsubseteq^{\bullet} U$ and Lemma \ref{Hprop1} we get $M[V/x] \sqsubseteq^{\bullet} L[U/x]$, hence by the induction hypothesis, $|M[V/x]|_n  \,\mathcal{O}(\sqsubseteq^{\bullet})\, |L[U/x]|$. Since $\lambda x : \TypeOne.L \sqsubseteq W$, we have $(\lambda x : \TypeOne.L)\,U \sqsubseteq WU = C$, hence with $|(\lambda x : \TypeOne.L)\,U| = |L[U/x]|$ it holds that $|L[U/x]|   \,\mathcal{O}(\sqsubseteq)\, |WU| = |C|$, so by Lemma \ref{Hprop2}, $|(\lambda x : \TypeOne.M)\,V]|_{n+1} = |M[V/x]|_n \,\mathcal{O}(\sqsubseteq^{\bullet})\, |C|$.
		
		\item $A = \textbf{fix}(F) : \TypeTwo \rightarrow \TypeOne$, we have $\textbf{fix}(F) \,\widehat{\sqsubseteq^{\bullet}}\, C$ which can only be concluded from \textbf{C7}, hence $C = \textbf{fix}(F')$ and $F \sqsubseteq^{\bullet} F'$. Take $f:  (\TypeTwo \rightarrow \TypeOne) \rightarrow ( \TypeTwo \rightarrow \TypeOne) \vdash N : \TypeTwo \rightarrow \TypeOne$ to be the following term:
		$$
		N := \textbf{return } \lambda x : \TypeTwo. \textbf{let } f(\lambda y : \TypeTwo.\textbf{let fix } f \Rightarrow z \textbf{ in } zy) \Rightarrow w \textbf{ in } wx
		$$
		This value term has been specifically chosen such that $\textbf{fix}(G)$ reduces in one step to $N[G/f]$ for any term $G$, so $|\textbf{fix}(G)|_{n+1} = |N[G/f]|_n$. We have $N \sqsubseteq^{\bullet} N$ by reflexivity, hence with $F \sqsubseteq^{\bullet} F'$ and Lemma \ref{Hprop1} we get $N[F]  \sqsubseteq^{\bullet} N[F'/f]$.  Hence $|\textbf{fix}(F)|_{n+1} = |N[F/f]|_n \,{\mathcal{O}(\sqsubseteq^{\bullet})}\, |N[F'/f]| = |\textbf{fix}(F')| = |C|$ by the induction hypothesis.
		
		\item $A = \textbf{case } V \textbf{ of } \{Z \Rightarrow M ; S(x) \Rightarrow N\} : \TypeOne$, then $A \,\widehat{\sqsubseteq^{\bullet}}\, C$ can only be concluded from \textbf{C8}, hence $C = \textbf{case } V' \textbf{ of } \{Z \Rightarrow M' ; S(x) \Rightarrow N'\}$ for some terms $V'$, $M'$, and $N'$ where $V \sqsubseteq^{\bullet} V'$, $M \sqsubseteq^{\bullet} M'$, and $\{x\} \vdash N \sqsubseteq^{\bullet} N'$. Since $V \sqsubseteq^{\bullet} V'$ we have $V = V'$ by Lemma \ref{Nequ}. We do a case distinction on $V$. 
		\begin{enumerate}
			\item If $V = Z = V'$, then $|A|_{n+1} = |M|_n$ by the reduction relation. Using the induction hypothesis on $M \sqsubseteq^{\bullet} M'$ we have $|M|_n \,\mathcal{O}(\sqsubseteq^{\bullet})\, |M'| = |\textbf{case } Z \textbf{ of } \{Z \Rightarrow M' ; S(x) \Rightarrow N'\}| = |C|$, so we are finished. 
			
			\item If $V = S(W) = V'$, then $|A|_{n+1} = |N[W/x]|_n$. By  $\{x\} \vdash N \sqsubseteq^{\bullet} N'$ we have $N[W/x] \sqsubseteq^{\bullet} N'[W/x]$ (Lemma \ref{Hprop1}), hence by the induction hypothesis, 
			
			$|A|_{n+1} = |N[W/x]|_n \,\mathcal{O}(\sqsubseteq^{\bullet})\, |N'[W/x]| = |C|$.
		\end{enumerate}
		\item $A = \textbf{let } M \Rightarrow x \textbf{ in } N : \TypeTwo$. Looking at the $\mathbf{let}$-term, one can observe that we have the following upper bound
		$$|\textbf{let } M \Rightarrow x \textbf{ in } N|_{n+1} \leq |M|_n[L \mapsto |N[L/x]|_n].$$
		The only derivation of $A \,\widehat{\sqsubseteq^{\bullet}}\, C$ is by \textbf{C9}, from which we know that $C = \textbf{let } M' \Rightarrow x \textbf{ in } N'$ for some terms $M'$ and $N'$, where $M \sqsubseteq^{\bullet} M'$ and $\{x\} \vdash N \sqsubseteq^{\bullet} N'$. Using the above observation we know that $|A|_{n+1} \leq |M|_n[L \mapsto |N[L/x]|_n]$ and $|C| = |M'|[L \mapsto |N'[L/x]|]$. Using property 1 of Corollary \ref{Rel_prop4}, we want to prove that $|M|_n[L \mapsto |N[L/x]|_n] \; \mathcal{O}(\sqsubseteq^{\bullet}) \; |M'|[L \mapsto |N'[L/x]|]$.
		
		Note that for all $L \sqsubseteq^{\bullet} L'$ it holds by $\{x\} \vdash N \sqsubseteq^{\bullet} N'$ and Lemma \ref{Hprop1} that $N[L/x] \sqsubseteq^{\bullet} N'[L'/x]$. Hence by the induction hypothesis: $|N[L/x]|_n \; \mathcal{O}(\sqsubseteq^{\bullet})\; |N'[L'/x]|$. Define function $f, g: \mathit{Val}(\TypeOne) \to T\mathit{Com}(\TypeTwo)$ where $f(L) := |N[L/x]|_n$ and $g(L') := |N'[L'/x]|$. We also have by the induction hypothesis on $M \sqsubseteq^{\bullet} M'$ that $|M|_n \,\mathcal{O}(\sqsubseteq^{\bullet})\, |M'|$. Take $t = |M|_n$ and $r = |M'|$. So we can use property 1 of Corollary \ref{Rel_prop4} on $f$, $g$, $t$, and $r$ to conclude that $|A|_{n+1} \leq$
		$$ |M|_n[L \mapsto |N[L/x]|_n] = t[L \mapsto f(L)] \; \mathcal{O}(\sqsubseteq^{\bullet}) \; r[L' \mapsto g(L')] = |M'|[L \mapsto |N'[L/x]|] = |C| \enspace .$$ 
		Using a combination of property 1 of Lemma \ref{Rel_prop2} to prove $|A|_{n+1} \,\mathcal{O}(\mathit{id})\, t[L \mapsto f(L)]$, and property 2 of Lemma \ref{Rel_prop1} we can derive $|A|_{n+1} \,\mathcal{O}(\sqsubseteq^{\bullet})\, |C|$.
		
		\item $A = \sigma(M_0,M_1,\dots,M_{m-1}) : \TypeOne$ where $\sigma: \alpha^{m} \rightarrow \alpha$. The statement $A \,\widehat{\sqsubseteq^{\bullet}}\, C$ can only follow from \textbf{CA}, hence $C = \sigma(M_0',M_1',\dots,M'_{m-1})$ where for all $i$, $M_i \sqsubseteq^{\bullet} M_i'$. By the induction hypothesis, $|M_i|_n \,\mathcal{O}(\sqsubseteq^{\bullet})\, |M_i'|$, hence $|A|_{n+1} = \sigma \{i \mapsto |M_i|_n\} \,\mathcal{O}(\sqsubseteq^{\bullet})\, \sigma \{i \mapsto |M_i'|\} = |C|$ by property 2 of Corollary \ref{Rel_prop4}.
		
		\item $A = \sigma(V) : \TypeOne$ where $\sigma: \alpha^{\NatType} \rightarrow \alpha$. The statement $A \,\widehat{\sqsubseteq^{\bullet}}\, C$ follows only from \textbf{CB}, hence $C = \sigma(V')$ where $V \sqsubseteq^{\bullet} V'$ meaning for any $\overline{m} : \NatType$ we have $V \overline{m} \sqsubseteq^{\bullet} V' \overline{m}$. By the induction hypothesis, $|V \overline{m}|_n \,\mathcal{O}(\sqsubseteq^{\bullet})\, |V' \overline{m}|$, hence $|A|_{n+1} = \sigma \{m \mapsto |V \overline{m}|_n\} \,\mathcal{O}(\sqsubseteq^{\bullet})\, \sigma \{m \mapsto |V' \overline{m}|\} = |C|$ by property 2 of Corollary \ref{Rel_prop4}.
		
		\item $A = \sigma(V,t) : \TypeOne$ where either $t = (M_0,\dots,M_{n-1})$ or $t = W$. Very similar to the previous cases, from either \textbf{CC} or \textbf{CD} we have $C = \sigma(V',t')$ where $t \sqsubseteq^{\bullet} t'$ and $V \sqsubseteq^{\bullet} V'$ hence $V = V' : \NatType$ by Lemma \ref{Nequ}. The rest follows by repeating the proof in the previous cases depending on what shape $t$ is.
	\end{enumerate}
\end{proof}

We can conclude that $M \sqsubseteq^{\bullet} N \Rightarrow |M|_n \,\mathcal{O}(\sqsubseteq^{\bullet})\, |N|$ for any $n$. Hence with $|M| = \sqcup_n |M|_n$ and Lemma \ref{Rel_prop2} we can finally conclude $|M| \,\mathcal{O}(\sqsubseteq^{\bullet})\, |N|$. Hence the relation $\sqsubseteq^{\bullet}$ satisfies all conditions of Definition \ref{definition:sim}, resulting in the following proposition.

\MainProp*

\begin{proof}
	By Lemma \ref{Nequ}, $\sqsubseteq^{\bullet}$ satisfies condition (1) of being an $\mathcal{O}$-simulation. 
	By Lemma \ref{lemma_key} and $|M| = \sqcup_n |M|_n$ for any term $M$, it satisfies condition (2). 
	By compatibility it satisfies condition (3).
	We conclude that $\sqsubseteq^{\bullet}$ is an $\mathcal{O}$-simulation.
\end{proof}

The Howe's method proof is finalised in Section \ref{section:Howe} with Theorem \ref{theorem:sim_com} A and B and their proofs.